\documentclass[lettersize,journal]{IEEEtran}
\usepackage{cite}
\usepackage{amsmath,amssymb,amsfonts,amsthm}
\usepackage{algorithmic}
\usepackage{graphicx}
\usepackage{textcomp}
\usepackage{xcolor}
\usepackage{amsmath,amsfonts}
\usepackage{algorithmic}
\usepackage{algorithm}
\usepackage{array}
\usepackage[caption=false,font=normalsize,labelfont=sf,textfont=sf]{subfig}
\usepackage{textcomp}
\usepackage{stfloats}
\usepackage{url}
\usepackage{verbatim}
\usepackage[justification=centering]{caption}
\usepackage{bm} 
\usepackage{xcolor}
\usepackage{booktabs}
\usepackage{enumitem}
\usepackage{hyperref}
\usepackage{siunitx}
\hypersetup{colorlinks=false}
\usepackage[capitalize,noabbrev]{cleveref}
\usepackage{tikz, pgfplots}
\usetikzlibrary{plotmarks,spy,angles,quotes}
\pgfplotsset{compat=newest}
\usepackage{caption}
\usepackage{cuted}
\usepackage{nccmath}
\usetikzlibrary{calc}

\DeclareMathOperator{\tr}{tr}
\DeclareMathOperator{\diag}{diag}

\DeclareMathOperator{\supp}{supp}

\DeclareMathOperator{\sign}{sign}
\DeclareMathOperator{\sinc}{sinc}
\newcommand{\T}{{\sf T}}
\renewcommand{\H}{{\sf H}}

\newcommand{\RR}{{\mathbb{R}}}
\newcommand{\CC}{{\mathbb{C}}}
\newcommand{\EE}{{\mathbb{E}}}
\newcommand{\NN}{{\mathcal{N}}}
\newcommand{\A}{{\mathbf{A}}}
\newcommand{\B}{{\mathbf{B}}}
\newcommand{\C}{{\mathbf{C}}}
\newcommand{\J}{{\mathbf{J}}}
\newcommand{\M}{{\mathbf{M}}}
\renewcommand{\S}{{\mathbf{S}}}
\renewcommand{\P}{{\mathbf{P}}}
\newcommand{\U}{{\mathbf{U}}}
\newcommand{\Q}{{\mathbf{Q}}}
\newcommand{\X}{{\mathbf{X}}}
\newcommand{\Z}{{\mathbf{Z}}}
\newcommand{\bPhi}{{\boldsymbol{\Phi}}}
\newcommand{\ee}{\mathbf{e}}
\newcommand{\uu}{\mathbf{u}}
\newcommand{\vv}{\mathbf{v}}

\newcommand{\s}{\mathbf{s}}
\newcommand{\x}{\mathbf{x}}

\newcommand{\z}{\mathbf{z}}
\newcommand{\zo}{\mathbf{0}}
\newcommand{\I}{\mathbf{I}}
\newcommand{\V}{\mathbf{V}}

\newcommand{\bLambda}{\boldsymbol{\Lambda}}
\newcommand{\bTheta}{\boldsymbol{\Theta}}

\definecolor{RED}{rgb}{0.7,0,0}
\definecolor{BLUE}{rgb}{0,0,0.69}
\definecolor{GREEN}{rgb}{0,0.6,0}
\definecolor{PURPLE}{rgb}{0.69,0,0.8}

\newcommand{\RED}{\color[rgb]{0.70,0,0}}
\newcommand{\BLUE}{\color[rgb]{0,0,0.69}}

\newtheorem{Definition}{Definition}
\newtheorem{Assumption}{Assumption}
\newtheorem{Theorem}{Theorem}
\newtheorem{Corollary}{Corollary}
\newtheorem{Proposition}{Proposition}
\newtheorem{Lemma}{Lemma}
\newtheorem{Remark}{Remark}
\newtheorem{Example}{Example}

\let\texdisplaystyle\displaystyle
\def\displaytotextstyle{\textstyle\let\displaystyle\texdisplaystyle}

\newenvironment{talign*}
 {\let\displaystyle\displaytotextstyle\csname align*\endcsname}
 {\endalign}

\definecolor{blue1}{RGB}{0,114,189}
\definecolor{orange1}{RGB}{217,83,25}
\definecolor{yellow1}{RGB}{237,177,32}
\definecolor{purple1}{RGB}{126,47,142}
\definecolor{green1}{RGB}{119,172,48}
\definecolor{red1}{RGB}{162,20,47}

\def\BibTeX{{\rm B\kern-.05em{\sc i\kern-.025em b}\kern-.08em
    T\kern-.1667em\lower.7ex\hbox{E}\kern-.125emX}}
\usepackage{balance}
\begin{document}
\title{A Large-Dimensional Analysis of ESPRIT DoA Estimation: Inconsistency and a Correction via RMT}

\author{Zhengyu~Wang,~\IEEEmembership{Student Member,~IEEE}, 
Wei~Yang,
Xiaoyi~Mai, 
Zenan~Ling,~\IEEEmembership{Member,~IEEE},\\
Zhenyu~Liao,~\IEEEmembership{Senior Member,~IEEE},
Robert C.~Qiu,~\IEEEmembership{Fellow,~IEEE}
\thanks{Z.~Wang, W.~Yang, Z.~Ling, Z~Liao, and R.~C.~Qiu are with the School of Electronic Information and Communications (EIC), Huazhong University of Science and Technology (HUST), Wuhan 430074, China. 
X.~Mai is with the Institut de Mathématiques de Toulouse (IMT), University of Toulouse-Jean Jaurès, Toulouse, France.
\textit{Corresponding author: Zhenyu Liao} (email: \href{mailto:zhenyu_liao@hust.edu.cn}{zhenyu\_liao@hust.edu.cn}).}%
\thanks{Part of this work was presented at the IEEE 32nd European Signal Processing Conference (EUSIPCO 2024), Lyon, France.}%
\thanks{Z.~Liao is supported by the National Natural Science Foundation of China (via NSFC-12571561 and NSFC-62206101) and the Fundamental Research Support Program of HUST (2025BRSXB0004).
R.\@ C.\@ Qiu is supported by the National Natural Science Foundation of China (via fund NSFC-12141107) and the Interdisciplinary Research Program of HUST (2023JCYJ012).
Z.~Ling is supported by the National Natural Science Foundation of China (via NSFC-62406119) and the Natural Science Foundation of Hubei Province (2024AFB074).
}
}

\maketitle

\begin{abstract}
In this paper, we perform asymptotic analyses of the widely used ESPRIT direction-of-arrival (DoA) estimator for large arrays, where the array size $N$ and the number of snapshots $T$ grow to infinity at the same pace.
In this large-dimensional regime, the sample covariance matrix (SCM) is known to be a poor eigenspectral estimator of the population covariance.
We show that the classical ESPRIT algorithm, that relies on the SCM, and as a consequence of the large-dimensional inconsistency of the SCM, produces \emph{inconsistent} DoA estimates as $N,T \to \infty$ with $N/T \to c \in (0,\infty)$, for both widely-~and~closely-spaced DoAs.
Leveraging tools from random matrix theory (RMT), we propose an improved G-ESPRIT method and prove its \emph{consistency} in the same large-dimensional setting. 
From a technical perspective, we derive a novel bound on the eigenvalue differences between two potentially non-Hermitian matrices, which may be of independent interest.
Numerical simulations are provided to corroborate our theoretical findings.
\end{abstract}


\section{Introduction}
\label{sec:intro}

Direction-of-arrival (DoA) estimation is a fundamental task in array signal processing with various applications, ranging from radar, sonar, and wireless communications, to medical imaging. 
Accurate DoA estimation enables systems to localize sources, optimize resource allocation, and enhance signal quality in complex environments. 
Among the numerous methods developed for DoA estimation, subspace-based methods such as MUSIC and its variants, as well as the ESPRIT approach, have gained significant popularity~\cite{schmidt1986multiple,rao1989performance,1457851}. 
Their effectiveness and computational efficiency (compared to, e.g., maximum likelihood estimators) have made them mainstays in modern array signal processing.

Subspace-based methods such as MUSIC and ESPRIT exploit the inherent orthogonality between the signal and noise subspace of the \emph{population covariance} to estimate DoAs from noisy observations. 
Since the population covariance is practically inaccessible, eigenspectral decomposition is performed on the sample covariance matrix (SCM) of the received signals to extract DoA information.
When the number of snapshots $T$ is \emph{much larger} than the array size $N$, the SCM provides an accurate estimate of the population covariance, and the statistical performance of, e.g., MUSIC~\cite{stoica1989music}, total least squares ESPRIT~\cite{ottersten1991performance}, and least squares ESPRIT~\cite{viberg1991sensor} has been well-studied in prior work.

However, these subspace-based methods are known to suffer from the so-called ``threshold effect,'' where their performance degrades drastically when either the signal-to-noise ratio (SNR) or the number of snapshots falls below a certain threshold~\cite{stoica1995resolution}. 
This limitation has sparked renewed interest in analyzing these methods in the regime of \emph{large arrays and limited snapshots}, where both $N$ and $T$ grow to infinity at the same pace, i.e.,
\begin{align*}
    N,T \to \infty, \quad N/T\to c \in (0,\infty),
\end{align*}
a setting that models scenarios where large sensing arrays acquire data within relatively short sampling durations. 

With the progress of random matrix theory (RMT) over the past decade, many methods in statistics, signal processing, and machine learning have been revisited in the large-dimensional regime, leading to novel insights and improved algorithms tailored for large-dimensional data~\cite{paul2014Random,couillet2022RMT4ML}.
A key takeaway from these developments is that when $T$ is \emph{not much larger} than $N$, the SCM becomes a poor eigenspectral estimator of the population covariance (see~\cite{couillet2022RMT4ML} and \Cref{subsec:review_rmt} for a brief review). 
In such cases, the sample eigenvectors/eigenspaces are \emph{biased} from their population counterparts.
Consequently, subspace methods \emph{cannot} be expected to provide consistent estimates of the true DoAs in scenarios where $N,T$ are both large and comparable.
Notably, the empirically observed ``threshold effect'' of these methods is a direct manifestation of the \emph{phase transition} behavior inherent in the large-dimensional SCM eigenspace.

In the case of MUSIC, it has been shown in~\cite{Vallet2015Performance} that despite the eigenspectral \emph{inconsistency} of the SCM in the large $N,T$ regime, MUSIC still provides consistent DoA estimates in widely-spaced DoA scenarios (see \Cref{ass:widely_spaced} for a precise definition), when above the phase transition threshold.
However, for closely-spaced sources (see \Cref{ass:closely_spaced}), where the separation between angles is of the order $O(N^{-1})$, the classical MUSIC approach \emph{fails to deliver accurate estimates}.
To address this limitation, a modified MUSIC algorithm, G-MUSIC, was introduced by ``correcting'' the sample signal subspace using RMT.
G-MUSIC is guaranteed to provide $N$-consistent\footnote{An estimator $\hat \theta$ is called $N$-consistent for $\theta$ if $N(\hat \theta-\theta) \to 0$ almost surely as $N,T \to \infty$, see \Cref{sec:inconsist} for a detailed discussion.} DoA estimates in the large $N,T$ regime~\cite{mestre2008modified,vallet2012improved,hachem2012large,Vallet2015Performance}, making it particularly valuable in closely-spaced DoA scenarios.
Beyond these first-order analyses of MUSIC and G-MUSIC, further research has explored the asymptotic properties of their Mean Square Errors~(MSEs), by establishing the second-order behavior of these estimators through Central Limit Theorem~(CLT), see for example~\cite{Vallet2012ACF,Mestre2011}.

\medskip
In this paper, we analyze the equally popular subspace-based DoA method ESPRIT~\cite{paulraj1986subspace} (reviewed in \Cref{subsec:models}) in this \emph{large array and limited snapshot} regime, where ESPRIT is also observed to suffer from the ``threshold effect.'' 
While ESPRIT, like MUSIC, also relies on the SCM, it exploits the rotational invariance property of the signal subspaces between different subarrays using a (more intricate) least squares approach.
This distinction makes prior analyses, such as those in \cite{mestre2008modified,vallet2012improved,hachem2012large,Vallet2015Performance} for MUSIC through eigenvector projections, not directly applicable.
Although it has been empirically observed that ESPRIT outperforms MUSIC in certain cases~\cite{roy1987Comparative} but not in others~\cite{lavate2010Performance}, its theoretical analysis remains an \emph{open problem}~\cite{vallet2017Performance} due to its mathematically involved nature compared to, e.g., MUSIC and G-MUSIC.

Our contribution is summarized as follows.
\begin{enumerate}
	\item We prove in \Cref{theo:main} that classical ESPRIT algorithm yields \emph{inconsistent} DoA estimates in the large-dimensional regime as $N,T\to \infty$ with $N/T \to c \in (0,\infty)$, except in the special case of uncorrelated and widely-spaced sources (\Cref{rem:wide_uncorrelated}). 
    \item We propose a novel G-ESPRIT method (\Cref{alg:G-ESPRIT}) and demonstrate in \Cref{prop:gesprit} that it provides consistent DoA estimates in the same regime, for both widely-~and~closely-spaced DoAs.
    \item As part of our analysis, we establish a novel bound on the eigenvalue differences between two non-Hermitian matrices in \Cref{theo:consiseigen} of \Cref{subsec:proof_framework}, which is of independent technical interest.
\end{enumerate}

\textit{Organization of the paper:}
The remainder of this paper is organized as follows.
In \Cref{sec:models}, we present the system model and review the classical ESPRIT algorithm. 
Additionally, we summarize some useful results from RMT on the eigenspectral behavior of the SCM in the large $N, T$ regime.
In \Cref{sec:inconsist}, we demonstrate the \emph{inconsistency} of classical ESPRIT in large-dimensional regime and provide related discussions.
In \Cref{sec:G-ESPRIT}, we introduce the G-ESPRIT method, which addresses the large-dimensional inconsistency of classical ESPRIT, and is shown to provide consistent estimates of both widely-~and~closely-spaced DoAs.
Simulation results supporting our theoretical findings are presented in \Cref{sec:simu}.
Finally, \Cref{sec:conclusion} concludes the paper.

\medskip

\textit{Notions:} Upper-case and lower-case boldface letters denote matrices and column vectors, respectively.
The operators $(\cdot)^\T$, $(\cdot)^{*}$, and $(\cdot)^\H$ denote the transpose, conjugate, and Hermitian transpose, respectively.
We denote $\RR$ the set of real numbers and $\CC$ the set of imaginary numbers, respectively.
For a matrix $\A$, we denote by $\tr(\A)$ and $\det(\A)$ its trace and determinant.
We use $\| \cdot \|$ to denote the Euclidean norm for vectors and spectral/operator norm for matrices.
$ \diag(\textbf{a}) $ returns a diagonal matrix with the elements in $\textbf{a}$ as its main diagonal entries.
Additionally, $\textbf{I}_n$ denotes $n \times n$ identity matrix and $\odot$ represents the Hadamard matrix product.
For a complex number $z$, we use $\Re[z]$, $\Im[z]$, and $\arg(z)$ to represent the real and imaginary parts, as well as the argument of $z$ respectively.
For a random variable $x$, $\EE[x]$ denotes its expectation.
We use $\NN(m,\sigma^2)$ for the real Gaussian distribution with mean $m$ and variance $\sigma^2$, and say $z$ follows a complex circular Gaussian distribution and denote $z \sim \mathcal{CN}(m,\sigma^2)$ if $z = x + \imath y$ with independent $x,y$ such that $x \sim \NN(\Re[m], \sigma^2/2)$ and $y \sim \NN(\Im[m], \sigma^2/2)$.
We use $O(\cdot)$ and $o(\cdot)$ notations as in standard asymptotic statistics~\cite{van2000asymptotic}.

\section{System Model and Preliminaries}
\label{sec:models}
In this section, we present the system model under study in \Cref{subsec:models} and revisit the ESPRIT algorithm in \Cref{subsec:ESPRIT}.
In \Cref{subsec:review_rmt}, we review some results on the eigenspectra of large sample covariance matrices as $N,T \to \infty$ at the same pace, to be used in the large-dimensional analysis of ESPRIT.
\subsection{System Model}
\label{subsec:models}

\begin{figure}[t]
    \centering
    \begin{tikzpicture}[scale=1.0]
    \def\numElements{2}
    \def\elementSpacing{2*\columnwidth/16}
    \def\AnteHeight{1*\columnwidth/16}
    \def\AnteBias{6*\columnwidth/16}
    \draw[line width=1pt] (0,0) -- (14*\columnwidth/16,0);
    \foreach \i in {1,...,\numElements} {
        \draw[line width=1pt] (\i*\elementSpacing,0) -- (\i*\elementSpacing,\AnteHeight);
        \draw[line width=1pt] 
        (\i*\elementSpacing,\AnteHeight) -- (\i*\elementSpacing+\AnteHeight/4,
        \AnteHeight+\AnteHeight/2 *0.866) -- 
        (\i*\elementSpacing-\AnteHeight/4,\AnteHeight+\AnteHeight/2*0.866)
         -- cycle ;
    }
    \foreach \i in {1,...,6}{
    \fill(2.3*\elementSpacing+\i*0.2*\elementSpacing,\AnteHeight*0.75) circle(1pt);
    }
    \foreach \i in {1,...,3} {
        \draw[line width=1pt] (\i*\elementSpacing+\AnteBias,0) -- (\i*\elementSpacing+\AnteBias,\AnteHeight);
        \draw[line width=1pt] 
        (\i*\elementSpacing+\AnteBias,\AnteHeight) -- (\i*\elementSpacing+\AnteHeight/4+\AnteBias,
        \AnteHeight+\AnteHeight/2 *0.866) -- 
        (\i*\elementSpacing-\AnteHeight/4+\AnteBias
        ,\AnteHeight+\AnteHeight/2*0.866)
         -- cycle ;
    }
    \foreach \i in {1,...,\numElements} {
        \node[below] at (\i*\elementSpacing,0) {\scriptsize$x_{\i}(t)$};
    }
    \node[below] at (1*\elementSpacing+\AnteBias,0) {\scriptsize$x_{N-2}(t)$};
    \node[below] at (2.2*\elementSpacing+\AnteBias,0) {\scriptsize$x_{N-1}(t)$};
    \node[below] at (3.2*\elementSpacing+\AnteBias,0) {\scriptsize$x_{N}(t)$};
\coordinate (A) at (\elementSpacing+0.4*\elementSpacing,\AnteHeight+\AnteHeight/2*0.866+1.5*\AnteHeight);
\coordinate (B) at(\elementSpacing,\AnteHeight+\AnteHeight/2*0.866);
\coordinate (C) at 
(\elementSpacing,\AnteHeight+\AnteHeight/2*0.866+1.5*\AnteHeight);

\draw[dashed,line width = 1pt] (B) -- (C);
\draw[line width = 1pt,->]  (A) --(B);

\pic[draw,angle radius=0.3cm,"$\theta$",angle eccentricity = 2]{angle = A--B--C};

    \draw[<->,line width = 1pt] (1.1*\elementSpacing,\AnteHeight*0.75)
    --  (1.9*\elementSpacing,\AnteHeight*0.75) node[midway,above]{$d$};
    \end{tikzpicture}
\caption{System diagram for DoA estimation. A far-field signal with incident angle $\theta$ impinges on a ULA of $N$ sensors spaced $d$ apart.}
\label{fig:ULA}
\end{figure}

In this paper, we consider a uniform linear array (ULA) of $N$ sensors that receives $K$ narrow-band and far-field source signals with DoA $\theta_1, \ldots, \theta_K$ as shown in Fig.~\ref{fig:ULA}.
The received signal at time $t= 1,\ldots, T$ is given by
\begin{equation}
\label{eq:def_X}
	\x(t) = \textstyle \sum_{k = 1}^K \mathbf{a}(\theta_k) s_k(t) + \mathbf{n}(t) \in \CC^N,
\end{equation}
with complex signal $s_k(t) \in \CC$, and complex circular Gaussian white noise $\mathbf{n}(t) \in \CC^N$ having i.i.d.\@ $\mathcal{CN}(0,\sigma^2)$ entries.
$\mathbf{a}(\theta_k) \in \CC^N$ represents the steering vector of source $k \in \{ 1, \ldots, K \}$ at DoA $\theta_k$, given by\footnote{The normalization by $\sqrt N$ is made so that $\mathbf{a}(\theta_k)$ is of unit norm. 
Here, we use $\theta_k$ for the DoA in the Fourier space as in \cite{Vallet2015Performance}, which is related to the ``physical'' angle $\phi_k$ of the source wave via $\theta_k = \frac{2\pi d}{\lambda_0} \sin(\phi_k)$.
}
\begin{equation}\label{eq:def_ULA}
	\mathbf{a}(\theta_k) = [1,~e^{\imath \theta_k}, \ldots, e^{\imath (N-1) \theta_k}]^\T/\sqrt N \in \CC^N.
\end{equation}

This model can be rewritten in matrix form as
\begin{equation}\label{eq:model_matrix_form}
    \X = \A \S + \mathbf{N}, \quad \A = [\mathbf{a}(\theta_1), \ldots, \mathbf{a}(\theta_K)] \in \CC^{N \times K},
\end{equation}
with $\X = [\x(1), \ldots, \x(T)] \in \CC^{N \times T}$ the matrix of received signals, $\A \in \CC^{N \times K}$ the matrix of steering vectors, $\S = [\s(1), \ldots, \s(T)] \in \CC^{K \times T}$ the matrix containing source signals, and random noise $\mathbf{N} = [\mathbf{n}(1), \ldots, \mathbf{n}(T)] \in \CC^{N \times T}$ modeled as spatially white with covariance $\sigma^2 \mathbf{I}_N$.
The source signals $\mathbf{s}(t)$ are assumed to be independent with the noise $\mathbf{n}(t)$.
Each source vector $\mathbf{s}(t)$ is modeled as a zero-mean random vector with covariance $\mathbf{P}=\mathbb{E}[\mathbf{s}(t)\mathbf{s}(t)^{\H}]$.
Then, the \emph{population} covariance of the received signal is given by
\begin{align}
    \vspace{-2pt}
    \label{eq:def_C}
    \C \equiv \EE[\X \X^\H]/T &= \EE [\A \mathbf{S} \mathbf{S}^\H \A^\H]/T + \EE[\mathbf{N} \mathbf{N}^\H]/T \notag\\
    &= \mathbf{A}\mathbf{P}\mathbf{A}^{\H} + \sigma^2\mathbf{I}_N.
    \vspace{-2pt}
\end{align}
Note from \eqref{eq:def_C} that the population covariance $\C$ is highly structured, in the sense that its top subspace relates to the subspace spanned by the steering vectors $\mathbf{a}(\theta_k)$, and thus provides information of the desired DoAs $\theta_k$.
The eigenspace associated with the $K$ largest eigenvalues of $\C$ is referred to in the literature as the ``signal subspace'' $\U_K$.
Since $\C$ is not available in practical situations, subspace methods are performed on the SCM constructed from $T$ observations as
\begin{align}
    \vspace{-2pt}
\label{eq:def_hat_C}
    \hat{\C} = \X \X^\H/T.
    \vspace{-2pt}
\end{align}

\subsection{The ESPRIT DoA Estimator}
\label{subsec:ESPRIT}

The ESPRIT method~\cite{32276} relies on the following structure of rotational invariance: For steering matrix $\A \in \CC^{N \times K}$ defined in \eqref{eq:model_matrix_form} and $\J_1$, $\J_2 \in \RR^{n \times N}$ two selection matrices that select $n$ out of $N$ rows of $\A$ with distance $\Delta \geq 1$, that is
\begin{equation}
\label{eq:def_J}
\vspace{-1pt}
    \J_1^\T = [\ee_\ell,\ldots,\ee_{n+\ell-1}],~\J_2^\T = [\ee_{\ell+\Delta},\ldots,\ee_{n+\ell+\Delta-1}],
    \vspace{-1pt}
\end{equation}
for $\ee_i$ the canonical vector of $\RR^{N}$ such that $[\ee_{i}]_{j} = \delta_{ij}$. 
Note that $\A$ is a Vandermonde matrix and satisfies
\begin{equation}\label{eq:ESPRIT_1}
    \vspace{-1pt}
    \J_1 \A \diag \{ e^{\imath \Delta \theta_k} \}_{k=1}^K = \J_2 \A .
    \vspace{-1pt}
\end{equation}
While $\A$ is unknown, it follows from \eqref{eq:def_C} that the top-$K$ subspace $\U_K \in \CC^{N \times K}$ of $\C$ is the same as the subspace spanned by the columns of $\A \P^{-1/2}$, so that
\begin{equation}\label{eq:ESPRIT_2}
    \vspace{-1pt}
    \U_K = \A \P^{-1/2} \M,
    \vspace{-1pt}
\end{equation}
for some invertible $\M \in \CC^{K \times K}$.
ESPRIT algorithm then exploits the rotational invariance property of signal subspaces spanned by the two subarrays selected by $\J_1$ and $\J_2$.
Combing \eqref{eq:ESPRIT_1}~with~\eqref{eq:ESPRIT_2}, the DoAs $\theta_k$ can be written as the angles of the $k$th complex eigenvalues of 
\begin{equation}\label{eq:def_Phi}
    \vspace{-1pt}
    \bPhi = (\U_K^\H \J_1^\H \J_1 \U_K)^{-1} \U_K^\H \J_1^\H \J_2 \U_K \equiv \bPhi_1^{-1} \bPhi_2, 
    \vspace{-1pt}
\end{equation}
assuming invertible $\bPhi_1 \equiv \U_K^\H \J_1^\H \J_1 \U_K \in \CC^{K \times K}$.
While the population signal subspace $\U_K$ is not practically available, ESPRIT proposes to estimate the DoAs by replacing $\U_K$ in \eqref{eq:def_Phi} with the empirical estimate $\hat \U_K$ obtained from the SCM $\hat \C$, assuming that $\hat \C$ is ``close'' to the population covariance $\C$ in some sense.
This leads to the ESPRIT DoA estimation procedure summarized in \Cref{alg:TraESPRIT}.

\begin{algorithm}[t!]
	\renewcommand{\algorithmicrequire}{\textbf{Input:}}
	\renewcommand{\algorithmicensure}{\textbf{Output:}}
	\caption{ESPRIT DoA estimation.}
	\label{alg:TraESPRIT}
	\begin{algorithmic}[1]
		\REQUIRE Received signal $\X \in \CC^{N \times T}$, number of sources $K$.
		\ENSURE Estimated DoA angles $\hat \theta_k, k \in \{ 1,\ldots, K\}$.
        \STATE Compute the SCM $\hat \C = \X \X^\H/T$ as in \eqref{eq:def_hat_C} to retrieve $\hat \U_K = [\hat \uu_1, \ldots, \hat \uu_K] \in \CC^{N \times K}$ the estimated signal subspace composed of the top-$K$ eigenvectors $\hat \uu_1, \ldots, \hat \uu_K \in \CC^N$ associated to the largest $K$ eigenvalues of $\hat \C$;

		\STATE Define two selection matrices $\J_1, \J_2 \in \RR^{n \times N}$ as in \eqref{eq:def_J} that both select $n$ among $N$ rows with a distance $\Delta \geq 1$;

        \STATE Compute $\hat{\mathbf{\Phi}} = (\hat \U_K^\H \J_1^\H \J_1 \hat \U_K)^{-1} \hat \U_K^\H \J_1^\H \J_2 \hat \U_K \in \CC^{K \times K}$, for invertible $\hat \U_K^\H \J_1^\H \J_1 \hat \U_K$, and then the \emph{angles} of $\lambda_k(\hat{\mathbf{\Phi}})$, the $k$th (complex) eigenvalue of $\hat{\mathbf{\Phi}}$;

		\STATE \textbf{return} $\hat \theta_k = \arg(\lambda_k(\hat{\mathbf{\Phi}}))/\Delta, k \in \{ 1,\ldots, K\}$.
	\end{algorithmic}  
    \vspace{-2pt}
\end{algorithm}

\subsection{Eigenspectral Inconsistency for Large-dimensional SCM}
\label{subsec:review_rmt} 

ESPRIT relies on the assumption that the signal subspace $\U_K$ can be accurately estimated by $\hat \U_K $.
This is typically valid when the number of observations $T$ is much larger than the array size $N$, making the sample covariance $\hat \C$ a good ``proxy'' of population covariance $\C$ in the sense that $\| \hat \C - \C \| \to 0$ as $T \to \infty$ for fixed $N$, by the law of large numbers.
However, in the case of large arrays and/or limited snapshots, where $N$ and $T$ are of the same order of magnitude, $\hat \C$ is \emph{not} a consistent estimator of $\C$ in a spectral norm sense. 
Consequently, we should \emph{not} expect that the top subspace $\hat \U_K$ used in ESPRIT is a good estimate of the true signal subspace $\U_K$.

In the following, we recall a few results from large-dimensional RMT that provide precise eigenspectral characterizations of SCM in the large $N,T$ regime.
We note that throughout this paper, we consider a regime in which the number of sources $K$ is fixed as $N,T \to \infty$. 
    This assumption is standard in large-dimensional analyses of subspace-based methods and spiked covariance models \cite{hachem2013subspace,baik2005phase}. 
    Allowing $K$ to scale with $N$ would destroy the separability between the signal and noise subspaces and fall outside the scope of the present work. Accordingly, we adopt the following assumption.

\begin{Assumption}[Large arrays and limited snapshots]\label{ass:large_array}
As $T \to \infty$, $N/T \to c \in (0,\infty)$, $n/N \to \tau \in (0,1)$ and $K$ fixed.
\end{Assumption}

\begin{Assumption}[Subspace separation]\label{ass:subspace}
Consider the eigen-decomposition of $\A \P \A^\H \in \mathbb{C}^{N \times N}$ in \eqref{eq:def_C} as
\begin{align}
\label{eq:def_ULU}
    \textstyle \A \P \A^\H = \sum_{k=1}^K \lambda_k(\A \P \A^\H) \cdot \uu_k \uu_k^\H.
\end{align}
Define the signal-to-noise ratio (SNR) of the $k$-th source as $\ell_k = \lambda_k(\A\mathbf{P}\A^{\H})/\sigma^{2}$.
Then, as $N,T\to\infty$, the top eigenvalues satisfy the separation condition
\begin{align}
    \ell_1 > \cdots > \ell_K > \sqrt{c}, 
\end{align}
where $c = \lim N/T$ in \Cref{ass:large_array}.
\end{Assumption}

For the eigenspectral characterization of large (random or deterministic) matrices, we define the empirical spectral measure and its Stieltjes transform as follows.
\begin{Definition}[Empirical spectral measure]\label{def:mu}
For a Hermitian matrix $\X \in \mathbb{C}^{N \times N}$, its empirical spectral measure is defined as the normalized counting measure of the eigenvalues $\lambda_1(\X), \ldots, \lambda_N(\X)$ of $\X$,
\begin{equation}
    \mu_{\X} = \frac{1}{N}\sum_{i=1}^N \delta_{\lambda_i(
    \X)},
\end{equation}
where $\delta_x$ represents the Dirac measure at $x$.
\end{Definition}

\begin{Definition}[Stieltjes transform]\label{def:ST}
For a probability measure $\mu$ (e.g., an empirical spectral measure in \Cref{def:mu}), its Stieltjes transform $ m_{\mu}(z)$ is defined, for $z \in \mathbb{C}\backslash \supp(\mu)$, as
\begin{equation}
    m_{\mu}(z) = \int \frac{\mu(dt)}{t-z} = \frac{1}{N} \tr \mathbf{Q}_\X(z),
\end{equation}
with $\mathbf{Q}_\X(z) = (\X- z \I_N)^{-1} \in \CC^{N \times N}$ the \emph{resolvent} of $\X$.
\end{Definition}
The resolvent and Stieltjes transform provide convenient access to the eigenspectral behavior of large random matrices.
We refer the interested readers to \cite[Section~2]{couillet2022RMT4ML} as well as \cite{hachem2007deterministic,rozanov1967stationary} for more discussions.

Under Assumptions~\ref{ass:large_array}~and~\ref{ass:subspace}, we have the following results, due to a sequence of previous efforts~\cite{marvcenko1967distribution,silverstein1995empirical,baik2006eigenvalues,benaych2011eigenvalues}.

\begin{Theorem}[Eigenspectral characterization of large-dimensional SCM~\cite{marvcenko1967distribution,silverstein1995empirical,baik2006eigenvalues,benaych2011eigenvalues}] \label{theo:spikes}
Under the settings and notations of \Cref{ass:large_array}, we have,  for $\X \in \CC^{N \times T}$ defined in \eqref{eq:model_matrix_form} and as $N,T \to \infty$ with $N/T \to c \in (0, \infty)$ that, with probability one, the empirical spectral measure in \Cref{def:mu} of the SCM $\hat \C = \X \X^\H/T$ converges weakly to the Mar{\u c}enko-Pastur law
\begin{equation*}
	\mu(dx) = (1+c^{-1})^+ \delta_0(x) + \frac{\sqrt{\big(x-E_-\big)^+ \big(E_+ - x\big)^+}}{2\pi c \sigma^2 x}\,dx,
\end{equation*}
with $E_\pm = \sigma^2 (1 \pm \sqrt c)^2$ and $(x)^+ = \max(x,0)$.
Its Stieltjes transform (see \Cref{def:ST}) converges to $m(z)$, the unique Stieltjes transform solution to the Mar\u{c}enko-Pastur equation~\cite{marvcenko1967distribution}
\begin{equation}
\label{eq:m(z)}
    z c \sigma^2 m^2(z) + \big(z + \sigma^2 (c-1)\big) m(z)  + 1 = 0.
\end{equation}
Moreover, let \Cref{ass:subspace} hold and denote $\hat \lambda_1 > \ldots > \hat \lambda_N$ the ordered eigenvalues of $\hat \C$ with corresponding eigenvectors $\hat{\mathbf{u}}_{1},\ldots,\hat{\mathbf{u}}_N$, we have 
\begin{equation}\label{eq:eigenvalue_bias}
    \hat \lambda_i \to
    \begin{cases}
        \bar \lambda_i
        = \sigma^2\!\left( 1 + \ell_i
        + c \dfrac{1+ \ell_i}{\ell_i} \right)
        >  E_+, & i \leq K,\\[4pt]
        E_+ =  \sigma^2 (1+\sqrt c)^2, & i > K
    \end{cases}; 
\end{equation} 
almost surely as $N,T \to \infty$.
Also, for all deterministic sequences of unit norm vectors $\mathbf{a},\mathbf{b} \in \mathbb{C}^{N}$, we have, for $\ell_k$ defined in \Cref{ass:subspace},
\begin{equation}\label{eq:eigenvectors_bias}
    \mathbf{a}^\H \hat{\mathbf{u}}_k
    \hat{\mathbf{u}}_k^\H
    \mathbf{b} - \frac{1-c\ell_k^{-2}}{1 + c\ell_k^{-1}}
    \mathbf{a}^\H \mathbf{u}_k
    \mathbf{u}_k^\H
    \mathbf{b} \to 0,~ k \in \{1,\ldots, K\},
\end{equation}
almost surely as $N,T \to \infty$, with $\uu_k\equiv \uu_k(\A \P \A^{\H})$ in \eqref{eq:def_ULU}.
\end{Theorem}
\noindent
\Cref{theo:spikes} states that for $N,T$ both large and comparable with ratio $c = \lim N/T$, the eigenvalues of the SCM $\hat \C$, instead of being close to those of its population counterpart $\C = \A \mathbf{P} \A^\H + \sigma^2 \I_N$ in \eqref{eq:def_C}, spread out on the interval $[E_-, E_+]$ of length $4 \sigma^2 \sqrt c \gg 0$.
Moreover, under the additional \Cref{ass:subspace}, it is known that the largest eigenvalues of $\hat \C$ (that are due to the ``signal'' $\A \P \A^\H$) are guaranteed to ``separate'' from those due to the random white noise.
However, even in this case, the empirical eigenvalues $\hat \lambda_i$ of $\hat \C$ are larger than the population ones (i.e., $\sigma^2 (1 + \ell_i)$), by a quantity that is proportional to $c = \lim N/T$, and eigenvectors $\hat{\mathbf{u}}_k$ are \emph{biased} estimate of the population eigenvectors $\uu_k(\A \P \A^{\H})$, in the sense that for arbitrary deterministic $\mathbf{a}, \mathbf{b} \in \CC^N$, the eigenspace $\hat{\mathbf{u}}_k
\hat{\mathbf{u}}_k^\H$ is \emph{biased} by a factor of $(1-c\ell_k^{-2} )/( 1 + c\ell_k^{-1})$ as in \eqref{eq:eigenvectors_bias}.
\Cref{theo:spikes} can be proven using the so-called ``Deterministic Equivalent for resolvent'' analysis framework.
This approach will be constantly exploited in the proof of our technical results in this paper.
We refer the interested readers to \Cref{sec:DE} of the appendix as well as \cite[Chapter~2]{couillet2022RMT4ML} for a detailed discussion of this approach.

\section{Inconsistency of ESPRIT for Large Arrays}
\label{sec:inconsist}
In this section, we present our main result in \Cref{theo:main} on the large-dimensional inconsistency of classical ESPRIT in \Cref{subsec:analysis_ESPRIT}.
The proof of \Cref{theo:main} relies on a novel bound on eigenvalue difference between two matrices derived in \Cref{subsec:proof_framework}, and is given in detail in \Cref{subsec:proof_of_main}.

Built upon recent advances in RMT, we perform in this section an in-depth analysis of the classical ESPRIT method in \Cref{alg:TraESPRIT} in the large array and limited snapshot setting of \Cref{ass:large_array}.
We show that, in general, classical ESPRIT provides \emph{inconsistent} estimates of the DoAs in the following two scenarios: \emph{widely-spaced} DoAs and \emph{closely-spaced} DoAs, defined respectively as follows.
\begin{Assumption}[Widely-spaced DoAs]\label{ass:widely_spaced}
    The DoAs $\theta_1, \ldots, \theta_K$ are fixed as $N \to \infty$, and they have angular separation much larger than a beam-width $2\pi/N$.
\end{Assumption}
\begin{Assumption}[Closely-spaced DoAs]\label{ass:closely_spaced}
    The DoAs $\theta_1,\ldots,\theta_K$ are spaced with a ``distance'' of order $O(N^{-1})$, that is
    \begin{equation}
        | \theta_k - \theta | = O(N^{-1}), \quad k \in \{1, \ldots, K \},
    \end{equation}
    for some $\theta > 0$ as $N \to \infty$.
\end{Assumption}

In the case of closely-spaced DoAs in \Cref{ass:closely_spaced}, the DoAs are within a ``distance'' of order $O(N^{-1})$.
As such, for an estimate $\hat \theta_k$ of the true DoA $\theta_k$ to be distinguished from other DoAs, one must have an estimation error of order $o(N^{-1})$.
We say, in this case, that the estimate $\hat \theta_k$ is $N$-consistent if $N(\hat{\theta}_k-\theta_k) \to 0$ as $N,T \to \infty$, see also~\cite{Vallet2015Performance}. 
Note that this differs from the widely-spaced DoA scenario in \Cref{ass:widely_spaced}, in which case we only need $\hat{\theta}_k$ to be a consistent estimation of $\theta_k$, that is $\hat{\theta}_k-\theta_k \to 0$ as $N,T \to \infty$.

The fundamental distinction between widely-~and closely-spaced DoAs can be further understood from the structure of the array response.
Under \Cref{ass:large_array} and in the case of widely-spaced DoAs in \Cref{ass:widely_spaced}, we have, as $N,n,T \to \infty$ at the same pace that $\| \A^\H \A - \I_K \| = O(N^{-1})$, so that the steering matrix $\A$ is (approximately for $N$ large) the same as the signal subspace $\U_K$, and that both $\A^\H  \J_1^\H \J_1 \A$ and $\A^\H  \J_1^\H \J_2 \A$ are asymptotically diagonal.
On the contrary, in the case of closely-spaced DoAs in \Cref{ass:closely_spaced}, $\A^\H \A$ is no longer asymptotically identity, and $\A^\H  \J_1^\H \J_1 \A$ and $\A^\H  \J_1^\H \J_2 \A$ are no longer asymptotically diagonal.
See \Cref{lem:AJJA} in \Cref{sec:techLemmas} of the appendix for a detailed characterization of these quantities.

\subsection{Large-dimensional behavior of ESPRIT}
\label{subsec:analysis_ESPRIT}

According to \Cref{alg:TraESPRIT}, the ESPRIT DoA estimates depend on the angles of the eigenvalues of 
\begin{equation}\label{eq:def_hat_Phi}
	\hat {\mathbf{\Phi}} = (\hat \U_K^\H \J_1^\H \J_1 \hat \U_K)^{-1} \hat \U_K^\H \J_1^\H \J_2 \hat \U_K \equiv \hat{\mathbf{\Phi}}_1^{-1} \hat{\mathbf{\Phi}}_2,
\end{equation}
where $\hat{\bPhi}, \hat{\bPhi}_1, \hat{\bPhi}_2$ are empirical estimates of their population counterparts in \eqref{eq:def_Phi}. 

In the following result, by studying the behavior of $\hat {\mathbf{\Phi}}$, we provide a precise large-dimensional characterization of the classical ESPRIT method in the large array and limited snapshot regime.
The key idea and technical challenges to prove \Cref{theo:main} will be discussed in \Cref{subsec:proof_of_main}.
\begin{Theorem}[Large-dimensional behavior of ESPRIT]\label{theo:main}
    Let Assumptions~\ref{ass:large_array} and~\ref{ass:subspace} hold.
    In addition, suppose that either \Cref{ass:widely_spaced} or \Cref{ass:closely_spaced} is satisfied.
    Define the deterministic matrix $\bar{ \bPhi} = \bar{ \bPhi}_1^{-1} \bar{ \bPhi}_2$ with
    \begin{equation}\label{eq:def_bar_Phi}
    \begin{split}
    \bar{ \bPhi}_1 &= \diag(\sqrt{\mathbf{g}}) \mathbf{\Phi}_1 \diag(\sqrt{\mathbf{g}}) + \tau \left(\I_K - \diag(\mathbf{g}) \right), \\ 
    \bar{ \bPhi}_2 & = \diag(\sqrt{\mathbf{g}}) \mathbf{\Phi}_2 \diag(\sqrt{\mathbf{g}}), 
    \end{split}
    \end{equation}
    where $\sqrt{\mathbf{g}} =[\sqrt{g_1}, \ldots, \sqrt{g_K}]^\T \in \RR^K$ and 
    \begin{equation}\label{eq:def_g}
    g_k \equiv \frac{1-c\ell_k^{-2}}{1 + c\ell_k^{-1}} \in (0,1).
    \end{equation}
    Further assume that $\bar{\bPhi}$ has distinct eigenvalues, and denote by
$\lambda_k(\bar{\bPhi})$ its $k$-th largest eigenvalue (in amplitude). 
Let $\hat{\theta}_k$ denote the DoA estimate obtained from the classical ESPRIT method in \Cref{alg:TraESPRIT}, and let $\bar{\theta}_k \equiv \arg(\lambda_k(\bar{\bPhi}))/\Delta$.
Then, for each $k\in\{1,\ldots,K\}$, we have 
\begin{equation}
    \Delta(\hat \theta_k - \bar{\theta}_k) \to 0,
\end{equation}
almost surely as $N,T \to \infty$.
\end{Theorem}
\noindent
\Cref{theo:main} tells us that in the large $N,T$ regime, the estimates $\hat \bPhi_1, \hat \bPhi_2$, due to the large-dimensional inconsistency of $\hat \C$ discussed in \Cref{subsec:review_rmt}, are \emph{biased} from their population counterparts $\bPhi_1, \bPhi_2$ defined in \eqref{eq:def_Phi}, and connect instead to $\bar \bPhi_1, \bar \bPhi_2$ in \eqref{eq:def_bar_Phi}.
As a direct consequence of \Cref{theo:main}, we have, in the case of large arrays, that ESPRIT method diverges from its original design discussed in \Cref{subsec:models} and should in general \emph{not} be able to provide consistent DoA estimates in neither widely-~nor closely-spaced DoAs scenarios.

In the following, we discuss special cases where the aforementioned large-dimensional inconsistency of classical ESPRIT holds or, by chance, fails.
We first note that, in the limit of infinite snapshots or sufficiently high SNR, the large-dimensional corrections characterized in \Cref{theo:spikes} disappear.
Specifically, as $c=\lim N/T \to 0$ \emph{or} as the SNR $\ell_k \to \infty$, one has $g_k \to 1$, which implies $\bar \bPhi = \bPhi_1^{-1} \bPhi_2 = \bPhi$.
In this regime, the sample-based ESPRIT estimator converges to its population counterpart and yields consistent DoA estimates.
One might also conjecture that a similar consistency could be recovered in the \emph{small-subarray} limit
$\tau=\lim n/N \to 0$, since the large-dimensional bias term $\diag(\sqrt{\mathbf{g}})$
in both $\bar{\bPhi}_1$ and $\bar{\bPhi}_2$ appears to vanish.
However, this intuition is misleading, since under both Assumptions~\ref{ass:widely_spaced}~and~\ref{ass:closely_spaced}, one has $\bPhi_1 \to 0$ as $\tau \to 0$, so that $\bar \bPhi = \bar\bPhi_1^{-1} \bar \bPhi_2$ is no longer well defined.

Beyond the limiting case discussed above, classical ESPRIT also holds consistency in the uncorrelated and widely-spaced DoAs scenario, as in the following remark.
\vspace{-1pt}
\begin{Remark}[Special case: widely-spaced DoAs with uncorrelated sources]\normalfont
\label{rem:wide_uncorrelated}
For widely-spaced DoAs in \Cref{ass:widely_spaced}, assume in addition that the sources are \emph{uncorrelated} so that $\mathbf{P}=\mathbb{E}[\mathbf{s}(t)\mathbf{s}(t)^{\H}]$ is diagonal.
In this case, the top-$K$ subspace $\U_K$ is approximately the same as that spanned by the steering vectors, and that $\bPhi_1 = \tau \I_K + O_{\| \cdot \|}(N^{-1})$, $\bPhi_2 = \tau \diag\{ e^{\imath \Delta \theta_k} \}_{k=1}^K + O_{\| \cdot \|}(N^{-1})$, so that $\bar \bPhi$ defined in \eqref{eq:def_bar_Phi} writes
    \vspace{-2pt}
\begin{equation*}
    \bar{ \mathbf{{\Phi}}} = \diag(\sqrt{\mathbf{g}}) \diag\{e^{\imath \Delta \theta_k}\}_{k=1}^K \diag(\sqrt{\mathbf{g}}) + O_{\| \cdot \|}(N^{-1}),
    \vspace{-1pt}
\end{equation*}
for $\mathbf{g}$ a real vector.
As such, $\bar{\bPhi}$ has the same eigenvalue \emph{angles} as $\bPhi$, and that $\hat \theta_k - \theta_k \to 0$ almost surely. 
Classical ESPRIT thus provides consistent DoA estimation in this setting. 
\end{Remark}
Beyond the \emph{small-subarray} limiting case and the special scenarios discussed in Remark~\ref{rem:wide_uncorrelated}, where the classical ESPRIT estimates $\hat \theta_k$ obtained from \Cref{alg:TraESPRIT} are ``lucky'' enough to be consistent, it can be shown that $\hat \theta_k$ in general deviates from the true DoA $\theta_k$ in the large $N,T$ regime.
We provide in the following two examples: widely-spaced DoAs with correlated sources and closely-spaced DoAs with equal power sources.

\begin{Remark}[Special case: widely-spaced DoAs with correlated sources]\normalfont
\label{rem:wide_correlated}
In the case of widely-spaced DoAs in \Cref{ass:widely_spaced}, consider a general scenario where the sources are \emph{correlated} with generic $\mathbf{P}=\mathbb{E}[\mathbf{s}(t)\mathbf{s}(t)^{\H}]$. 
This is in contrast to the uncorrelated source setting in \Cref{rem:wide_uncorrelated}.
Denote $\P = \U_{\P} \mathbf{L} \U_{\P}^\H$ the eigen-decomposition of $\P$, it then follows from \Cref{lem:AJJA} in \Cref{sec:techLemmas} of the appendix that $\bPhi_1 = \U_{\P}^\H \A^\H \J_1^\H \J_1 \A \U_{\P} + O_{\lVert \cdot \rVert}(N^{-1}) = \tau \I_K + O_{\lVert \cdot \rVert}(N^{-1}) $ and similarly that $\bPhi_2 = \tau \U_{\P}^\H \diag\{ e^{\imath \Delta \theta_i} \}_{i=1}^K \U_{\P} + O_{\lVert \cdot \rVert}(N^{-1})$, so that
\begin{equation*}
	\bar{ \mathbf{{\Phi}}} = \diag(\mathbf{g})\U_{\P}^\H  \diag\{e^{\imath \Delta \theta_i}\}_{i=1}^{k} \U_{\P} \diag(\mathbf{g}) + O_{\lVert \cdot \rVert}(N^{-1}).
\end{equation*}
As such, $\bar{ \mathbf{{\Phi}}}$ has, in general, its eigenvalues different from those of $\bPhi$.
This, by \Cref{theo:main}, leads to \emph{inconsistent} ESPRIT estimates such that $\arg(\lambda_k(\hat{\bPhi}))/\Delta - \theta_k \nrightarrow 0$ as $N,T \to \infty$.

It can be checked, in the case of $K=2$ sources with different DoAs $\Delta \theta_1 \neq \Delta \theta_2 + m \pi$ for positive integer $m$, that the classical ESPRIT estimates \emph{cannot} be consistent \emph{unless} $\U_\P = \I_2$, that is, when the two sources are \emph{uncorrelated}.
See \Cref{subsec:proof_of_rem_wide_correlated} in the appendix for a detailed exposition of this counterexample.
\end{Remark}

\begin{Remark}[Special case: closely-spaced DoAs with equal power sources]\normalfont
\label{rem:closely_equal}
In the case of closely-spaced DoAs in \Cref{ass:closely_spaced}, we consider $K=2$ sources with $\theta_2=\theta_1+\alpha/N$ for some $\alpha >0$, and assume uncorrelated signals with equal powers, that is, $\P=\I_2$.
It can be checked that classical ESPRIT is \emph{not} $N$-consistent in this case, that is $N(\arg(\lambda_k(\hat{\bPhi}))/\Delta - \theta_k) \nrightarrow 0$ as $N,T \to \infty$, see \Cref{subsec:proof_rem_closely_equal} in the appendix for a detailed proof.
\end{Remark}

\subsection{A novel bound on eigenvalue difference}
\label{subsec:proof_framework}
In this section, we first establish sufficient conditions for eigenvalue approximation to prove Theorem~2, which are summarized in \Cref{theo:main}. 
In particular, these conditions are shown to hold for $\hat\bPhi$ and $\bar\bPhi$ in \Cref{theo:consiseigen} of \Cref{subsec:proof_of_main}, and may also be of independent interest.
We then provide the proof of \Cref{theo:main}.
The major technical challenge in characterizing the large-dimensional behavior of ESPRIT in \Cref{alg:TraESPRIT} lies in the fact that the corresponding DoA estimates, which are the angles of the complex eigenvalues of the $K$-by-$K$ random matrix $\hat \bPhi$ defined in \eqref{eq:def_hat_Phi}, depend on the entries of two \emph{strongly dependent} random matrices $\hat \bPhi_1$ and $\hat \bPhi_2$ in a non-trivial manner.
In addition, the off-diagonal complex entry of $\hat \bPhi_2$ in \eqref{eq:def_hat_Phi}, for $i\neq j$, is given by 
\begin{equation}
    [\hat \bPhi_2]_{ij} = \hat \uu_i^\H \J_1^\H \J_2 \hat \uu_j = \textstyle\sum_{m=\ell}^{n+\ell-1} \ee_{m+\Delta}^\H \hat \uu_i \hat \uu_j^\H \ee_{m},
\end{equation}
and cannot be handled using standard RMT techniques.
The same holds true for $[\hat \bPhi_1]_{ij}$, the off-diagonal entries of $\hat \bPhi_1$.
Indeed, standard RMT and contour integration techniques provide direct access to the following bilinear forms in the large $N,T \to \infty$ limit,
\begin{equation}
    \mathbf{a}^\H \hat \uu_i \hat \uu_i^\H \mathbf{b},
\end{equation}
for $\mathbf{a}, \mathbf{b} \in \CC^N$ of bounded norm, see for example \cite[Section~2.5]{couillet2022RMT4ML}. 
This thus provides access to (limits of) the diagonal entries of $\hat \bPhi_1$ and $\hat \bPhi_2$, but \emph{not} their off-diagonal entries.

To resolve this technical challenge, we introduce the following bound on the eigenvalue difference between two (possibly non-Hermitian) matrices, using the products of their entries with indices forming a circle. 
\begin{Theorem}[Eigenvalue approximation between two matrices] \label{theo:consiseigen}
    Let $\A, \B \in \CC^{K \times K}$ be two matrices, and assume that $\A$ has only simple eigenvalues. Further assume that for every $m$-cycle of indices 
    $1 \le i_1 < \cdots < i_m \le K$, their entries satisfy, for some $\varepsilon \in(0,1)$,
    \begin{equation}\label{eq:condi}
    |A_{i_1 i_2} A_{i_{2} i_{3}} \ldots A_{i_{m} i_{1}} - B_{i_1 i_2} B_{i_{2} i_{3}} \ldots B_{i_{m} i_{1}}| \leq \varepsilon.
    \end{equation}
    Then there exists a permutation $\Pi$ on $\{1,\ldots,K\}$ and a constant $C>0$, depending only on $K$ and 
    $\max\{ \max_{i,j}|A_{ij}|,\; \max_{i,j}|B_{ij}|\}$, such that
    \begin{equation}
    \bigl|\lambda_k(\A)-\lambda_{\Pi(k)}(\B)\bigr|
    \;\le\; C\,\varepsilon^{1/K},
    \qquad k=1,\ldots,K . 
    \end{equation}
\end{Theorem}

\begin{proof}[Proof of \Cref{theo:consiseigen}]
To prove \Cref{theo:consiseigen}, we work on the characteristic polynomials of $\A$ and $\B$.
It is known, e.g., from \cite{horn2012matrix} that the characteristic polynomial of $\A \in \CC^{K \times K}$ writes
\begin{equation}\label{eq:det_principle_minor}
\!\!\det(\lambda \I_K -\A ) = \lambda^K -S_1(\A ) \lambda^{K-1} \ldots + (-1)^K {S}_K(\A ), 
\end{equation}
for $S_k(\A )$ the sum of all $k$-by-$k$ principal minors of $\A $, with $S_1(\A ) = \tr(\A )$ and $S_K(\A ) = \det(\A ) $.
Note that this is a polynomial (of $\lambda$) of degree $K$, and it suffices to evaluate its coefficients (i.e., the sum of principal minors).

Consider one of the $k$-by-$k$ principal minors of two matrices $\A$ and $\B$, denoted by $\A[\mathcal{I}_k]$ and $\B [\mathcal{I}_k]$, respectively, with ordered indices $\mathcal{I}_k: 1\leq i_1 < \ldots < i_k \leq K$, we have, by definition of principal minor, that 
\begin{align} 
    &| \A [\mathcal{I}_k] - \B [\mathcal{I}_k] | \notag \\
    &= \Big| \sum_{\sigma} \sign(\sigma) \Big( \prod_{j \in \mathcal I_k} A_{j \sigma(j)} - \prod_{j \in \mathcal I_k} B_{j \sigma(j)} \Big) \Big| \notag \\
    &\leq \sum_{\sigma} \Big| \prod_{j \in \mathcal I_k} A_{j \sigma(j)} - \prod_{j \in \mathcal I_k} B_{j \sigma(j)} \Big|, \label{eq:hat-bar}
\end{align}
where we denote $\sigma \colon \{ i_1, \ldots, i_k \} \to \{ i_1, \ldots, i_k \}$ the permutations of the index set $\mathcal I_k$ as in \cite[Section 0.3.2]{horn2012matrix} (that are $k!$ distinct permutations).
It is known that every permutation $\sigma$ of a finite set can be uniquely decomposed into a product of disjoint cycles (and the uniqueness is up to the order of the cycles), see for example \cite{scott2012group}.
See \Cref{example:circle_decom} in \Cref{sec:techLemmas} of the appendix for an example of such decomposition.
Then, it follows from \eqref{eq:hat-bar} that
\begin{align*} 
    &| \A [\mathcal{I}_k] - \B [\mathcal{I}_k] | \\
    &\leq \sum_{\sigma} \Big| \prod_{j \in \mathcal I_k} A_{j \sigma(j)} - \prod_{j \in \mathcal I_k} B_{j \sigma(j)}  \Big| \\ 
    &= \sum_{\sigma} \Big| \prod_{i=1}^p \underbrace{A_{j_1 j_2} \ldots A_{j_{m_i} j_1}}_{\text{indices form an $m_i$-node circle}}  - \prod_{i=1}^p  B_{j_1 j_2} \ldots B_{j_{m_i} j_1} \Big| \\
    &=  \sum_{\sigma} \Big| \prod_{i=1}^p \alpha_i  - \prod_{i=1}^p \beta_i \Big|,
\end{align*}
where, in the second line, we use the fact that each permutation $\sigma$ can be decomposed into $1 \le p \le k$ disjoint cycles of lengths $m_1,\ldots,m_p$, and we denote
\begin{align*}
  \alpha_i = A_{j_1 j_2} \ldots A_{j_{m_i} j_1}, \; \beta_i =  B_{j_1 j_2} \ldots B_{j_{m_i} j_1}.
\end{align*}
Since we assume that \eqref{eq:condi} holds for any $m$-node circle, we have $|\alpha_i - \beta_i| \le \varepsilon, i = 1,\ldots,p$. Meanwhile, let $\rho \equiv \max\{ \max_{i,j} |A_{ij}|, \max_{i,j} |B_{ij}| \}$, then 
\begin{align*}
  |\alpha_i| \le \rho^{m_i}, |\beta_i| \le \rho^{m_i}.
\end{align*}
Using the standard telescoping expansion for product differences, we obtain
\begin{align*}
  &\Big| \prod_{i=1}^p \alpha_i  - \prod_{i=1}^p \beta_i \Big| = \sum_{s=1}^p (\alpha_s-\beta_s) \Big( \prod_{i<s} \alpha_i \Big) \Big( \prod_{i>s} \beta_i \Big) \\
  & \le \sum_{s=1}^p |\alpha_s-\beta_s| \Big| \prod_{i<s} \alpha_i \Big| \Big| \prod_{i>s} \beta_i \Big| \le  \sum_{s=1}^p \varepsilon \rho^{k-m_s} \le p \rho^{k-1} \varepsilon.
\end{align*}
Summing over all permutations yields
\begin{align}
  | \A [\mathcal{I}_k] - \B [\mathcal{I}_k] | \le \textstyle \sum_{\sigma} p \rho^{k-1} \varepsilon \le k! p \rho^{k-1} \varepsilon.
\end{align}

As such, for the $k$-th order coefficient of the characteristic polynomial of $\A$ and $\B$ as in \eqref{eq:det_principle_minor}, we have
\begin{align*}
    &|S_k(\A) - S_k(\B)| = \left|\sum\nolimits_{|\mathcal{I}_k|=k}^{ \binom{K}{k} } \left| \A [\mathcal{I}_k] - \B [\mathcal{I}_k] \right| \right| \\
    &\leq  \binom{K}{k} \max_{|\mathcal{I}_k|=k} \left(\left| \A [\mathcal{I}_k] - \B [\mathcal{I}_k] \right| \right) \leq C_K \varepsilon,
\end{align*}
for some constant $C_K$ that only depends on $\rho$ and $K$.
When $K$ is fixed under Assumption~1, $C_K$ is an absolute constant independent of $N$ and $T$, and one has $|S_k(\A) - S_k(\B)| = O(\varepsilon)$.
To convert this coefficient bound into a bound on the eigenvalues, we invoke the following quantitative root-continuity theorem.

\begin{Theorem}[Continuity of roots of a polynomial,~{\cite[Theorem~5]{nathanson2024Continuity}}]\label{theo:continuity}
Let $f(z)$ be a polynomial of degree $K$ with only simple roots, whose coefficients are denoted by $\{a_i\}_{i=0}^K$. 
Then, there exist constants $C'>0$ and $\delta_0>0$, depending only on the coefficients $\{a_i\}_{i=0}^K$, such that for every $0<\delta\le\delta_0$, 
whenever a polynomial $g(z)$ with coefficients $\{b_i\}_{i=0}^K$ satisfies $|a_i-b_i|\le\delta$ for $i$, the polynomial $g(z)$ also has only simple roots and, moreover, each root $r_f$ of $f(z)$ corresponds to a root $r_g$ of $g(z)$ with $|r_f-r_g|\le C'\,\delta^{1/K}$.
\end{Theorem} 

Hence, by combining the bound in (29) with Theorem~4, we obtain that each root of the characteristic polynomial of $\A$ admits a corresponding root (up to a permutation $\Pi(\cdot)$) of that of $\B$ satisfying
\begin{align*}
  \bigl|\lambda_k(\A)-\lambda_{\Pi(k)}(\B)\bigr| \;\le\; C'(C_K \varepsilon)^{1/K}  \;\le\; C\,\varepsilon^{1/K}, 
\end{align*}
for $k=1,\ldots,K,$ and some constant $C$ that depends only on $K$ and $\rho$.
This completes the proof of Theorem 3.
\end{proof}

\subsection{\texorpdfstring{Proof of \Cref{theo:main}}{Proof of Theorem 2}}
\label{subsec:proof_of_main}

Here, we present the detailed proof of \Cref{theo:main}, following the same line of arguments as that of \Cref{theo:consiseigen}.
The major difference is that, to prove \Cref{theo:main}, we need to evaluate the eigenvalues of the \emph{product} of two matrices $\hat \bPhi = \hat \bPhi_1^{-1} \hat \bPhi_2 $ as in \eqref{eq:def_hat_Phi}, as opposed to the single matrix ($\A$) in \Cref{theo:consiseigen}. 

Following the idea of \Cref{theo:consiseigen}, we provide, in the following result, characterization of the diagonal entries of $\hat \bPhi_1, \hat \bPhi_2$ \emph{and} the product of the off-diagonal entries, \emph{when their indices form a circle}.
\begin{Theorem}[Large-dimensional characterization of $\hat\bPhi_1$ and $\hat\bPhi_2$]\label{theo:CDT}
Under the same settings and notations of \Cref{theo:main}, we have, for $\hat\bPhi_1, \hat\bPhi_2$ defined in \eqref{eq:def_hat_Phi}, that 
\begin{enumerate}
    \item their diagonal entries satisfy, for $k \in \{ 1, \ldots, K \}$,
    \begin{equation}
        [\hat\bPhi_1]_{kk} - [\bar\bPhi_1]_{kk} \to 0, \quad [\hat\bPhi_2]_{kk} - [\bar\bPhi_2]_{kk} \to 0,
    \end{equation}
    almost surely as $N,T \to \infty$, for $\bar\bPhi_1, \bar\bPhi_2$ defined in \eqref{eq:def_bar_Phi};
    \item for indices $1\leq k_1 < \ldots < k_m \leq K$ that form an $m$-node cycle, denote $\mathbf{M}_{k_j} = \J_1^\H\J_2~\mbox{or}~\J_1^\H\J_1, j \in \{ 1, \ldots, m \}$, so that the off-diagonal entries $[\hat{\bPhi}_1]_{k  _i k_j}$ or $[\hat{\bPhi}_2]_{k_i k_j}$ of $\hat\bPhi_1, \hat\bPhi_2$ can be uniformly written as $\hat \uu_{k_i}^\H \mathbf{M}_{k_j} \hat \uu_{k_j}$, we have 
    \begin{equation}
        \hat{\psi}_{k_1, \ldots, k_m} - \bar{\psi}_{k_1, \ldots, k_m} \to 0,
    \end{equation}
    with $\hat{\psi}_{k_1, \ldots, k_m} \equiv \hat \uu_{k_m}^\H \mathbf{M}_{k_1} \hat \uu_{k_1} (\prod_{j=1}^{m-1} \hat \uu_{k_j}^\H \mathbf{M}_{k_{j+1}} \hat \uu_{k_{j+1}} )$, $\bar{\psi}_{k_1, \ldots, k_m} \equiv g_1 \uu_{k_m}^\H \mathbf{M}_{k_1} \uu_{k_1} (\prod_{j=1}^{m-1} g_{j+1} \uu_{k_j}^\H \mathbf{M}_{k_{j+1}}  \uu_{k_{j+1}} )$ almost surely as $N,T \to \infty$.\footnote{Note that $\bar{\psi}_{k_1, \ldots, k_m}$ is nothing but the product of off-diagonal entries of $\bar\bPhi_1, \bar\bPhi_2$ defined in \eqref{eq:def_bar_Phi}.}
\end{enumerate}
\end{Theorem}
\begin{proof}[Proof of \Cref{theo:CDT}]
See \Cref{subsec:proof_CDT} of the appendix.
\end{proof}
\noindent
\Cref{theo:CDT} provides, in the large-dimensional regime of \Cref{ass:large_array}, characterizations of the diagonal entries of $\hat \bPhi_1,\hat \bPhi_2$ and \emph{any} product of their off-diagonal entries when their indices form a circle.
Using \Cref{theo:CDT}, we have the following result for the product of the off-diagonal entries of $\hat\bPhi = \hat \bPhi_1^{-1} \hat \bPhi_2$, again when their indices form a circle.

\begin{Lemma}[Large-dimensional characterization of $\hat\bPhi$]\label{lem:entrywise_bphi}
For $\hat\bPhi$, $\bar\bPhi$ defined in \Cref{theo:main}, and \emph{any} $m$-node cycle of indices $1\leq i_1 < \ldots < i_m \leq K$, the entries of $\hat\bPhi$, $\bar\bPhi$ satisfy
\begin{align*}
    [\hat \bPhi]_{i_1 i_2}[\hat \bPhi]_{i_2 i_3}\ldots[\hat \bPhi]_{i_m i_1} - [\bar \bPhi]_{i_1 i_2}[\bar \bPhi]_{i_2 i_3}\ldots[\bar \bPhi]_{i_m i_1} \to 0,
\end{align*}
almost surely as $N,T \to \infty$.
\end{Lemma}

\begin{proof}[Proof of \Cref{lem:entrywise_bphi}]
See \Cref{subsec:proof_of_lem:entrywise_bphi} of the appendix.
\end{proof}

With \Cref{lem:entrywise_bphi} at hand, we then have, for any $k$-by-$k$ principal minors of $\hat \bPhi$ and $\bar \bPhi$, denoted respectively as $\hat\bPhi [\mathcal{I}_k], \bar\bPhi [\mathcal{I}_k]$ that
\begin{align*}
    &| \hat\bPhi [\mathcal{I}_k] - \bar\bPhi [\mathcal{I}_k] | \\
    &= \Big| \sum_{\sigma} \sign(\sigma) \Big( \prod_{j \in \mathcal I_k} [\hat \bPhi]_{j \sigma(j)} - \prod_{j \in \mathcal I_k} [\bar \bPhi]_{j \sigma(j)} \Big) \Big| \\
    &= \sum_{\sigma} \Big| \prod_{i=1}^p \underbrace{[\hat \bPhi]_{j_1 j_2} \ldots [\hat \bPhi]_{j_{m_i} j_1}}_{\text{indices form an $m_i$-node circle}}  - \prod_{i=1}^p  [\bar \bPhi]_{j_1 j_2} \ldots [\bar \bPhi]_{j_{m_i} j_1} \Big| \\
    & \to 0,
\end{align*}
almost surely as $N,T \to \infty$.
Furthermore, the $k$-th order coefficient of the characteristic polynomial of $\hat \bPhi$ and $\bar \bPhi$ satisfies $|S_k(\hat \bPhi) - S_k(\bar \bPhi)|\to 0$.
Since the DoAs are pairwise distinct and $\bar{\bPhi}$ is assumed in \Cref{theo:main} to be nonsingular with distinct eigenvalues, Theorem 4 ensures that the eigenspectral discrepancy $\lambda_k(\hat{\bm{\Phi}})-\lambda_k(\bar{\bm{\Phi}})$ converges to zero as $N,T \to \infty$.
Hence, $\frac{\lambda_k(\hat{\bm{\Phi}})}{\lambda_k(\bar{\bm{\Phi}})} \to 1$.
Therefore, this ratio eventually remains in a neighborhood of 1 on which the argument is continuous, and thus $\Delta(\hat{\theta}_k-\bar{\theta}_k) = \arg(\frac{\lambda_k(\hat{\bm{\Phi}})}{\lambda_k(\bar{\bm{\Phi}})}) \to 0$.
Consequently, the sorted eigenvalue phases admit a one-to-one correspondence, which completes the proof of \Cref{theo:main}.

\section{Consistent DoA Estimation with G-ESPRIT}
\label{sec:G-ESPRIT}

We have seen in \Cref{theo:main} and the discussions thereafter that classical ESPRIT in \Cref{alg:TraESPRIT} is, in general, \emph{incapable} of providing consistent DoA estimates in the large array and limited snapshot regime.
In this section, we present an improved approach: the generalized ESPRIT (G-ESPRIT) method that fixes the large-dimensional inconsistency of classical ESPRIT in this regime.

The G-ESPRIT method is as simple as classical ESPRIT, but with the large-dimensional ``bias'' terms of latter consistently estimated and removed.
Precisely, it follows from \Cref{theo:main} that the top subspace $\hat \U_K$ of SCM contains additional large-dimensional bias terms (in $\bar \bPhi$) of the form $g_k = \frac{ 1- c \ell_k^{-2} }{1 + c \ell_k^{-1}}$ defined in \eqref{eq:def_g}.
These quantities, for known dimension ratio $N/T$, can be consistently estimated from the SCM per the following result.

\begin{Lemma}[Consistent estimates of $g_k$]\label{lem:estimate_g_h}
Under Assumptions~\ref{ass:large_array}~and~\ref{ass:subspace}, let $\hat \lambda_k$ be the $k$th largest eigenvalue of SCM $\hat \C$, $g_k$ be defined in \eqref{eq:def_g}, $k \in \{1, \ldots, K \}$,
Define
\begin{equation}\label{eq:def_hat_g_h}
    \hat g_k 
    = \frac{1 - \frac{N}{T}\hat\ell_k^{-2}}
           {1 + \frac{N}{T}\hat\ell_k^{-1}},
\end{equation}
where the estimated SNR $\hat\ell_k$ is defined as
\begin{equation*}
\hat\ell_k = \frac{1}{2}\Big(\frac{\hat\lambda_k}{\hat{\sigma}^2} - 1 - \frac{N}{T} + \sqrt{ \left(\frac{\hat\lambda_k}{\hat{\sigma}^2} - 1 - \frac{N}{T}\right)^2- \frac{4N}{T} } \Big),
\end{equation*}
with $\hat{\sigma}^2$ denoting the noise variance estimate $\hat\sigma^2 = \frac{1}{N-K} \sum_{i=K+1}^N \hat\lambda_i$.
Then, $\hat\ell_k - \ell_k \to 0$ and $\hat g_k - g_k \to 0$ almost surely as $N,T\to\infty$.
\end{Lemma}
\begin{proof}[Proof of \Cref{lem:estimate_g_h}]
Under Assumptions~\ref{ass:large_array} and \ref{ass:subspace}, it follows from \Cref{theo:spikes} that $\hat\lambda_k$ satisfies $\hat\lambda_k \to \sigma^2\!\big( 1 + \ell_k + c\frac{1+\ell_k}{\ell_k}\big)$ almost surely as $N,T\to\infty$.
In practice, the true noise variance $\sigma^2$ is not known and is replaced by the consistent estimator $\hat\sigma^2 = \frac{1}{N-K} \sum_{i=K+1}^N \hat\lambda_i$, which satisfies $\hat\sigma^2\!\to\!\sigma^2$ almost surely.
Inverting the expression and using continuous mapping theorem, we have that $\hat\ell_k - \ell_k \to 0$.
Finally, by the definition of $g_k$ and continuous mapping theorem, we obtain $\hat g_k - g_k \to 0$, which concludes the proof.
\end{proof}

\noindent
\Cref{lem:estimate_g_h} provides consistent estimates of the bias terms in classical ESPRIT.
These estimates, combining with \Cref{theo:main}, lead to the G-ESPRIT DoA estimation procedure summarized in \Cref{alg:G-ESPRIT}, with the following guarantee.

\begin{Proposition}[Consistent DoA estimation with G-ESPRIT]\label{prop:gesprit}
Assume that Assumptions~\ref{ass:large_array}~and~\ref{ass:subspace} hold.
In addition, suppose that either \Cref{ass:widely_spaced}~or~\Cref{ass:closely_spaced} is satisfied.
Denote $\theta_k$ the true DoAs, we have, for $k \in \{ 1, \ldots, K \}$ and $ \hat{\theta}_k^G \equiv \arg(\lambda_k(\hat{\bPhi}^G))/\Delta$ that
    \begin{equation}
        \Delta( \hat{\theta}_k^G  - \theta_k) \to 0,
    \end{equation}
    almost surely as $N,T \to \infty$, with $\lambda_k(\hat{\bPhi}^G)$ the $k$th largest eigenvalue of $\hat{{\bPhi}}^G = (\hat{{\bPhi}}^G_1)^{-1}\hat{\bPhi}_2^G$ with
    \begin{equation}\label{eq:def_Tilde_Phi}
    \begin{split}
        \hat{ \bPhi}_1^G &= \diag(\hat{\mathbf{g}}^{-1/2}) \left( \hat{\bPhi}_1 - \tau \I_K \right) \diag(\hat{\mathbf{g}}^{-1/2})+\tau \I_K , \\
        \hat{ \bPhi}_2^G & = \diag(\hat{\mathbf{g}}^{-1/2}) \hat{\bPhi}_2 \diag(\hat{\mathbf{g}}^{-1/2}),
    \end{split}
    \end{equation}
    for $\hat{\mathbf{g}}^{-1/2} =[1/\sqrt{\hat g_1}, \ldots, 1/\sqrt{\hat g_K}]^\T$ and $\hat g_k$ as defined in \eqref{eq:def_hat_g_h}.
\end{Proposition}
\noindent
The proof of \Cref{prop:gesprit} follows the same line of argument as that of \Cref{theo:main}.
In particular, the almost sure convergence of $\hat \ell_k$ and $\hat g_k$ in \Cref{lem:estimate_g_h}, together with the arguments in Sections~\ref{subsec:proof_framework}~and~\ref{subsec:proof_of_main}, implies $\lambda_k(\hat{\bm{\Phi}}^G) \to \lambda_k(\bm{\Phi}) $almost surely. Since $\lambda_k(\bPhi)\neq 0$, we have $\lambda_k(\hat{\bm{\Phi}}^G)/\lambda_k(\bm{\Phi}) \to 1$, and hence $\Delta( \hat{\theta}_k^G  - \theta_k)=\arg(\lambda_k(\hat{\bm{\Phi}}^G)/\lambda_k(\bm{\Phi}) ) \to 0$ almost surely. 
In practical scenarios where the number of sources $K$ is not known a priori, the proposed G-ESPRIT estimator can be seamlessly combined with standard source enumeration techniques, such AIC, MDL~\cite{1164557}, PDL~\cite{valaee2004information}, and RMT-based estimation methods~\cite{baik2006eigenvalues}, as a pre-processing step.

A few remarks and discussions on \Cref{prop:gesprit} are in order.
\begin{Remark}[Special case: G-ESPRIT for widely-spaced DoAs]\normalfont
For widely-spaced DoAs in \Cref{ass:widely_spaced}, it follows from \Cref{rem:wide_correlated} and \Cref{lem:AJJA} in \Cref{sec:techLemmas} of the appendix that $\hat \bPhi_1$ is approximately diagonal with \emph{real} diagonal entries, and it thus suffices to ``de-bias'' $\hat \bPhi_2$ as
\begin{equation}
    \hat \bPhi_2^G = \diag(\hat{\mathbf{g}}^{-1/2}) \hat{\bPhi}_2 \diag(\hat{\mathbf{g}}^{-1/2}),
\end{equation}
and that both $\lambda_k( \hat {\bPhi}_1^{-1} \hat{\bPhi}_2^G)/\Delta$ and $\lambda_k(\hat{\bPhi}_2^G)/\Delta$ alone in fact yield consistent DoA estimation.
\end{Remark}

\begin{algorithm}[t!]
    \renewcommand{\algorithmicrequire}{\textbf{Input:}}
    \renewcommand{\algorithmicensure}{\textbf{Output:}}
    \caption{The proposed G-ESPRIT DoA estimation.}
    \label{alg:G-ESPRIT}    
    \begin{algorithmic}[1]
        \REQUIRE Received signal $\X \in \CC^{N \times T}$, number of sources $K$.
        \ENSURE Estimated DoA angles $\hat\theta_k^G, k \in \{ 1,\ldots, K\}$.
        \STATE Compute the SCM $\hat \C = \X \X^\H/T$ to retrieve $\hat \U_K = [\hat \uu_1, \ldots, \hat \uu_K] \in \CC^{N \times K}$ the estimated signal subspace composed of the top-$K$ eigenvectors $\hat \uu_1, \ldots, \hat \uu_K \in \CC^N$ associated to the largest $K$ eigenvalues of $\hat \C$;

    \STATE Define two selection matrices $\J_1, \J_2 \in \RR^{n \times N}$ as in \eqref{eq:def_J} that both select $n$ among $N$ rows with distance $\Delta \geq 1$;

    \STATE Compute $\hat{\mathbf{\Phi}}_1, \hat{\mathbf{\Phi}}_2$ using $\hat \U_K$ and $\J_1,\J_2$ as in \eqref{eq:def_hat_Phi};

    \STATE Compute $\hat{\bPhi}^G$ as in \Cref{prop:gesprit} and then the \emph{angles} of $\lambda_k(\hat{\mathbf{\Phi}}^G)$, the $k$th complex eigenvalue of $\hat{\mathbf{\Phi}}^G$;

    \STATE \textbf{return} $\hat\theta_k^G = \arg(\lambda_k(\hat{\mathbf{\Phi}}^G))/\Delta, k \in \{ 1,\ldots, K\}$.
\end{algorithmic}  
\end{algorithm}

\Cref{prop:gesprit} tells us that the G-ESPRIT approach in \Cref{alg:G-ESPRIT} provides consistent DoA estimates in the large array and limited snapshot regime while the classical ESPRIT in \Cref{alg:TraESPRIT} cannot.
However, \Cref{prop:gesprit} alone provides theoretical guarantees for G-ESPRIT in the widely-spaced DoA (\Cref{ass:widely_spaced}) but \emph{not} in the closely-spaced DoA (\Cref{ass:closely_spaced}) setting.
In the latter case, one needs $N$-consistent estimator to separate closely-spaced DoAs, see our discussion in \Cref{sec:inconsist} above.
This $N$-consistency can be shown under an additional assumption on the subarray ``distance'' $\Delta$ as follow.

\begin{Corollary}[$N$-consistency of G-ESPRIT]\label{coro:N_consistency}
Under the notations and settings of \Cref{prop:gesprit}, assume in addition that the subarray distance $\Delta$ satisfies that $\Delta/N $ converges in $(0,1)$, then, the G-ESPRIT estimates $\hat{\theta}_k^G$ are $N$-consistent, that is
\begin{equation}
    N (\hat{\theta}_k^G - \theta_k ) \to 0, \quad k \in \{1, \ldots , K \},
\end{equation}
almost surely as $N,T \to \infty$.
\end{Corollary}

\noindent
\Cref{coro:N_consistency} is a direct consequence of \Cref{prop:gesprit} and the assumption that $\Delta$ is of order $N$.\footnote{As we shall see below in \Cref{sec:simu}, taking a large subarray distance $\Delta$ of order $N$ leads to small MSEs that are empirically close to the theoretically optimal Cramér--Rao Bound. 
The theoretical investigation of this observation, however, needs second-order analysis (of the fluctuation) of ESPRIT and/or G-ESPRIT estimators and is beyond the scope of this paper.}
It shows that G-ESPRIT in fact provides $N$-consistent DoA estimations in both widely-~and closely-spaced DoA scenarios, with estimation errors of order $o(N^{-1})$.
Typically, the estimation error is of order $O(N^{-3/2})$, as per the following remark.

\begin{Remark}[Precise estimation error for two DoAs]\normalfont \label{rem:N-con_for_two}
In the case of $K = 2$ DoAs, under the same notations and settings of \Cref{coro:N_consistency}, one has that
\begin{equation}
    \hat{\theta}_k^G - \theta_k = O(N^{-3/2}), \quad k \in \{ 1, 2 \}.
\end{equation}
See \Cref{subsec:proof_of_rem:N-con_for_two} for a proof of this result.
For $K>2$, it follows from \Cref{theo:continuity} that one has instead $\lambda(\hat{\bPhi}^G) - \lambda(\bPhi) = O(N^{-1/(2K)})$ and $\hat{\theta}_k^G - \theta_k = O(N^{-1 - 1/(2K)})$.
\end{Remark}

\section{Numerical Simulations}
\label{sec:simu}

In this section, we provide numerical simulations to validate our theoretical analysis of classical ESPRIT in \Cref{sec:inconsist} and the proposed G-ESPRIT method in \Cref{sec:G-ESPRIT}.
Precisely, in \Cref{subsec:validation} we provide simulations on not-so-large $N,T$, to validate our asymptotic analysis (as $N,T \to \infty$ together) in Sections~\ref{sec:inconsist}~and~\ref{sec:G-ESPRIT} for finite dimensional signals. 
Then, in \Cref{subsec:comparision}, we compare the proposed G-ESPRIT approach to other DoA estimation methods such as MUSIC and G-MUSIC~\cite{Vallet2012ACF}, as well as to the (theoretically optimal) Cramér--Rao Bound (CRB).
Code to reproduce the results in this section is available at \url{https://github.com/zhengyuwang0/GESPRIT}.

\subsection{On the choice of subarray size $n$ and distance $\Delta$}
The performance of ESPRIT-type algorithms depends critically on the subarray distance~$\Delta$ and the subarray size $n$.
To ensure unambiguous estimation, $\Delta$ must satisfy the \emph{angular unambiguity condition} $-\pi \leq \Delta \theta \leq \pi$, so that the phase difference between the two subarrays remains within the 
principal interval. 
Accordingly, in the simulations performed in this section, we choose,
\begin{enumerate}
 \item in the case of widely-spaced DoAs and without prior information about the DoAs, $\Delta=1$ to avoid this phase ambiguity, and $n = N-1 $; and
 \item in the case of closely-spaced DoAs with a narrow field-of-view $\theta \in [\theta_{\min},\theta_{\max}]$, $\Delta$ is chosen close to, but strictly smaller than $\lfloor \frac{\pi}{\max_{\theta\in[\theta_{\min},\theta_{\max}]} |\theta|} \rfloor $, so as to avoid phase ambiguity while maintaining high angular resolution. 
\end{enumerate} 
When widely-~and closely-spaced DoAs coexist, no single displacement $\Delta$ is optimal for all sources. 
A practical solution is to use multiple shift-invariant subarray pairs with different spacings and to fuse the resulting estimates, see also~\cite{765149}.

\begin{figure}
    \hspace{-10pt}
    \begin{minipage}[t]{0.545\columnwidth}
    \centering
    \begin{tikzpicture}
      \renewcommand{\axisdefaulttryminticks}{4} 
      \pgfplotsset{every major grid/.style={densely dashed}}       
      \tikzstyle{every axis y label}+=[yshift=-10pt] 
      \tikzstyle{every axis x label}+=[yshift=5pt]
      \pgfplotsset{
        every axis legend/.append style={
          cells={anchor=west},
          fill=white,
          at={(1,0.75)},
          anchor=north east,
          font=\footnotesize
        }
      }
      \begin{axis}[
        width=\columnwidth,
        height=0.9\columnwidth,
        ymajorgrids=true,
        grid=both,
        xmin = 4,
        xmax = 198,
        ymax = 130,
        ymin = 2e-2,
        xtick={20,60,100,140,180},
        xlabel = {$n$},
        ylabel = {MSE (in deg$^2$)},
        ymode=log,
        ylabel style={yshift=-5pt},
      ]
        \addplot[smooth, BLUE!60!white, line width=1pt , mark=o]
          table[x=n, y=ges] {./results/fig1_n.txt};
        \addlegendentry{G-ESPRIT};
        \addplot[smooth, RED!60!white, line width=1pt , mark=o]
          table[x=n, y=es] {./results/fig1_n.txt};
        \addlegendentry{ESPRIT};
      \end{axis}
      \node at (current bounding box.south) [yshift=-5pt,xshift=17pt] {\small (a)};
    \end{tikzpicture}
\end{minipage}
\hspace{-20pt}
\begin{minipage}[t]{0.545\columnwidth}
\centering
\begin{tikzpicture}
      \renewcommand{\axisdefaulttryminticks}{4} 
      \pgfplotsset{every major grid/.style={densely dashed}}       
      \tikzstyle{every axis y label}+=[yshift=-10pt] 
      \tikzstyle{every axis x label}+=[yshift=5pt]
      \pgfplotsset{
        every axis legend/.append style={
          cells={anchor=west},
          fill=white,
          at={(1,0.75)},
          anchor=north east,
          font=\tiny
        }
      }
      \begin{axis}[
        width=\columnwidth,
        height=0.9\columnwidth,
        ymajorgrids=true,
        grid=both,
        xmin = 1,
        xmax = 200,
        ymax = 0.35,
        ymin = 5e-4,
        xtick={20,60,100,140,180},
        xlabel = {$\Delta$},
        ylabel = \empty,
        ymode=log,
      ]
        \addplot[smooth, BLUE!60!white, line width=1pt , mark=o]
          table[x=Delta, y=ges] {./results/fig1_tradeoff.txt};
        \addlegendentry{G-ESPRIT};
        \addplot[smooth, RED!60!white, line width=1pt , mark=o]
          table[x=Delta, y=es] {./results/fig1_tradeoff.txt};
        \addlegendentry{ESPRIT};
        \addplot[dashed, color=gray, line width=1.5pt]
          coordinates {(62,1e-8) (62,1)};
        \addlegendentry{\!$\Delta \!=\! \lfloor \!\frac{\pi}{|\theta_3|}\! \rfloor $};
      \end{axis}
      \node at (current bounding box.south) [yshift=-5pt,xshift=16pt] {\small (b)};
    \end{tikzpicture}
\end{minipage} 
\vspace{-16pt}
\caption{
\textbf{(a):} MSE of G-ESPRIT as a function of $n$ with $\Delta=1$. The three widely-spaced sources are at DoA $\theta_1 = 0$, $\theta_2 = \frac{\pi}{4}$,  $\theta_3 = \frac{2\pi}{3}$, and the power matrix is $\mathbf{P} = 2\I_3$. 
\textbf{(b):} MSE of G-ESPRIT as a function of distance $\Delta$ with $n = N-\Delta$. 
The three closely-spaced sources are at DoA $\theta_1 = 0$, $\theta_2 = 0.8\times 2\pi/N$, $\theta_3 = 1.6\times 2\pi/N$, and the power matrix is $\mathbf{P} = 2\I_3$. Results are obtained by
averaging over 500 independent trials.} 
 \label{fig:delta-n}
 \vspace{-6pt}
\end{figure}

\begin{figure}[!tb]
\centering
\begin{minipage}[t]{\columnwidth}
\begin{minipage}[t]{0.495\linewidth}
\centering
\begin{tikzpicture}
    \renewcommand{\axisdefaulttryminticks}{4}
    \pgfplotsset{every major grid/.style={densely dashed}}
    \tikzstyle{every axis y label}+=[yshift=-10pt]
    \tikzstyle{every axis x label}+=[yshift=5pt]
    \pgfplotsset{every axis legend/.append style={cells={anchor=east}, fill=white, at={(1,1)}, anchor=north east, font=\footnotesize}}

    \begin{axis}[
        width=1.1\columnwidth,
        height=1\columnwidth,
        ymajorgrids=true,
        xlabel = { $N$},
        ylabel = {Spectral norm errors},
        ylabel style={yshift=5pt},
        scaled ticks=true,
    ]
      \addplot[only marks,mark=o,color=BLUE!60!white, 
        error bars/.cd, y dir=both, y explicit, error bar style={color=gray}] table[x index=0, y index=1, y error index=2] {./results/spnorm.txt};
      \addlegendentry{$\|\hat{\bPhi} - \bar{\bPhi}\|$};

      \addplot [RED!60!white, line width=1.5pt, domain=50:400]{0.467*x^(-0.61)};
      \addlegendentry{$0.47 N^{-0.61}$};
    \end{axis}
\end{tikzpicture}
\end{minipage}
\begin{minipage}[t]{0.495\linewidth}
\hspace{3pt}
\centering
\begin{tikzpicture}
    \renewcommand{\axisdefaulttryminticks}{4}
    \pgfplotsset{every major grid/.style={densely dashed}}
    \tikzstyle{every axis y label}+=[yshift=-10pt]
    \tikzstyle{every axis x label}+=[yshift=5pt]
    \pgfplotsset{every axis legend/.append style={cells={anchor=east}, fill=white, at={(1,1)}, anchor=north east, font=\footnotesize}}

    \begin{axis}[
        width=1.1\columnwidth,
        height=1\columnwidth,
        xlabel={ $N$},
        ymajorgrids=true,
        scaled ticks=true,
    ]
        \addplot[only marks,mark=o,color=BLUE!60!white, 
        error bars/.cd, y dir=both, y explicit, error bar style={color=gray}] table[x index=0, y index=3, y error index=4] {./results/spnorm.txt};
        \addlegendentry{$\|\hat{\bPhi}^G - \bPhi\|$}

        \addplot[RED!60!white, line width=1.5pt, domain=50:400]{0.593*x^(-0.636)};
        \addlegendentry{$0.59 N^{-0.63}$}
    \end{axis}
\end{tikzpicture}
\vspace{1mm}
\end{minipage}
\vspace{3mm}
\begin{minipage}{\textwidth}
    (a) widely-spaced DoAs with correlated sources at DoA $\theta_1 = 0$ and $\theta_2 = \pi/4$, power matrix $\P = (\begin{smallmatrix}
    3 & 1.2 \\
    1.2 & 3
\end{smallmatrix})$, subarray size $n=N-1$, and distance $\Delta=1$.
\end{minipage}
\end{minipage}
\vspace{1mm}
\begin{minipage}[t]{\columnwidth}
\centering
\begin{minipage}[t]{0.495\linewidth}
\centering
\begin{tikzpicture}
    \renewcommand{\axisdefaulttryminticks}{4}
    \pgfplotsset{every major grid/.style={densely dashed}}
    \tikzstyle{every axis y label}+=[yshift=-10pt]
    \tikzstyle{every axis x label}+=[yshift=5pt]
    \pgfplotsset{every axis legend/.append style={cells={anchor=east}, fill=white, at={(1,1)}, anchor=north east, font=\footnotesize}}

    \begin{axis}[
        width=1.1\columnwidth,
        height=1\columnwidth,
        ymajorgrids=true,
        xlabel = {$N$},
        ylabel = {Spectral norm errors},
        scaled ticks=true,
    ]
      \addplot[only marks,mark=o,color=BLUE!60!white, 
        error bars/.cd, y dir=both, y explicit, error bar style={color=gray}] table[x index=0, y index=5, y error index=6] {./results/spnorm.txt};
      \addlegendentry{$\|\hat{\bPhi} - \bar{\bPhi}\|$};

      \addplot [RED!60!white, line width=1.5pt, domain=50:400]{2.26*x^(-0.65)};
      \addlegendentry{$2.26 N^{-0.65}$};
    \end{axis}
\end{tikzpicture}
\end{minipage}
\hspace{-5pt}
\begin{minipage}[t]{0.495\linewidth}
\centering
\begin{tikzpicture}
    \renewcommand{\axisdefaulttryminticks}{4}
    \pgfplotsset{every major grid/.style={densely dashed}}
    \tikzstyle{every axis y label}+=[yshift=-10pt]
    \tikzstyle{every axis x label}+=[yshift=5pt]
    \pgfplotsset{every axis legend/.append style={cells={anchor=east}, fill=white, at={(1,1)}, anchor=north east, font=\footnotesize}}

    \begin{axis}[
        width=1.1\columnwidth,
        height=1\columnwidth,
        xlabel={$N$},
        ymajorgrids=true,
        scaled ticks=true,
    ]
        \addplot[only marks,mark=o,color=BLUE!60!white, 
        error bars/.cd, y dir=both, y explicit, error bar style={color=gray}] table[x index=0, y index=7, y error index=8] {./results/spnorm.txt};
        \addlegendentry{$\|\hat{\bPhi}^G - \bPhi\|$}

        \addplot[RED!60!white, line width=1.5pt, domain=50:400]{1.474*x^(-0.604)};
        \addlegendentry{$1.47 N^{-0.6}$}
    \end{axis}
\end{tikzpicture}
\vspace{0.5mm}
\end{minipage}
\begin{minipage}{\textwidth}
    
(b) closely-spaced DoAs ($\theta_1 = 0$ and $\theta_2 = 0.8\times 2 \pi/N$) having equal powers with $\P =2\I_2$, $\Delta=\lfloor \pi/|\theta_2| \rfloor $, and $n=N-\Delta$.
\end{minipage}
\end{minipage}
\caption{ 
Approximation errors in spectral norm versus array length $N$, for $T = 2N$ and $\sigma^2 = 1$. 
Simulation results in \textbf{\BLUE blue} are averaged over $2\,000$ trials, with polynomial fit in \textbf{\RED red}.}
\label{fig:spectral_error}
\vspace{-10pt}
\end{figure}

In the left plot of Fig.~\ref{fig:delta-n}, we investigate the MSE of ESPRIT and G-ESPRIT as a function of the subarray size $n$ in the widely-spaced DoA settings, with $\Delta=1$, $N = 200$ and $T = 400$.
In the subsequent experiments in this section, the source signals are generated according to $\mathbf{s}(t)=\mathbf{P}^{\frac1 2} \mathbf{w}(t)$ where $\mathbf{w}(t) \sim \mathcal{CN}(\mathbf{0}_K,\mathbf{I}_K)$, such that the source covariance matrix satisfies $\EE [\mathbf{s}(t) \mathbf{s}(t)^\H]=\mathbf{P}$.
The mean squared error (MSE) is defined as
\begin{align*}
    \mathrm{MSE}=\frac{1}{K N_{MC}}\sum_{i=1}^{N_{MC}}\sum_{k=1}^{K}(\hat{\phi}_{k,i}-\phi_k)^2,
\end{align*}
where $N_{MC}$ is the number of Monte Carlo trials, $\hat{\phi}_{k,i}$ and $\phi_k$ denote the sorted physical angles of the estimated DoAs and the true DoAs, respectively.
The results show that the estimator achieves lower MSE as $n$ increases, and attains its minimum at $n=N-1$.
In the right plot of Fig.~\ref{fig:delta-n}, we illustrate the MSE as a function of $\Delta$ in the closely-spaced DoA settings, with $n = N-\Delta$.
It can be observed that when $\Delta$ is small, the resulting phase difference $\Delta\theta$ is insufficient to provide adequate angular sensitivity, resulting in high MSE.
As $\Delta$ increases, the MSE decreases and reaches its minimum near the largest unambiguous spacing $\Delta_{\max}=\lfloor \pi/|\theta_3| \rfloor $ (marked by the dashed line).
When $\Delta$ exceeds $\Delta_{\max}$, phase wrapping occurs and the MSE increases sharply.
The experimental results in Fig.~\ref{fig:delta-n} validate the proposed $(n,\Delta)$ selection strategy described above.

\subsection{Validation of asymptotic theoretical results}
\label{subsec:validation}

We start by showing that the proposed asymptotic analyses in \Cref{theo:main} and \Cref{prop:gesprit} provide accurate predictions on finite dimensional signals.

Fig.~\ref{fig:spectral_error} compares the approximation errors $\| \hat{\bPhi} - \bar{\bPhi} \|$ and $\| \hat{\bPhi}^G - \bPhi \|$, corresponding to the classical ESPRIT analyzed in \Cref{theo:main} and the proposed G-ESPRIT in \Cref{prop:gesprit}, respectively.
We see, in both widely-spaced (\Cref{ass:widely_spaced}) and closely-spaced (\Cref{ass:closely_spaced}) scenarios, that as the array length $N$ grows, the spectral norm errors empirically decay at a rate of $O(N^{-1/2})$, as in line with \Cref{rem:N-con_for_two} above.

\begin{figure*}[!tp]
    \centering
    \hspace{12pt} 
    \begin{minipage}[t]{0.65\columnwidth}
    \centering
    \begin{tikzpicture}
        \def \deg2rady{180/3.1415926}
        \centering
        \renewcommand{\axisdefaulttryminticks}{4} 
        \pgfplotsset{every major grid/.style={densely dashed}}       
        \tikzstyle{every axis x label}+=[yshift=10pt]
        \pgfplotsset{every axis legend/.append style={cells={anchor=west},fill=white, at={(1,0.74),sacel=0.2}, anchor=north east, font=\footnotesize,legend columns=2 }} 
        \begin{axis}[
        ytick=\empty,
        ytick style={draw=none},
        width=1\columnwidth,
        height=0.9\columnwidth,
        ymajorgrids=true,
        scaled ticks=true,
        grid=both,
        ytick={0, 3.1415926/8, 3.1415926/4},
        yticklabels={$0$, $\frac{\pi}{8}$, $\frac{\pi}{4}$},
        xlabel = { $N$ },
        xlabel style={yshift=3pt},
        ylabel = { $\theta$ (in deg) },
        xmin=20,xmax=200,ymin=-0.1,ymax=0.9,
        scaled ticks=true]
        \addplot[smooth, RED!60!white, line width=1pt ,mark=o] table[x=N, y=ES1] {./results/fig2_wide_left.txt};
        \addlegendentry{$\hat{\theta}_1$};
        \addplot[smooth, RED!60!white, line width=1pt ,mark=o] table[x=N, y=ES2] {./results/fig2_wide_left.txt};
        \addlegendentry{$\hat{\theta}_2$};
        \addplot[smooth, BLUE!60!white, line width=1pt ,mark=square] table[x=N, y=GES1] {./results/fig2_wide_left.txt};
        \addlegendentry{$\hat{\theta}_1^G$};
        \addplot[smooth, BLUE!60!white, line width=1pt ,mark=square] table[x=N, y=GES2] {./results/fig2_wide_left.txt};
        \addlegendentry{$\hat{\theta}_2^G$};
        \addplot[densely dashed, gray!60!black, line width=1pt ] table[x=N, y=True1] {./results/fig2_wide_left.txt};
        \addlegendentry{$\theta_1$};
        \addplot[densely dashed, gray!60!black, line width=1pt ] table[x=N, y=True2] {./results/fig2_wide_left.txt};
        \addlegendentry{$\theta_2$};
        \addplot[smooth, GREEN!80!white, line width=1pt ,mark=x] table[x=N, y=bar1] {./results/fig2_wide_left.txt};
        \addlegendentry{$\bar{\theta}_1$};
        \addplot[smooth, GREEN!80!white, line width=1pt ,mark=x] table[x=N, y=bar2] {./results/fig2_wide_left.txt};
        \addlegendentry{$\bar{\theta}_2$};
    \end{axis}
    \node at (current bounding box.south) [yshift=-5pt,xshift=12pt] {\small (a)};
\end{tikzpicture}
\vspace{-17pt} 
\end{minipage}
\hspace{-18pt} 
\vspace{-1pt}
    \begin{minipage}[t]{0.65\columnwidth}
    \hspace{2pt}
    \begin{tikzpicture}
    \def \deg2rad2{180/3.1415926}
    \centering
    \renewcommand{\axisdefaulttryminticks}{4} 
    \pgfplotsset{every major grid/.style={densely dashed}}       
    \tikzstyle{every axis x label}+=[yshift=5pt]
    \pgfplotsset{
        every axis legend/.append style={cells={anchor=west},fill=white,at={(0.62,1)},anchor=north east,font=\tiny,legend columns=2}} 
    \begin{axis}[
        name=main,              
        ytick=\empty,
        ytick style={draw=none},
        width=1\columnwidth,
        height=0.9\columnwidth,
        ymajorgrids=true,
        grid=both,
        ytick={0, 0.02*3.1415926, 0.07*3.1415926},
        yticklabels={$0$, $\frac{\pi}{50}$, $\frac{7\pi}{100}$},
        xlabel = { $N$ },
        xlabel style={yshift=3pt},
        ylabel = \empty,
        xmin=20,xmax=200,
        ymin=-0.05,ymax=1.6*3/20, 
    ]
        \addplot[smooth, RED!60!white, line width=1pt ,mark=o] table[x=N, y=ES1] {./results/fig2_near_left.txt};
        \addlegendentry{$\hat{\theta}_1$};
        \addplot[smooth, RED!60!white, line width=1pt ,mark=o] table[x=N, y=ES2] {./results/fig2_near_left.txt};
        \addlegendentry{$\hat{\theta}_2$};
        \addplot[smooth, BLUE!60!white, line width=1pt ,mark=square] table[x=N, y=GES1] {./results/fig2_near_left.txt};
        \addlegendentry{$\hat{\theta}_1^G$};
        \addplot[smooth, BLUE!60!white, line width=1pt ,mark=square] table[x=N, y=GES2] {./results/fig2_near_left.txt};
        \addlegendentry{$\hat{\theta}_2^G$};
        \addplot[densely dashed, gray!60!black, line width=1pt ] table[x=N, y=True1] {./results/fig2_near_left.txt};
        \addlegendentry{$\theta_1$};
        \addplot[densely dashed, gray!60!black, line width=1pt ] table[x=N, y=True2] {./results/fig2_near_left.txt};
        \addlegendentry{$\theta_2$};
        \addplot[smooth, GREEN!80!white, line width=1pt ,mark=x] table[x=N, y=bar1] {./results/fig2_near_left.txt};
        \addlegendentry{$\bar{\theta}_1$};
        \addplot[smooth, GREEN!80!white, line width=1pt ,mark=x] table[x=N, y=bar2] {./results/fig2_near_left.txt};
        \addlegendentry{$\bar{\theta}_2$};

        \coordinate (zoomsw) at (axis cs:130,-0.015);
        \coordinate (zoomne) at (axis cs:180,0.045);
        \draw[black, thick] (zoomsw) rectangle (zoomne);
        \coordinate (zoomanchor) at ($(zoomne)!0.5!(zoomsw |- zoomne)$);
    \end{axis}

    \begin{axis}[
        name=inset,
        at={(rel axis cs:0.62,0.53)}, 
        anchor=south west,
        width=0.55\columnwidth,
        height=0.5\columnwidth,
        xtick=\empty, 
        ytick=\empty,
        xticklabels=\empty, 
        yticklabels=\empty,
        xmin=130, xmax=180,
        ymin=-0.01, ymax=0.05,
        axis line style={black},
    ]
        \addplot[smooth, RED!60!white, line width=0.8pt ,mark=o] table[x=N, y=ES1] {./results/fig2_near_left.txt};
        \addplot[smooth, RED!60!white, line width=0.8pt ,mark=o] table[x=N, y=ES2] {./results/fig2_near_left.txt};
        \addplot[smooth, BLUE!60!white, line width=0.8pt ,mark=square] table[x=N, y=GES1] {./results/fig2_near_left.txt};
        \addplot[smooth, BLUE!60!white, line width=0.8pt ,mark=square] table[x=N, y=GES2] {./results/fig2_near_left.txt};
        \addplot[densely dashed, gray!60!black, line width=0.8pt ] table[x=N, y=True1] {./results/fig2_near_left.txt};
        \addplot[densely dashed, gray!60!black, line width=0.8pt ] table[x=N, y=True2] {./results/fig2_near_left.txt};
        \addplot[smooth, GREEN!80!white, line width=0.8pt ,mark=x] table[x=N, y=bar1] {./results/fig2_near_left.txt};
        \addplot[smooth, GREEN!80!white, line width=0.8pt ,mark=x] table[x=N, y=bar2] {./results/fig2_near_left.txt};
    \end{axis}
    \draw[black, thick] (inset.south west) rectangle (inset.north east);
    \draw[black, thin] (zoomanchor) -- ($(inset.south west)!0.5!(inset.south east)$);
    \node at (current bounding box.south) [yshift=-5pt,xshift=10pt] {\small (b)};

\end{tikzpicture}
\vspace{-17pt} 
\end{minipage}
\vspace{-1pt}
\hspace{-8pt} 
    \begin{minipage}[t]{0.65\columnwidth}
    \centering
    \begin{tikzpicture}
        \def \deg2rady{180/3.1415926}
        \centering
        \renewcommand{\axisdefaulttryminticks}{4} 
        \pgfplotsset{every major grid/.style={densely dashed}}       
        \tikzstyle{every axis x label}+=[yshift=5pt]
        \pgfplotsset{every axis legend/.append style={cells={anchor=west},fill=white, at={(1,0.74),sacel=0.2}, anchor=north east, font=\footnotesize,legend columns=2 }} 
        \begin{axis}[
        ytick=\empty,
        ytick style={draw=none},
        width=1\columnwidth,
        height=0.9\columnwidth,
        ymajorgrids=true,
        scaled ticks=true,
        grid=both,
        ytick={0, 3.1415926/8, 3.1415926/4},
        yticklabels={$0$, $\frac{\pi}{8}$, $\frac{\pi}{4}$},
        xlabel = { $N$ },
        xlabel style={yshift=3pt},
        ylabel = \empty,
        xmin=20,xmax=200,ymin=-0.1,ymax=0.9,
        scaled ticks=true]
        \addplot[smooth, RED!60!white, line width=1pt ,mark=o] table[x=N, y=ES1] {./results/fig2_unco_left.txt};
        \addlegendentry{$\hat{\theta}_1$};
        \addplot[smooth, RED!60!white, line width=1pt ,mark=o] table[x=N, y=ES2] {./results/fig2_unco_left.txt};
        \addlegendentry{$\hat{\theta}_2$};
        \addplot[smooth, BLUE!60!white, line width=1pt ,mark=square] table[x=N, y=GES1] {./results/fig2_unco_left.txt};
        \addlegendentry{$\hat{\theta}_1^G$};
        \addplot[smooth, BLUE!60!white, line width=1pt ,mark=square] table[x=N, y=GES2] {./results/fig2_unco_left.txt};
        \addlegendentry{$\hat{\theta}_2^G$};
        \addplot[densely dashed, gray!60!black, line width=1pt ] table[x=N, y=True1] {./results/fig2_unco_left.txt};
        \addlegendentry{$\theta_1$};
        \addplot[densely dashed, gray!60!black, line width=1pt ] table[x=N, y=True2] {./results/fig2_unco_left.txt};
        \addlegendentry{$\theta_2$};
        \addplot[smooth, GREEN!80!white, line width=1pt ,mark=x] table[x=N, y=bar1] {./results/fig2_unco_left.txt};
        \addlegendentry{$\bar{\theta}_1$};
        \addplot[smooth, GREEN!80!white, line width=1pt ,mark=x] table[x=N, y=bar2] {./results/fig2_unco_left.txt};
        \addlegendentry{$\bar{\theta}_2$};
    \end{axis}
    \node at (current bounding box.south) [yshift=-5pt,xshift=10pt] {\small (c)};
\end{tikzpicture}
\vspace{-17pt} 
\end{minipage}
\vspace{-1pt} 
\begin{minipage}[t]{0.65\columnwidth}
\centering
\begin{tikzpicture}
    \centering
    \renewcommand{\axisdefaulttryminticks}{4} 
    \pgfplotsset{every major grid/.style={densely dashed}}    
    \tikzstyle{every axis x label}+=[yshift=5pt]
    \pgfplotsset{every axis legend/.append style={cells={anchor=west},fill=white, at={(1,1)}, anchor=north east, font=\footnotesize,legend columns=2,}} %
    \begin{axis}[
        width=1\columnwidth,
        height=0.9\columnwidth,
        ymajorgrids=true,
        scaled ticks=true,
        xtick={100,200,300,400},
        xmin=50,xmax=400,
        xlabel = { $N$ },
        xlabel style={yshift=3pt},
        ylabel = { MSE (in deg$^2$) },
        scaled ticks=true,
        ymode=log,
        grid = both,
        ]
        \addplot[smooth, BLUE!60!white, line width=1pt ,mark=o] table[x=N, y=Emse] {./results/rerunFig4d.txt};
        \addlegendentry{$\text{MSE}[\hat{\theta}]$};
        \addplot[smooth, BLUE!60!white, line width=1pt ,mark=x] table[x=N, y=Evar] {./results/rerunFig4d.txt};
        \addlegendentry{ $\text{Var}[\hat{\theta}]$};
        \addplot[smooth, RED!60!white, line width=1pt ,mark=o] table[x=N, y=GEmse] {./results/rerunFig4d.txt};
        \addlegendentry{$\text{MSE}[\hat{\theta}^G]$};
        \addplot[smooth, RED!60!white, line width=1pt ,mark=x] table[x=N, y=GEvar] {./results/rerunFig4d.txt};
        \addlegendentry{ $\text{Var}[\hat{\theta}^G]$};
    \end{axis}
        \node at (current bounding box.south) [yshift=-5pt,xshift=14pt] {\small (d)};
\end{tikzpicture}
\end{minipage}
\vspace{-2pt}
\hspace{-6pt}
\begin{minipage}[t]{0.65\columnwidth}
\centering
\begin{tikzpicture}
    \centering
    \renewcommand{\axisdefaulttryminticks}{4} 
    \pgfplotsset{every major grid/.style={densely dashed}}    
    \tikzstyle{every axis x label}+=[yshift=5pt]
    \pgfplotsset{every axis legend/.append style={
        cells={anchor=west},
        fill=white,
        at={(1,1)}, anchor=north east,
        font=\footnotesize,legend columns=2,}} 
    \begin{axis}[
        width=1\columnwidth,
        height=0.9\columnwidth,
        ymajorgrids=true,
        scaled ticks=true,
        xtick={100,200,300,400},
        xmin=50,xmax=400,
        xlabel = { $N$ },
        xlabel style={yshift=3pt},
        ylabel = \empty,
        ymode=log,
        grid = both,
        ]
        \addplot[smooth, BLUE!60!white, line width=1pt ,mark=o] table[x=N, y expr=\thisrow{Emse}*(180/pi)^2] {./results/fig2_near_right.txt};
        \addlegendentry{$\text{MSE}[\hat{\theta}]$};
        \addplot[smooth, BLUE!60!white, line width=1pt ,mark=x] table[x=N, y expr=\thisrow{Evar}*(180/pi)^2] {./results/fig2_near_right.txt};
        \addlegendentry{ $\text{Var}[\hat{\theta}]$};
        \addplot[smooth, RED!60!white, line width=1pt ,mark=o] table[x=N, y expr=\thisrow{GEmse}*(180/pi)^2] {./results/fig2_near_right.txt};
        \addlegendentry{$\text{MSE}[\hat{\theta}^G]$};
        \addplot[smooth, RED!60!white, line width=1pt ,mark=x] table[x=N, y expr=\thisrow{GEvar}*(180/pi)^2] {./results/fig2_near_right.txt};
        \addlegendentry{ $\text{Var}[\hat{\theta}^G]$};

    \end{axis}
        \node at (current bounding box.south) [yshift=-5pt,xshift=12pt] {\small (e)};
\end{tikzpicture}
\end{minipage}
\vspace{-2pt}
\hspace{-15pt}
\begin{minipage}[t]{0.65\columnwidth}
\centering
\begin{tikzpicture}
    \centering
    \renewcommand{\axisdefaulttryminticks}{4} 
    \pgfplotsset{every major grid/.style={densely dashed}}    
    \tikzstyle{every axis x label}+=[yshift=5pt]
    \pgfplotsset{every axis legend/.append style={cells={anchor=west},fill=white, at={(1,1)}, anchor=north east, font=\footnotesize,legend columns=2,}} %
    \begin{axis}[
        width=1\columnwidth,
        height=0.9\columnwidth,
        ymajorgrids=true,
        scaled ticks=true,
        xtick={100,200,300,400},
        xmin=50,xmax=400,
        ymin=1e-01,
        xlabel = { $N$ },
        xlabel style={yshift=3pt},
        ylabel = \empty,
        scaled ticks=true,
        ymode=log,
        grid = both,
        ]
        \addplot[smooth, BLUE!60!white, line width=1pt ,mark=o] table[x=N, y=Emse] {./results/fig2_unco_right.txt};
        \addlegendentry{$\text{MSE}[\hat{\theta}]$};
        \addplot[smooth, BLUE!60!white, line width=1pt ,mark=x] table[x=N, y=Evar] {./results/fig2_unco_right.txt};
        \addlegendentry{ $\text{Var}[\hat{\theta}]$};
        \addplot[smooth, RED!60!white, line width=1pt ,mark=o] table[x=N, y=GEmse] {./results/fig2_unco_right.txt};
        \addlegendentry{$\text{MSE}[\hat{\theta}^G]$};
        \addplot[smooth, RED!60!white, line width=1pt ,mark=x] table[x=N, y=GEvar] {./results/fig2_unco_right.txt};
        \addlegendentry{ $\text{Var}[\hat{\theta}^G]$};
    \end{axis}
        \node at (current bounding box.south) [yshift=-5pt,xshift=12pt] {\small (f)};
\end{tikzpicture}
\end{minipage} 
\vspace{-1pt}
\caption{
\textbf{Top}: (a)-(c) comparison of DoA estimates under a widely-spaced scenario.
\textbf{Bottom}: (d)-(f) MSEs and variances of DoA estimates versus $N$, averaged over $2\,000$ independent trials. Subfigures (a) and (d) correspond to widely-spaced correlated sources, (b) and (e) to closely-spaced uncorrelated sources, and (c) and (f) to widely-spaced uncorrelated sources. The number of snapshots is set as $T=2N$.  }
\label{fig:correlated}
\vspace{-6pt}
\end{figure*}

Fig.~\ref{fig:correlated} provides empirical support for our discussion of the three special cases. We first consider the case of widely-spaced DoAs with correlated sources, as discussed in \Cref{rem:wide_correlated}.
Here, under the same setting as in the subfigure (a) of Fig.~\ref{fig:spectral_error}, we compare, in  Fig.~\ref{fig:correlated}(a), the DoAs estimates $\hat \theta$ from classical ESPRIT (as well as the theoretical characterizations $\bar \theta$ given in \Cref{theo:main}), to $\hat{\theta}^G$ from the G-ESPRIT method in  \Cref{prop:gesprit}. The true DoAs $\theta$ are indicated by the gray dashed lines.
We observe that, as $N$ increases with a fixed ratio $N/T = 1/2$,
\begin{enumerate}
    \item the theoretical characterization ($\bar \theta$) perfectly matches the classical ESPRIT estimate ($\hat \theta$, that is observed to diverge from the true DoAs $\theta$); and
    \item the proposed G-ESPRIT estimates $\hat{\theta}^G$ remove this bias.
\end{enumerate}

In Fig.~\ref{fig:correlated}(d), we compare the MSEs and variances of both classical ESPRIT ($\hat \theta$) and G-ESPRIT ($\hat{\theta}^G$) estimates. 
We observe that: 
\begin{enumerate}
    \item classical ESPRIT provides \emph{inconsistent} DoA estimates, with MSE much larger than the variance; and
    \item G-ESPRIT provides consistent estimates and, in addition, yields \emph{smaller variances} than ESPRIT.
\end{enumerate}

In Fig.~\ref{fig:correlated}(b) and \ref{fig:correlated}(e), we investigate the case of closely-spaced DoAs with equal power discussed in \Cref{rem:closely_equal}, in the same setting as the bottom panel of Fig.~\ref{fig:spectral_error}.
We observe that:
\begin{enumerate}
    \item classical ESPRIT fails to distinguish two closely-spaced DoAs with a distance of order $O(N^{-1})$; and
    \item the proposed G-ESPRIT is $N$-consistent in this setting, with variance coinciding with the MSE.
\end{enumerate}

To illustrate the ``lucky'' consistency of classical ESPRIT in the case of widely-spaced DoAs from uncorrelated sources discussed in \Cref{rem:wide_uncorrelated}, we show, in  Fig.~\ref{fig:correlated}(c) and \ref{fig:correlated}(f), the DoA estimation results in the setting as in Fig.~\ref{fig:correlated}(a) but with power matrix $\P =\I_K$.
We observe that the classical ESPRIT estimates are in agreement with those of G-ESPRIT, and are close to the true DoAs.
In this case, both ESPRIT are G-ESPRIT are unbiased estimators, having their MSEs coinciding with variances.

\begin{figure}
\begin{minipage}[t]{0.9\columnwidth}
\centering
\begin{tikzpicture}
    \centering
    \renewcommand{\axisdefaulttryminticks}{4} 
    \pgfplotsset{every major grid/.style={densely dashed}}    
    \tikzstyle{every axis x label}+=[yshift=5pt]
    \pgfplotsset{every axis legend/.append style={cells={anchor=west},fill=white, at={(1,1)}, anchor=north east, font=\tiny,legend columns=1}} %
    \begin{axis}[
      width=1\columnwidth,
      height=0.8\columnwidth,
      xmin = 50, xmax = 400,
      ymajorgrids=true,
      scaled ticks=true,
      xlabel = {  $N$},
      ylabel = { MSE (in deg$^2$) },
      grid=major,
      scaled ticks=true,
      ymode=log
      ]
      \addplot[smooth,mark=o,color=yellow1,line width=1pt] table[x=N, y=es1] {./results/power_matrix.txt};
      \addlegendentry{ ESPRIT, Wishart-type $\mathbf{P}$};
      \addplot[smooth,mark=triangle,color=yellow1,line width=1pt] table[x=N, y=es2] {./results/power_matrix.txt};
      \addlegendentry{ ESPRIT, Wishart-type $\mathbf{P}$};

      \addplot[smooth, mark=o,color=BLUE!60!white,line width=1pt] table[x=N, y=ges1] {./results/power_matrix.txt};
      \addlegendentry{ G-ESPRIT, Random-eigendecomp. $\mathbf{P}$ };
      \addplot[smooth, mark=triangle,color=BLUE!60!white,line width=1pt] table[x=N, y=ges2] {./results/power_matrix.txt};
      \addlegendentry{ G-ESPRIT, Random-eigendecomp. $\mathbf{P}$ };

      \addplot[smooth,mark=o, color=black,line width=0.8pt] table[x=N, y=crb1] {./results/power_matrix.txt};
      \addlegendentry{CRB, Wishart-type $\mathbf{P}$};
      \addplot[smooth,mark=triangle, color=black,line width=0.8pt] table[x=N, y=crb2] {./results/power_matrix.txt};
      \addlegendentry{CRB, Random-eigendecomp. $\mathbf{P}$ };
  \end{axis}
\end{tikzpicture}
\end{minipage} 
\vspace{-4pt}
\caption{MSE versus array size $N$ for different random power matrix. Two closely-spaced sources are located at $\theta_1=0, \theta_2=0.8\times \frac{2\pi}{N}$. Results are obtained by averaging over $2\,000$ independent trials.}
\label{fig:diff-P}
\vspace{-4pt}
\end{figure}

In Fig.~\ref{fig:diff-P}, we compare the MSE performance of ESPRIT and G-ESPRIT, together with the corresponding Cramér--Rao Bound (CRB)\footnote{Here we compute the Cramér--Rao Bound for ULA according to \cite[Theorem 4.3]{stoica1989music} as $ \text{CRB}=\frac{\sigma^2}{2N} \left\{ \Re \left[ \left\{ \mathbf{D}^\H (\I_N - \A(\A^\H \A)^{-1}\A^\H) \mathbf{D} \right\} \odot \P^\T  \right] \right\}^{-1}$ with $\mathbf{D} = [\frac{\partial \mathbf{a}(\theta_1)}{\partial  \theta_1},\ldots,\frac{\partial \mathbf{a}(\theta_K)}{\partial  \theta_K}] \in \mathbb{C}^{N \times K}$. The CRB for the physical angle $\bm{\phi}$ is then obtained via the Jacobian transformation $
\text{CRB}_{\bm{\phi}} = \mathbf{J}^{-1}\text{CRB}_{\bm{\theta}}(\mathbf{J}^{-1})^{\H}$, where $\mathbf{J} = \diag\!\big(2\pi(d/\lambda)\cos\phi_k\big)_{k=1}^K$.}.
Two classes of randomly generated non-diagonal signal power matrices are considered, with an independent realization generated in each Monte Carlo trial. 
Specifically, the first class is a Wishart-distributed matrix, while the second class is a random eigen-decomposition model of the form $\mathbf{P}=\mathbf{U}\bm{\Lambda}\mathbf{U}^\H$, where $\mathbf{U}$ is a random unitary (i.e., Haar-distributed) matrix and the diagonal entries of $\bm{\Lambda}$ are independently drawn from a uniform distribution on $[1/2,2]$.
For a fair comparison, each randomly generated $\mathbf{P}$ is scaled such that the $K$-th signal spike attains the same target $\text{SNR}=5$ dB.
Here, we consider the same closely-spaced DoA scenario as in Figs.~\ref{fig:spectral_error} and Fig.~\ref{fig:correlated}.
The results in Fig.~\ref{fig:diff-P} show that the performance trends under the two types of random $\mathbf{P}$ are highly consistent, and G-ESPRIT consistently achieves lower MSE than ESPRIT. 
This indicates that the proposed method exhibits good robustness with respect to different covariance structures.

\begin{figure}[!t]
\begin{minipage}[t]{0.9\columnwidth}
\centering
\begin{tikzpicture}
    \centering
    \renewcommand{\axisdefaulttryminticks}{4} 
    \pgfplotsset{every major grid/.style={densely dashed}}    
    \tikzstyle{every axis x label}+=[yshift=5pt]
    \pgfplotsset{every axis legend/.append style={cells={anchor=west},fill=white, at={(1,1)}, anchor=north east, font=\tiny,legend columns=2}} %
    \begin{axis}[
      width=1\columnwidth,
      height=0.8\columnwidth,
      xmin = 50,
      xmax = 400,
      ymax = 2e-1,
      ymin = 2e-07,
      ymajorgrids=true,
      scaled ticks=true,
      xlabel = { $ N$},
      xlabel style={yshift=2pt},
      ylabel = { MSE (in deg$^2$) },
      grid=major,
      scaled ticks=true,
      ymode=log
      ]
      \addplot[smooth,mark=o,color=yellow1,line width=1pt] table[x=N, y=es1] {./results/near_corr.txt};
      \addlegendentry{ ESPRIT, $\rho = 0$};
      \addplot[smooth, mark=o,color=BLUE!60!white,line width=1pt] table[x=N, y=ges1] {./results/near_corr.txt};
      \addlegendentry{ G-ESPRIT, $\rho = 0$};

      \addplot[smooth,mark=triangle,color=yellow1,line width=1pt] table[x=N, y=es2] {./results/near_corr.txt};
      \addlegendentry{ ESPRIT, $\rho = 0.2$};
      \addplot[smooth, mark=triangle,color=BLUE!60!white,line width=1pt] table[x=N, y=ges2] {./results/near_corr.txt};
      \addlegendentry{ G-ESPRIT, $\rho = 0.2$};

      \addplot[smooth,mark=x,color=yellow1,line width=1pt] table[x=N, y=es3] {./results/near_corr.txt};
      \addlegendentry{ ESPRIT, $\rho = 0.6$};
      \addplot[smooth, mark=x,color=BLUE!60!white,line width=1pt] table[x=N, y=ges3] {./results/near_corr.txt};
      \addlegendentry{ G-ESPRIT, $\rho = 0.6$};

      \addplot[smooth,mark=o, color=black,line width=0.8pt] table[x=N, y=crb1] {./results/near_corr.txt};
      \addlegendentry{CRB, $\rho = 0$};
      \addplot[smooth,mark=triangle, color=black,line width=0.8pt] table[x=N, y=crb2] {./results/near_corr.txt};
      \addlegendentry{CRB, $\rho = 0.2$};
      \addplot[smooth,mark=x, color=black,line width=0.8pt] table[x=N, y=crb3] {./results/near_corr.txt};
      \addlegendentry{CRB, $\rho = 0.6$};
  \end{axis}
\end{tikzpicture}
\end{minipage} 
\vspace{-2pt}
\caption{MSE versus array size $N$ for different correlation level $\rho$ with the power matrix $\mathbf{P} =  2000\big((1-\rho)\mathbf{I}_4 + \rho \mathbf{1}\mathbf{1}^\T\big)$. Four closely-spaced sources are located at $\theta_1=-0.3\times \frac{2\pi}{N}$, $\theta_2=0$, $\theta_3=0.2\times \frac{2\pi}{N}$, $\theta_4=0.5\times \frac{2\pi}{N}$. Results are obtained by averaging over $2\,000$ independent trials.}
\label{fig:diff-rho}
\vspace{-2pt}
\end{figure}
\begin{figure}[!t]
\begin{minipage}[!ht]{0.9\columnwidth}
\centering
\begin{tikzpicture}
    \centering
    \renewcommand{\axisdefaulttryminticks}{4} 
    \pgfplotsset{every major grid/.style={densely dashed}}    
    \tikzstyle{every axis x label}+=[yshift=5pt]
    \pgfplotsset{every axis legend/.append style={cells={anchor=west},fill=white, at={(1,1)}, anchor=north east, font=\footnotesize,legend columns=1}} %
    \begin{axis}[
      width=1\columnwidth,
      height=0.8\columnwidth,
      xmin = -4,
      xmax = 28,
      ymax = 1000,
      ymin = 4e-09,
      ymajorgrids=true,
      scaled ticks=true,
      xlabel = {  Relative SNR $R_{\rm SNR}$ ($\text{dB}$) },
      ylabel = { MSE (in deg$^2$) },
      grid=major,
      scaled ticks=true,
      ymode=log
      ]
      \addplot[smooth,mark=x,color=yellow1,line width=1pt] table[x=snr, y=es] {./results/paper_far.txt};
      \addlegendentry{ESPRIT};
      \addplot[smooth, mark=o,color=BLUE!60!white,line width=1pt] table[x=snr, y=ges] {./results/paper_far.txt};
      \addlegendentry{G-ESPRIT};
      \addplot[smooth,mark=x,color=GREEN!80!white,line width=1pt] table[x=snr, y=mu] {./results/paper_far.txt};
      \addlegendentry{MUSIC};
      \addplot[smooth,mark=o,color=RED!60!white,line width=1pt] table[x=snr, y=gmu] {./results/paper_far.txt};
      \addlegendentry{G-MUSIC};
      \addplot[densely dashed,color=black,line width=1pt] table[x=snr, y=crb] {./results/paper_far.txt};
      \addlegendentry{CRB};
      \addplot [densely dotted, color=gray, line width=1.5pt] coordinates {(-3.5,1e-09) ( -3.5,1000)} node [pos=0.5, right] {};
  \end{axis}
\end{tikzpicture}
\end{minipage} 
\vspace{-2pt}
\caption{Empirical MSEs for widely-spaced DoAs versus relative SNR, where $\Delta=1$, $n = N-1$, $N=400$ and $T=800$. Four sources are located at $\theta_1=-0.6\pi$, $\theta_2=-0.22\pi$, $\theta_3=0.25\pi$, $\theta_4=0.4\pi$, and the power matrix is given by $\mathbf{P} = 2\big((1-\rho)\mathbf{I}_4 + \rho\mathbf{1}\mathbf{1}^\T\big)$ with a correlation coefficient $\rho=0.2$. Results are obtained by averaging over $2\,000$ independent trials.}
\label{fig:compare_widely_spaced}
\end{figure}
\begin{figure}[!t]
\begin{minipage}[!ht]{0.9\columnwidth}
\centering
\begin{tikzpicture}
    \centering
    \renewcommand{\axisdefaulttryminticks}{4} 
    \pgfplotsset{every major grid/.style={densely dashed}}    
    \tikzstyle{every axis x label}+=[yshift=5pt]
    \pgfplotsset{every axis legend/.append style={cells={anchor=west},fill=white, at={(1,1)}, anchor=north east, font=\footnotesize,legend columns=1}} %
    \begin{axis}[
      width=1\columnwidth,
      height=0.8\columnwidth,
      xmin = -4,
      xmax = 28,
      ymax = 2,
      ymin = 5e-08,
      ymajorgrids=true,
      scaled ticks=true,
      xlabel = {  Relative SNR $R_{\rm SNR}$ ($\text{dB}$) },
      ylabel = { MSE (in deg$^2$) },
      grid=major,
      scaled ticks=true,
      ymode=log
      ]
      \addplot[smooth,mark=x,color=yellow1,line width=1pt] table[x=snr, y=es] {./results/rerunFig7.txt};
      \addlegendentry{ESPRIT};
      \addplot[smooth, mark=o,color=BLUE!60!white,line width=1pt] table[x=snr, y=ges] {./results/rerunFig7.txt};
      \addlegendentry{G-ESPRIT};
      \addplot[smooth,mark=x,color=GREEN!80!white,line width=1pt] table[x=snr, y=mu] {./results/rerunFig7.txt};
      \addlegendentry{MUSIC};
      \addplot[smooth,mark=o,color=RED!60!white,line width=1pt] table[x=snr, y=gmu] {./results/rerunFig7.txt};
      \addlegendentry{G-MUSIC};
      \addplot[densely dashed,color=black,line width=1pt] table[x=snr, y=crb] {./results/rerunFig7.txt};
      \addlegendentry{CRB};
      \addplot [densely dotted, color=gray, line width=1.5pt] coordinates {( 1.8,5e-08) ( 1.8,2)} node [pos=0.5, right] {};
  \end{axis}
\end{tikzpicture}
\end{minipage} 
\vspace{-2pt}
\caption{ Empirical MSEs for closely-spaced DoAs versus relative SNR, where $\Delta=120$, $n = N-\Delta$, $N=400$ and $T=800$. Four sources are located at $\theta_1=-0.3\!\times\! \frac{2\pi}{N}$, $\theta_2\!=\!0$, $\theta_3=0.2\times \frac{2\pi}{N}$, $\theta_4=0.5\times \frac{2\pi}{N}$, and the power matrix is given by $\mathbf{P} = 2000\big((1-\rho)\mathbf{I}_4 + \rho \mathbf{1}\mathbf{1}^\T\big)$ with a correlation coefficient $\rho=0.2$. Results are obtained by averaging over $2\,000$ independent trials.} 
\label{fig:compare_closely_spaced}
\vspace{-4pt}
\end{figure}

\begin{figure}[!t]
\begin{minipage}[t]{0.9\columnwidth}
\centering
\begin{tikzpicture}
    \centering
    \renewcommand{\axisdefaulttryminticks}{4} 
    \pgfplotsset{every major grid/.style={densely dashed}}    
    \tikzstyle{every axis x label}+=[yshift=5pt]
    \pgfplotsset{every axis legend/.append style={cells={anchor=west},fill=white, at={(1,1)}, anchor=north east, font=\tiny,legend columns=2}} %
    \begin{axis}[
      width=1\columnwidth,
      height=0.8\columnwidth,
      xmin = 1,
      xmax = 10,
      ymax = 3e+2,
      ymin = 4e-05,
      ymajorgrids=true,
      scaled ticks=true,
      xtick={2,4,6,8,10},
      xticklabels={$\frac{0.4\pi}{N}$,$\frac{0.8\pi}{N}$,$\frac{1.6\pi}{N}$,$\frac{3.2\pi}{N}$,$\frac{6.4\pi}{N}$},
      xlabel={Angular separation $d_s$},
      ylabel = { MSE (in deg$^2$) },
      grid=major,
      scaled ticks=true,
      ymode=log
      ]
      \addplot[smooth,mark=x,color=BLUE!60!white ,line width=1pt] table[x=coeff, y=ges1] {./results/closer_angle.txt};
      \addlegendentry{ G-ESPRIT, $K = 8$};
      \addplot[smooth, mark=o,color=BLUE!60!white,line width=1pt] table[x=coeff, y=gmu1] {./results/closer_angle.txt};
      \addlegendentry{ G-MUSIC, $K = 8$};

      \addplot[smooth,mark=x,color=RED!60!white,line width=1pt] table[x=coeff, y=ges2] {./results/closer_angle.txt};
      \addlegendentry{ G-ESPRIT, $K = 4$};
      \addplot[smooth, mark=o,color=RED!60!white,line width=1pt] table[x=coeff, y=gmu2] {./results/closer_angle.txt};
      \addlegendentry{ G-MUSIC, $K = 4$};

      \addplot[smooth,mark=x,color=yellow1,line width=1pt] table[x=coeff, y=ges3] {./results/closer_angle.txt};
      \addlegendentry{ G-ESPRIT, $K = 2$};
      \addplot[smooth, mark=o,color=yellow1,line width=1pt] table[x=coeff, y=gmu3] {./results/closer_angle.txt};
      \addlegendentry{ G-MUSIC, $K = 2$};

  \end{axis}
\end{tikzpicture}
\end{minipage} 
\vspace{-2pt}
\caption{MSE versus angular interval under different numbers of sources $K=2,4,8$, with $N=400$ and $T=800$. The power matrix is $\mathbf{P} = 2000\big((1-\rho)\mathbf{I}_K + \rho\mathbf{1}\mathbf{1}^\T\big)$ with $\rho=0.2$. Results are obtained by averaging over $2\,000$ independent trials.} 
\label{fig:closer_angle}
\vspace{-2pt}
\end{figure}

\begin{figure*}[!th]
    \centering
    \begin{minipage}[t]{0.65\columnwidth}
    \centering
    \begin{tikzpicture}
        \def \deg2rady{180/3.1415926}
        \centering
        \renewcommand{\axisdefaulttryminticks}{4} 
        \pgfplotsset{every major grid/.style={densely dashed}}       
        \tikzstyle{every axis x label}+=[yshift=5pt]
        \pgfplotsset{every axis legend/.append style={cells={anchor=east},fill=white, at={(2.8,1.2),sacel=0.2}, anchor=north east, font=\footnotesize,legend columns=5 }} 
        \begin{axis}[
        width=1\columnwidth,
        height=0.9\columnwidth,
        xmin = -4,
        xmax = 28,
        ymax = 1,
        ymin = 2e-8,
        ymajorgrids=true,
        scaled ticks=true,
        xlabel = {  SNR $R_{\rm SNR}$ ($\text{dB}$) },
        ylabel = {MSE (in deg$^2$) },
        grid=major,
        scaled ticks=true,
        ymode=log
        ]
        \addplot[smooth,mark=x,color=yellow1,line width=1pt] table[x=snr, y=es] {./results/noise/Rademacher.txt};
        \addlegendentry{ESPRIT};
        \addplot[smooth, mark=o,color=BLUE!60!white,line width=1pt] table[x=snr, y=ges] {./results/noise/Rademacher.txt};
        \addlegendentry{G-ESPRIT};
        \addplot[smooth,mark=x,color=GREEN!80!white,line width=1pt] table[x=snr, y=mu] {./results/noise/Rademacher.txt};
        \addlegendentry{MUSIC};
        \addplot[smooth,mark=o,color=RED!60!white,line width=1pt] table[x=snr, y=gmu] {./results/noise/Rademacher.txt};
        \addlegendentry{G-MUSIC};
        \addplot[densely dashed,color=black,line width=1pt] table[x=snr, y=crb] {./results/noise/Rademacher.txt};
        \addlegendentry{CRB};
        \end{axis}
        \node at (current bounding box.south) [yshift=-7pt,xshift=-85pt] {\small (a)};
\end{tikzpicture}
\end{minipage}
\hspace{-10pt} 
\begin{minipage}[t]{0.65\columnwidth}
\centering
\begin{tikzpicture}
    \centering
    \renewcommand{\axisdefaulttryminticks}{4} 
    \pgfplotsset{every major grid/.style={densely dashed}}    
    \tikzstyle{every axis x label}+=[yshift=5pt]
    \pgfplotsset{every axis legend/.append style={cells={anchor=east},fill=white, at={(0.98,0.7)}, anchor=north east, font=\footnotesize,legend columns=2,}} %
    \begin{axis}[
        width=1\columnwidth,
        height=0.9\columnwidth,
        xmin = -4,
        xmax = 28,
        ymax = 1,
        ymin = 2e-8,
        ymajorgrids=true,
        scaled ticks=true,
        xlabel = {  SNR $R_{\rm SNR}$ ($\text{dB}$) },
        ylabel = \empty,
        grid=major,
        scaled ticks=true,
        ymode=log
        ]
        \addplot[smooth,mark=x,color=yellow1,line width=1pt] table[x=snr, y=es] {./results/noise/uniform.txt};
        \addplot[smooth, mark=o,color=BLUE!60!white,line width=1pt] table[x=snr, y=ges] {./results/noise/uniform.txt};
        \addplot[smooth,mark=x,color=GREEN!80!white,line width=1pt] table[x=snr, y=mu] {./results/noise/uniform.txt};
        \addplot[smooth,mark=o,color=RED!60!white,line width=1pt] table[x=snr, y=gmu] {./results/noise/uniform.txt};
        \addplot[densely dashed,color=black,line width=1pt] table[x=snr, y=crb] {./results/noise/uniform.txt};
    \end{axis}
    \node at (current bounding box.south) [yshift=-7pt,xshift=15pt] {\small (b)};
\end{tikzpicture}
\end{minipage} 
\hspace{-15pt}
\begin{minipage}[t]{0.65\columnwidth}
\centering
\begin{tikzpicture}
    \centering
    \renewcommand{\axisdefaulttryminticks}{4} 
    \pgfplotsset{every major grid/.style={densely dashed}}    
    \tikzstyle{every axis x label}+=[yshift=5pt]
    \pgfplotsset{every axis legend/.append style={cells={anchor=east},fill=white, at={(0.98,0.7)}, anchor=north east, font=\footnotesize,legend columns=2,}} %
    \begin{axis}[
        width=1\columnwidth,
        height=0.9\columnwidth,
        xmin = -4,
        xmax = 28,
        ymax = 1,
        ymin = 2e-8,
        ymajorgrids=true,
        scaled ticks=true,
        xlabel = { SNR $R_{\rm SNR}$ ($\text{dB}$) },
        ylabel =\empty,
        grid=major,
        scaled ticks=true,
        ymode=log
        ]
        \addplot[smooth,mark=x,color=yellow1,line width=1pt] table[x=snr, y=es] {./results/noise/heavy-tailed.txt};
        \addplot[smooth, mark=o,color=BLUE!60!white,line width=1pt] table[x=snr, y=ges] {./results/noise/heavy-tailed.txt};
        \addplot[smooth,mark=x,color=GREEN!80!white,line width=1pt] table[x=snr, y=mu] {./results/noise/heavy-tailed.txt};
        \addplot[smooth,mark=o,color=RED!60!white,line width=1pt] table[x=snr, y=gmu] {./results/noise/heavy-tailed.txt};
    \end{axis}
    \node at (current bounding box.south) [yshift=-7pt,xshift=15pt] {\small (c)};
\end{tikzpicture}
\end{minipage} 
\caption{MSE performance under three non-Gaussian noise models: \textbf{(a)} Rademacher noise, \textbf{(b)} uniform noise, and \textbf{(c)} heavy-tailed noise. Results are obtained by averaging over $2\,000$ independent trials.} 
\label{fig:nonGaussian}
\vspace{-2pt}
\end{figure*}

In Fig.~\ref{fig:diff-rho}, we adopt $\mathbf{P} = 2000\big((1-\rho)\mathbf{I}_4 + \rho \mathbf{1}\mathbf{1}^\T\big)$ as the power matrix for correlated signals, and investigate the impact of the correlation coefficient $\rho$ on the algorithm performance. 
The case $\rho=0$ corresponds to uncorrelated sources.
We observe that increasing source correlation substantially degrades the performance of classical ESPRIT, consistent with its well-known sensitivity to subspace leakage.
The proposed G-ESPRIT, on the other hand, maintains a significantly smaller estimation error.

\subsection{Comparison to other DoA approaches}
\label{subsec:comparision}
In this subsection, we compare the performances of classical ESPRIT in \Cref{alg:TraESPRIT}, the proposed G-ESPRIT in \Cref{alg:G-ESPRIT}, to other popular DoA estimation methods such as MUSIC~\cite{schmidt1986multiple}, and G-MUSIC~\cite{vallet2012improved,Vallet2015Performance}.


Fig.~\ref{fig:compare_widely_spaced} compares the MSEs of the aforementioned four DoA estimation methods as a function of the (relative) signal-to-noise ratio (SNR), in the case of widely-spaced DoAs.
The noise power is set as $\sigma^2 = 1$.
In this case, the minimum subspace separation condition in \Cref{ass:subspace} for the four sources is satisfied at relative SNR $-3.5 \mathrm{dB}$, as indicated by the vertical dotted lines. 
Interestingly, a clear \emph{phase-transition} behavior can be observed for all subspace-based methods.
This is an empirical manifestation of the counterintuitive large-dimensional behavior of SCM eigenspectral discussed in \Cref{theo:spikes}.
For widely spaced DoAs, G-ESPRIT outperforms classical ESPRIT over the SNR range of $[-4, 8]$ dB, whereas its performance advantage becomes marginal at high SNRs.
Moreover, ESPRIT-type methods yield significantly higher MSEs than MUSIC-type methods in this setting.
To have a theoretical grasp of this empirical observation, a second-order behavior analysis is needed.

We then compare in Fig.~\ref{fig:compare_closely_spaced} the MSEs of four subspace methods for closely-spaced DoAs.
To further assess robustness, a deterministic signal model is adopted, where the empirical covariance of the source signals is constructed to match the prescribed power matrix $\mathbf{P}$.
The minimum subspace separation among the four sources occurs at a relative SNR of $1.8~\mathrm{dB}$.
We observe that the classical MUSIC performs poorly and has a relatively large MSE in this setting, in accordance with the inconsistency proven in \cite{Vallet2015Performance}.
Different from the case of widely-spaced DoA in Fig.~\ref{fig:compare_widely_spaced}, here the proposed G-ESPRIT beats classical ESPRIT over a wide SNR range (approximately $4-18$ dB) and gets close to the CRB. 
We believe that this is due to the better choice of large $\Delta$.
We also note that the slight apparent rise of the ESPRIT curve in the mid-SNR region is only caused by line smoothing; the actual simulated MSE values remain monotonically decreasing.

To further validate the superiority of G-ESPRIT in resolving closely-spaced DoAs, Fig.~\ref{fig:closer_angle} analyzes its performance under different angular separation and numbers of sources.
The scaling of $\mathbf{P}$ is chosen such that the weakest source operates at approximately the same SNR level , which guarantees that the phase transition condition is satisfied.
In Fig.~\ref{fig:closer_angle}, the DoAs are uniformly arranged with an identical angular separation between adjacent sources, given by $d_s$, such that $\theta_k=\theta_1+(k-1)d_s, k=1,\ldots,K$.  
When the sources are closely-spaced (small angular separation), both subspace estimators experience degraded performance due to increased coherence among steering vectors. 
Nonetheless, G-ESPRIT maintains significantly lower MSE compared to G-MUSIC across all values of $K$, demonstrating its robustness in challenging high-coherence scenarios.
For example, the MSE gap between the two methods at $d_s=\frac{0.6\pi}{N}$ with $K=4$ is approximately 2 dB, which is consistent with the performance observed in Fig.~\ref{fig:compare_closely_spaced} at low-to-moderate SNRs.
As the angular separation increases, e.g., at $d_s=6.4\pi/N$, G-MUSIC improves and even exhibits a lower MSE than G-ESPRIT, which is also consistent with the widely-spaced source case discussed in Fig.~\ref{fig:compare_widely_spaced}.
Increasing the number of sources naturally raises the difficulty of the problem, but the performance trends remain consistent, further confirming the performance advantage of G-ESPRIT in closely-spaced DoAs and multi-source scenarios.

Although Gaussian noise is assumed for analytical tractability, the main RMT results invoked in this work are universal for a broad class of sub-Gaussian noise distributions~\cite{yin1988on}.
Finally, in Fig.~\ref{fig:nonGaussian}, we further examine the robustness of the proposed algorithms from the perspective of noise Gaussianity.
Specifically, we consider three non-Gaussian noise models, namely Rademacher noise, uniform noise, and heavy-tailed noise, while keeping the same parameter settings as in Fig.~\ref{fig:compare_closely_spaced}. For the first two noise models, which both satisfy the finite-variance condition, the simulation results remain consistent with the theoretical analysis: the classical ESPRIT and MUSIC suffer from significant bias, whereas the proposed G-ESPRIT and G-MUSIC exhibit stable performance. In contrast, for the heavy-tailed noise model, the finite-variance assumption is violated and the theoretical guarantees no longer apply. This further suggests that the theoretical results hold potential for extension to certain sub-Gaussian noise settings.

\section{Conclusion}
\label{sec:conclusion}

In this paper, we perform a large-dimensional analysis of the classical ESPRIT DoA estimation method in the regime of large arrays and limited snapshots, where the array length $N$ and the number of samples $T$ are both large and comparable. 
Our study covers both widely-spaced and closely-spaced DoA scenarios. 
We show that while classical ESPRIT is generally \emph{inconsistent} in these settings, this bias can be corrected using the proposed G-ESPRIT approach.
G-ESPRIT exhibits a pronounced performance advantage over conventional ESPRIT and G-MUSIC methods especially in the closely spaced DoA scenario, demonstrating its ability to resolve extremely small angular separations between sources.

From a technical perspective, we establish a novel bound on the eigenvalue differences between two possibly non-Hermitian matrices, which may be of independent interest in non-Hermitian spectral analysis and related structured matrix problems, for example, approximately gauge-equivalent weighted digraphs and gain-graph switching models. In these models, vertex gauge transformations or switching operations naturally preserve cycle products of edge weights, making cycle-product control more suitable than standard entrywise or norm-based perturbation measures in approximate settings; see, e.g.,~\cite{engel1982algorithms,cavaleri2022godsil}.
Numerical simulations validate the consistency of G-ESPRIT and highlight its reduced variance compared to classical ESPRIT -- though a rigorous theoretical characterization of this variance reduction is beyond the scope of this paper.

It would be of future interest to extend the RMT analysis framework in this paper to characterize the (e.g., CLT-type) second-order fluctuation of both ESPRIT and G-ESPRIT, as to assess quantitatively their performance gaps from the CRB.

\bibliographystyle{IEEEbib}
\bibliography{ref}

@article{nathanson2024Continuity,
  title = {Continuity of the Roots of a Polynomial},
  author = {Nathanson, Melvyn B. and Ross, David A.},
  year = {2024},
  month = jun,
  journal = {Commun Algebra.},
  volume = {52},
  number = {6},
  pages = {2509--2518},
  publisher = {Taylor \& Francis},
  issn = {0092-7872},
  doi = {10.1080/00927872.2023.2301540},
  urldate = {2024-12-28}
}

@article{hachem2007deterministic, 
year = {2007}, 
title = {{Deterministic equivalents for certain functionals of large random matrices}}, 
author = {Hachem, Walid and Loubaton, Philippe and Najim, Jamal}, 
journal = {Ann Appl Probab.}, 
issn = {1050-5164}, 
doi = {10.1214/105051606000000925}, 
eprint = {math/0507172}, 
pages = {875--930}, 
number = {3}, 
volume = {17}
}

@book{rozanov1967stationary, 
year = {1967}, 
title = {{Stationary Random Processes}}, 
author = {Rozanov, Yu. A.}, 
url = {https://openlibrary.org/books/OL21849368M/}, 
series = {Holden-Day series in time series analysis}, 
publisher = {Holden-Day, San Francisco}
}

@article{schmidt1986multiple,
  title={Multiple emitter location and signal parameter estimation},
  author={Schmidt, Ralph},
  journal={IEEE Trans Antennas Propag.},
  volume={34},
  number={3},
  pages={276--280},
  year={1986},
  publisher={IEEE}
}

@ARTICLE{1457851,
  author={Paulraj, A. and Roy, R. and Kailath, T.},
  journal={Proc. IEEE}, 
  title={A subspace rotation approach to signal parameter estimation}, 
  year={1986},
  volume={74},
  number={7},
  pages={1044-1046},
  doi={10.1109/PROC.1986.13583}}

@article{rao1989performance,
  title={Performance analysis of root-MUSIC},
  author={Rao, Bhaskar D and Hari, KV Sl},
  journal={IEEE Trans Acoust.},
  volume={37},
  number={12},
  pages={1939--1949},
  year={1989},
  publisher={IEEE}
}

@article{stoica1995resolution,
  title={On the resolution performance of spectral analysis},
  author={Stoica, Petre and {\v{S}}imonyte, Virginija and S{\"o}derstr{\"o}m, Torsten},
  journal={Signal processing},
  volume={44},
  number={2},
  pages={153--161},
  year={1995},
  publisher={Elsevier}
}

@article{stoica1989music,
  title={MUSIC, maximum likelihood, and Cramer-Rao bound},
  author={Stoica, Petre and Nehorai, Arye},
  journal={IEEE Trans Acoust.},
  volume={37},
  number={5},
  pages={720--741},
  year={1989},
  publisher={IEEE}
}

@article{ottersten1991performance,
  title={Performance analysis of the total least squares ESPRIT algorithm},
  author={Ottersten, Bjorn and Viberg, Mats and Kailath, Thomas},
  journal={IEEE Trans Signal Process.},
  volume={39},
  number={5},
  pages={1122--1135},
  year={1991},
  publisher={IEEE}
}

@article{viberg1991sensor,
  title={Sensor array processing based on subspace fitting},
  author={Viberg, Mats and Ottersten, Bjorn},
  journal={IEEE Trans Signal Process.},
  volume={39},
  number={5},
  pages={1110--1121},
  year={1991},
  publisher={IEEE}
}

@article{paul2014Random,
  title = {Random Matrix Theory in Statistics: {{A}} Review},
  shorttitle = {Random Matrix Theory in Statistics},
  author = {Paul, Debashis and Aue, Alexander},
  year = {2014},
  month = jul,
  journal = {J Stat Plan Inference.},
  volume = {150},
  pages = {1--29},
  issn = {0378-3758},
  doi = {10.1016/j.jspi.2013.09.005},
  urldate = {2023-09-07}
}

@inproceedings{lavate2010Performance,
  title = {Performance {{Analysis}} of {{MUSIC}} and {{ESPRIT DOA Estimation Algorithms}} for {{Adaptive Array Smart Antenna}} in {{Mobile Communication}}},
  booktitle = {2010 {{Second International Conference}} on {{Computer}} and {{Network Technology}}},
  author = {Lavate, T.B. and Kokate, V.K. and Sapkal, A.M.},
  year = {2010},
  month = apr,
  pages = {308--311},
  doi = {10.1109/ICCNT.2010.45}
}

@inproceedings{roy1987Comparative,
  title = {Comparative Performance of {{ESPRIT}} and {{MUSIC}} for Direction-of-Arrival Estimation},
  booktitle = {{{ICASSP}} '87. {{IEEE International Conference}} on {{Acoustics}}, {{Speech}}, and {{Signal Processing}}},
  author = {Roy, R. and Paulraj, A. and Kailath, T.},
  year = {1987},
  month = apr,
  volume = {12},
  pages = {2344--2347},
  doi = {10.1109/ICASSP.1987.1169322}
}

@article{vallet2017Performance,
  title = {On the {{Performance}} of {{MUSIC With Toeplitz Rectification}} in the {{Context}} of {{Large Arrays}}},
  author = {Vallet, Pascal and Loubaton, Philippe},
  year = {2017},
  month = nov,
  journal = {IEEE Trans Signal Process.},
  volume = {65},
  number = {22},
  pages = {5848--5859},
  issn = {1941-0476},
  doi = {10.1109/TSP.2017.2742988}
}

@article{paulraj1986subspace, 
year = {1986}, 
title = {{A subspace rotation approach to signal parameter estimation}}, 
author = {Paulraj, A. and Roy, R. and Kailath, T.}, 
journal = {Proc. IEEE}, 
issn = {0018-9219}, 
doi = {10.1109/proc.1986.13583}, 
pages = {1044--1046}, 
number = {7}, 
volume = {74}
}

@article{vallet2012improved,
  title={Improved subspace estimation for multivariate observations of high dimension: the deterministic signals case},
  author={Vallet, Pascal and Loubaton, Philippe and Mestre, Xavier},
  journal={IEEE Trans Inf Theory.},
  volume={58},
  number={2},
  pages={1043--1068},
  year={2012},
  publisher={IEEE}
}

@article{hachem2012large,
  title = {Large Information plus Noise Random Matrix Models and Consistent Subspace Estimation in Large Sensor Networks},
  author = {Hachem, Walid and Loubaton, Philippe and Mestre, Xavier and Najim, Jamal and Vallet, Pascal},
  year = {2012},
  journal = {Random Matrices: Theory and Applications},
  volume = {1},
  number = {02},
  pages = {1150006},
  issn = {2010-3263},
  doi = {10.1142/s2010326311500067}
}

@article{mestre2008modified,
  title = {Modified Subspace Algorithms for {{DoA}} Estimation with Large Arrays},
  author = {Mestre, Xavier and Lagunas, Miguel {\'A}ngel},
  year = {2008},
  journal = {IEEE Trans Signal Process.},
  volume = {56},
  number = {2},
  pages = {598--614},
  issn = {1053-587X},
  doi = {10.1109/tsp.2007.907884}
}

@article{Vallet2012ACF,
  title={A CLT for the G-MUSIC DoA estimator},
  author={Pascal Vallet and Xavier Mestre and Philippe Loubaton},
  journal={Proc. European Signal Processing Conference (EUSIPCO)},
  year={2012},
  pages={2298-2302},
  url={https://api.semanticscholar.org/CorpusID:17803135}
}

@article{Mestre2011,
author = {Mestre, Xavier and Vallet, Pascal and Loubaton, Philippe and Hachem, Walid},
booktitle={IEEE Stat Signal Processing Workshop.}, 
year = {2011},
month = {07},
pages = {677-680},
title = {Asymptotic Analysis of a Consistent Subspace Estimator for Observations of Increasing Dimension},
doi = {10.1109/SSP.2011.5967792}
}

@book{van2000asymptotic, 
year = {2000}, 
title = {{Asymptotic Statistics}}, 
author = {Vaart, Aad W. Van der}, 
isbn = {9780521784504}, 
volume = {3}, 
series = {Cambridge Series in Statistical and Probabilistic Mathematics}, 
publisher = {Cambridge University Press}, 
doi = {10.1017/cbo9780511802256}
}

@article{Vallet2015Performance,
  title = {Performance Analysis of an Improved {{MUSIC DoA}} Estimator},
  author = {Vallet, Pascal and Mestre, Xavier and Loubaton, Philippe},
  year = {2015},
  journal = {IEEE Trans Signal Process.},
  volume = {63},
  number = {23},
  pages = {6407--6422},
  issn = {1053-587X},
  doi = {10.1109/tsp.2015.2465302}
}

@ARTICLE{32276,
  author={Roy, R. and Kailath, T.},
  journal={IEEE Trans Acoust.}, 
  title={ESPRIT-estimation of signal parameters via rotational invariance techniques}, 
  year={1989},
  volume={37},
  number={7},
  pages={984-995},
  doi={10.1109/29.32276}}

@article{marvcenko1967distribution,
  year     = {1967},
  title    = {{Distribution of eigenvalues for some sets of random matrices}},
  author   = {Marcenko, Vladimir A and Pastur, Leonid Andreevich},
  journal  = {Mathematics of the USSR-Sbornik},
  issn     = {0025-5734},
  doi      = {10.1070/sm1967v001n04abeh001994},
  pages    = {457},
  number   = {4},
  volume   = {1}
}

@article{silverstein1995empirical, 
year = {1995}, 
title = {{On the Empirical Distribution of Eigenvalues of a Class of Large Dimensional Random Matrices}}, 
author = {Silverstein, Jack W. and Bai, Zhidong}, 
journal = {J Multivar Anal.}, 
issn = {0047-259X}, 
doi = {10.1006/jmva.1995.1051}, 
pages = {175--192}, 
number = {2}, 
volume = {54}
}

@article{benaych2011eigenvalues, 
year = {2011}, 
title = {{The eigenvalues and eigenvectors of finite, low rank perturbations of large random matrices}}, 
author = {Benaych-Georges, Florent and Nadakuditi, Raj Rao}, 
journal = {Adv Math (N Y).}, 
issn = {0001-8708}, 
doi = {10.1016/j.aim.2011.02.007}, 
pages = {494--521}, 
number = {1}, 
volume = {227}
}

@article{baik2006eigenvalues, 
year = {2006}, 
title = {{Eigenvalues of large sample covariance matrices of spiked population models}}, 
author = {Baik, Jinho and Silverstein, Jack W.}, 
journal = {J Multivar Anal.}, 
issn = {0047-259X}, 
doi = {10.1016/j.jmva.2005.08.003}, 
pages = {1382--1408}, 
number = {6}, 
volume = {97}
}

@book{couillet2022RMT4ML,
  title     = {Random Matrix Methods for Machine Learning},
  author    = {Couillet, Romain and Liao, Zhenyu},
  year      = {2022},
  isbn      = {9781009186742},
  publisher = {Cambridge University Press}
}

@inbook{horn2012matrix,
  title={Matrix analysis},
  author={Horn, Roger A and Johnson, Charles R},
  year={2012},
  publisher={Cambridge university press},
  pages = {38-42}
}

@inbook{scott2012group,
  title={Group theory},
  author={Scott, William Raymond},
  year={2012},
  publisher={Courier Corporation},
  pages   = {8-11}
}

@ARTICLE{1164557,
  author={Wax, M. and Kailath, T.},
  journal={IEEE Transactions on Acoustics, Speech, and Signal Processing}, 
  title={Detection of signals by information theoretic criteria}, 
  year={1985},
  volume={33},
  number={2},
  pages={387-392},
  doi={10.1109/TASSP.1985.1164557}}

@ARTICLE{765149,
  author={Lemma, A.N. and van der Veen, A.-J. and Deprettere, E.F.},
  journal={IEEE Transactions on Signal Processing}, 
  title={Multiresolution ESPRIT algorithm}, 
  year={1999},
  volume={47},
  number={6},
  pages={1722-1726},
  doi={10.1109/78.765149}}

@article{hachem2013subspace,
  title={A subspace estimator for fixed rank perturbations of large random matrices},
  author={Hachem, Walid and Loubaton, Philippe and Mestre, Xavier and Najim, Jamal and Vallet, Pascal},
  journal={Journal of Multivariate Analysis},
  volume={114},
  pages={427--447},
  year={2013},
  publisher={Elsevier}
}

@article{baik2005phase,
  title={Phase transition of the largest eigenvalue for nonnull complex sample covariance matrices},
  author={Baik, Jinho and Ben Arous, G{\'e}rard and P{\'e}ch{\'e}, Sandrine},
  year={2005}
}

@article{valaee2004information,
  title={An information theoretic approach to source enumeration in array signal processing},
  author={Valaee, Shahrokh and Kabal, Peter},
  journal={IEEE Transactions on Signal Processing},
  volume={52},
  number={5},
  pages={1171--1178},
  year={2004},
  publisher={IEEE}
}

@article{yin1988on,
author = {Yin, YQ and Bai, Z. and Krishnaiah, P.},
year = {1988},
month = {08},
pages = {509-521},
title = {On the limit of the largest eigenvalue of the large dimensional sample covariance matrix},
volume = {78},
journal = {Probability Theory and Related Fields},
doi = {10.1007/BF00353874}
}

@article{cavaleri2022godsil,
  title={Godsil-McKay switchings for gain graphs},
  author={Cavaleri, Matteo and Donno, Alfredo and Spessato, Stefano},
  journal={arXiv preprint arXiv:2207.10986},
  year={2022}
}

@article{engel1982algorithms,
  title={Algorithms for testing the diagonal similarity of matrices and related problems},
  author={Engel, Gernot M and Schneider, Hans},
  journal={SIAM Journal on Algebraic Discrete Methods},
  volume={3},
  number={4},
  pages={429--438},
  year={1982},
  publisher={SIAM}
}
    
\clearpage
\begin{appendices}
\onecolumn



\section{Technical Lemmas}
\label{sec:techLemmas}

In this section, we prepare the readers with a few technical lemmas and examples.
\begin{Lemma}[Woodbury identity]\label{lem:Woodbury}
    For $\A \in \CC^{p \times p}$, $\U$, $\V \in \CC^{p \times n}$ such that both $\A$ and $\A + \U \V^\H$ are invertible, we have
    \begin{equation*}
        (\A + \U \V^\H)^{-1} = \A^{-1} - \A^{-1} \U (\I_n + \V^\H \A^{-1} \U)^{-1} \V^\H \A^{-1}.
    \end{equation*}
\end{Lemma}

\begin{Lemma}\label{lem:ABA=ABA}
    For $\A \in \CC^{p \times n}$ and $\B \in \CC^{n \times p}$, we have
    \begin{equation*}
        \A(\B \A - z \I_n)^{-1} = (\A \B - z\I_p)^{-1} \A,
    \end{equation*}
    for $z \in \CC$ distinct from zero and from the eigenvalues of $\A \B$.
\end{Lemma}

\begin{Lemma}[Large-dimensional approximations involving steering matrix]\label{lem:AJJA}
Under the large-dimensional setting in \Cref{ass:large_array}, we have, as $N,n,T \to \infty$ at the same pace:
\begin{enumerate}
    \item in the case of widely-spaced DoAs in \Cref{ass:widely_spaced} that, 
    \begin{equation}
        \| \A^\H \A - \I_K \| = O(N^{-1}),
    \end{equation}
    and
    \begin{equation}
        \begin{aligned}\relax
            &\| \A^\H  \J_1^\H \J_1 \A - \tau \I_K \| = O(N^{-1})\\
            &\| \A^\H \J_1^\H \J_2 \A - \tau \diag\{ e^{\imath \Delta \theta_i} \}_{i=1}^K \| = O(N^{-1}),
        \end{aligned}
    \end{equation}
    with $\tau = \lim n/N \in (0,1)$,
    so that the steering matrix $\A$ is (approximately for $N$ large) the same as $\U_K$, the top-$K$ subspace of $\C$, and that both $\A^\H  \J_1^\H \J_1 \A$ and $\A^\H  \J_1^\H \J_2 \A$ are asymptotically diagonal; and
    \item in the case of closely-spaced DoAs in \Cref{ass:closely_spaced} with $K=2$ sources such that $\theta_2=\theta_1+\alpha/N$ for $ \alpha>0$, 
    \begin{equation*}
        \A^\H \A = \begin{bmatrix}
            1 & e^{\imath \alpha /2} \sinc(\frac\alpha2) \\
            e^{-\imath \alpha /2} \sinc(\frac\alpha2) & 1
        \end{bmatrix}  + O_{\| \cdot \|}(N^{-1}),
    \end{equation*}
    where $\sinc(t) = \sin(t)/t$, and
    \begin{equation}\label{eq:AJJA2}
        \begin{split}
        \A^\H \J_1^\H \J_1 \A &=  \begin{bmatrix}
            \tau & \frac{1-e^{\imath\alpha \tau}}{-\imath \alpha} \\
            \frac{1-e^{-\imath\alpha \tau}}{\imath \alpha} & \tau
        \end{bmatrix} + O_{\| \cdot \|}(N^{-1}),  \\
        \A^\H \J_1^\H \J_2 \A &= \begin{bmatrix}
            \tau e^{\imath \Delta \theta_1} &\frac{1-e^{\imath\alpha \tau}}{-\imath \alpha} e^{\imath \Delta \theta_2} \\
            \frac{1-e^{-\imath\alpha \tau}}{\imath \alpha} e^{\imath \Delta \theta_1} & \tau e^{\imath \Delta \theta_2}
        \end{bmatrix} + O_{\| \cdot \|}(N^{-1}),
    \end{split} 
    \end{equation}
    so that the steering vectors are no longer orthogonal, and $\A^\H \J_1^\H \J_1 \A$ and $\A^\H \J_1^\H \J_2 \A$ are no long diagonal.
    In particular, it can be checked that in this setting asymptotically as $N,T \to \infty$, $\A^\H \A$ admits $1 \pm |\sinc\frac{\alpha }{2} |$ as eigenvalues, with corresponding eigenvectors 
    \begin{equation}\label{eq:def_V}
        \vv_\pm = [e^{\frac{\imath \alpha}{2}}, \pm 1]^\T/\sqrt 2.
    \end{equation}
    In particular, the approximation errors of both eigenvalues and eigenvectors (in an Euclidean norm sense) are of order $O(N^{-1})$, by Weyl's inequality and Davis--Kahan theorem, respectively.
\end{enumerate}
\end{Lemma}

\begin{proof}[Proof of \Cref{lem:AJJA}]
In the case of widely-spaced DoAs under \Cref{ass:widely_spaced}, we have, per the definition of $\mathbf{a}(\theta_i)$ in \eqref{eq:def_ULA} and of $\J_1, \J_2$ in \eqref{eq:def_J}, that
\begin{align*}
    &\mathbf{a}(\theta_i)^\H \mathbf{a}(\theta_j) = \frac1N \sum_{k=1}^N e^{ \imath (k-1) (\theta_j - \theta_i)  } = \begin{cases}
        \frac{n}N, & i = j, \\ 
        \frac{ e^{\imath (N-1) (\theta_j - \theta_i)} }N \frac{ 1- e^{\imath N (\theta_j - \theta_i)} }{ 1 - e^{\imath (\theta_j - \theta_i)} } = O(N^{-1}), & i \neq j;
    \end{cases}
\end{align*}
and similarly
\begin{align*}
    \mathbf{a}(\theta_i)^\H \J_1^\H \J_1 \mathbf{a}(\theta_j) &=\frac{1}{N}\sum_{k=\ell}^{\ell+n-1} e^{\imath (k-1)(\theta_j-\theta_i)} = \begin{cases}
        \frac{n}N, & i = j, \\ 
        O(N^{-1}), & i \neq j; 
    \end{cases}\\
    \mathbf{a}(\theta_i)^\H \J_1^\H \J_2 \mathbf{a}(\theta_j) &= \frac{1}{N}\sum_{k=\ell}^{\ell+n-1} e^{\imath\Delta \theta_j} e^{\imath (k-1)(\theta_j-\theta_i)} = \begin{cases}
        \frac{n}N e^{-\imath \Delta \theta_i}, & i = j; \\ 
        O(N^{-1}), & i \neq j.
    \end{cases}
\end{align*}
This concludes the proof of the first item of \Cref{lem:AJJA}.

In the case of closely-spaced DoAs in \Cref{ass:closely_spaced} with $K=2$ sources such that $\theta_2=\theta_1+\alpha/N$ for $ \alpha>0$, we have 
\begin{align*}
     &\mathbf{a}(\theta_1)^\H \J_1^\H \J_1 \mathbf{a}(\theta_2) =\frac{1}{N}\sum_{k=\ell}^{\ell+n-1} e^{\imath (k - 1)(\theta_2-\theta_1)} = \frac{e^{\imath (\ell-1) \alpha/N}(1-e^{\imath \alpha n/N})}{N(1-e^{\imath \alpha/N})} =  \frac{1-e^{\imath\alpha \tau}}{-\imath \alpha} + O(N^{-1}),
\end{align*}
by Taylor expansion with $\tau = \lim n/N$, and similarly
\begin{align*}
    &\mathbf{a}(\theta_1)^\H \J_1^\H \J_2 \mathbf{a}(\theta_2) =\frac{1-e^{\imath\alpha \tau}}{-\imath \alpha}e^{\imath\Delta \theta_2} + O(N^{-1}).
\end{align*}
This concludes the proof of the second item of \Cref{lem:AJJA}.
\end{proof}

\begin{Example}[Circle decomposition of permutation]\label{example:circle_decom}
For example, the permutation written in two-line notation as
\begin{align*}
    \sigma = \begin{bmatrix}
        1&4&6&8&3&9&7&2&5\\
        4&6&8&3&9&1&2&7&5
    \end{bmatrix}
\end{align*}
has the decomposition of one 6-node cycle, one 2-node cycle, and a 1-node cycle. 
Its cycle diagram is shown below.
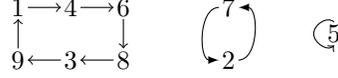
\begin{center}
    \begin{tikzpicture}[scale=0.7]
        \def\step{1};
        \node[inner sep=1pt] (1) at (0,\step) {$1$};
        \node[inner sep=1pt] (4) at (\step,\step) {$4$};
        \node[inner sep=1pt] (6) at (2*\step,\step) {$6$};
        \node[inner sep=1pt] (8) at (2*\step,0) {$8$};
        \node[inner sep=1pt] (3) at (1*\step,0) {$3$};
        \node[inner sep=1pt] (9) at (0,0) {$9$};
        \draw[->] (1) -- (4);
        \draw[->] (4) -- (6);
        \draw[->] (6) -- (8);
        \draw[->] (8) -- (3);
        \draw[->] (3) -- (9);
        \draw[->] (9) -- (1);
        \node[inner sep=1pt] (7) at (4*\step,\step) {$7$};
        \node[inner sep=1pt] (2) at (4*\step,0) {$2$};
        \draw[-latex] (7.180) to[out = 180,in=180]  (2.180);
        \draw[-latex] (2.0) to[out = 0,in=0]  (7.0);
        \node[inner sep=1pt] (5) at (6*\step,0.5*\step) {$5$};
        \draw[->] (5.north) to[out=135,in=90]([xshift=-0.2cm]5.west) to[out=-90,in=-135] (5.south);
    \end{tikzpicture}
    \end{center}
After deleting node 4 in the first row and deleting node 6 in the second row, the permutation becomes
\begin{align*}
    \sigma = \begin{bmatrix}
        1&6&8&3&9&7&2&5\\
        4&8&3&9&1&2&7&5
    \end{bmatrix}
\end{align*}
which can be decomposed into a path from node 6 to node 4 and several disjoint cycles, as shown below.
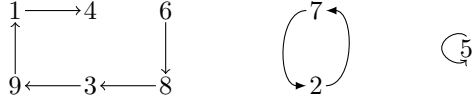
\begin{center}
\begin{tikzpicture}
    \def\step{1};
        \node[inner sep=1pt] (1) at (0,\step) {$1$};
        \node[inner sep=1pt] (4) at (\step,\step) {$4$};
        \node[inner sep=1pt] (6) at (2*\step,\step) {$6$};
        \node[inner sep=1pt] (8) at (2*\step,0) {$8$};
        \node[inner sep=1pt] (3) at (1*\step,0) {$3$};
        \node[inner sep=1pt] (9) at (0,0) {$9$};
        \draw[->] (1) -- (4);
        \draw[->] (6) -- (8);
        \draw[->] (8) -- (3);
        \draw[->] (3) -- (9);
        \draw[->] (9) -- (1);
        \node[inner sep=1pt] (7) at (4*\step,\step) {$7$};
        \node[inner sep=1pt] (2) at (4*\step,0) {$2$};
        \draw[-latex] (7.180) to[out = 180,in=180]  (2.180);
        \draw[-latex] (2.0) to[out = 0,in=0]  (7.0);
        \node[inner sep=1pt] (5) at (6*\step,0.5*\step) {$5$};
        \draw[->] (5.north) to[out=135,in=90]([xshift=-0.2cm]5.west) to[out=-90,in=-135] (5.south);
\end{tikzpicture}
\end{center}
\end{Example}

\section{Deterministic Equivalent for Resolvent}
\label{sec:DE}

In this section, we present the Deterministic Equivalent as a convenient technical tool to assess the asymptotic behavior of (eigenspectral) scalar observations of large random matrices. 
We refer the interested readers to \cite[Chapter~2]{couillet2022RMT4ML} for a review.

\begin{Definition}[Deterministic Equivalent]\label{def:DE}
For Hermitian random matrix $\Q \in \mathbb{C}^{N \times N}$, we say a deterministic matrix $\bar{\Q} \in \mathbb{C}^{N \times N} $ is a Deterministic Equivalent of $\Q$ and denote
\begin{equation}\label{eq:def_DE}
    \Q \leftrightarrow \bar \Q,
\end{equation}
if for all deterministic matrices $\A  \in \mathbb{C}^{N \times N}$ and vectors $\mathbf{a}, \mathbf{b}  \in \mathbb{C}^N$ of unit spectral and Euclidean norm, respectively, we have
\begin{align*}
    \frac{1}{N} \tr(\A(\Q-\bar{\Q})) \to 0, \quad \mathbf{a}^\H (\Q-\bar{\Q}) \mathbf{b}\to 0,
\end{align*}
almost surely as $N \to \infty$.
\end{Definition}

\begin{Lemma}[First-~and~second-order Deterministic Equivalents for resolvent,~{\cite[Theorem~2.4]{couillet2022RMT4ML}}]\label{lem:DE}
For random matrix $\Z \in \CC^{N \times T}$ having i.i.d.\@ $\mathcal{CN}(0, \sigma^2)$ entries, $z_1, z_2 \in \CC$ not eigenvalue of $\Z \Z^\H/T$ and deterministic matrix $\mathbf{B} \in \CC^{N \times N}$ of bounded spectral norm, then, for the resolvent $\Q(z) = (\Z \Z^\H/T - z \I_N)^{-1}$, the following deterministic equivalents hold
\begin{align*}
        \Q(z_1) &\leftrightarrow m(z_1) \I_N, \\
        \Q(z_1) \mathbf{B} \Q(z_2) &\leftrightarrow m(z_1) m(z_2)  \mathbf{B} + \eta(z_1, z_2) \frac1T \tr (\mathbf{B})  \I_N,
\end{align*}
with
\begin{equation}\label{eq:def_eta}
    \eta(z_1, z_2) = \frac{\sigma^4 m^2(z_1) m^2(z_2) }{ (1+\sigma^2cm(z_1)) (1+\sigma^2cm(z_2)) - c \sigma^4 m(z_1) m(z_2) },
\end{equation}
and $m(z)$ the unique Stieltjes transform solution to the Mar\u{c}enko-Pastur equation~\cite{marvcenko1967distribution} as defined in \eqref{eq:m(z)} of \Cref{theo:spikes}.
In particular, for $z_1 = z_2 = z$, we obtain $\eta(z,z) = \frac{\sigma^4 m'(z) m^2(z)}{(1+c \sigma^2 m(z))^2}$ with $m'(z)=\frac{m^2(z)}{1-c\sigma^4 m^2(z)/(1+c\sigma^2m(z))^2}$.
\end{Lemma}

\begin{Lemma}[Further Deterministic Equivalent results]\label{lem:E_Omega}
For random matrix $\Z \in \CC^{N \times T}$ having i.i.d.\@ $\mathcal{CN}(0,\sigma^2)$ entries, its resolvent $\Q(z) = (\Z \Z^\H/T - z \I_N)^{-1}$, and matrices $\J_1$, $\J_2$ defined in \eqref{eq:def_J}, the following Deterministic Equivalents hold
\begin{align*}
    \frac{1}{T} \Z^\H \Q(z_1) \J_1^\T \J_2 \Q(z_2) \Z & \leftrightarrow \mathbf{0}_T, \\
    \frac{1}{T} \Z^\H \Q(z_1) \J_1^\T \J_1 \Q(z_2) \Z &\leftrightarrow \gamma(z_1,z_2)  \I_T,
\end{align*}
where
\begin{equation}\label{eq:def_gamma}
    \gamma(z_1,z_2)= c \tau \sigma^2 \times \frac{m(z_1)m(z_2)+c\eta(z_1,z_2)}{(1+\sigma^2cm(z_1))(1+\sigma^2cm(z_2))},
\end{equation}
for $c = \lim N/T$, $\tau = \lim n/N$, $\eta(z_1, z_2)$ defined in \eqref{eq:def_eta} and $m(z)$ defined in \eqref{eq:m(z)}.
\end{Lemma}

\begin{proof}[Proof of \Cref{lem:E_Omega}]
The proof of Deterministic Equivalents generally comes in two steps: 
\begin{enumerate}
    \item approximation (in a spectral norm sense) of the expectation of the random matrix model of interest; and
    \item concentration of trace and bilinear norms as in \Cref{def:DE} around the corresponding expectations.
\end{enumerate}
Here, we provide detailed derivation of the first step for the results in \Cref{lem:E_Omega}, the second concentration step is rather standard, see \cite[Chapter~2]{couillet2022RMT4ML}.

We evaluate the expectation $\frac{1}{T} \EE[\Z^\H \Q(z_1) \J_1^\T \J_1 \Q(z_2) \Z]$, and that of $\frac{1}{T} \Z^\H \Q(z_1) \J_1^\T \J_2 \Q(z_2) \Z$ can be derived similarly.

Consider the $(i,i)$th diagonal entry of the expectation $\frac1T \EE[\Z^\H \Q(z_1) \J_1^\T \J_1 \Q(z_2) \Z]$ as
\begin{align*}
    &\frac{1}{T} \EE[\z_i^\H \Q(z_1) \J_1^\T \J_1 \Q(z_2) \z_i] = \frac{1}{T} \EE\left[\frac{\z_i^\H \Q_{-i}(z_1) \J_1^\T \J_1 \Q_{-i}(z_2) \z_i}{(1+\frac{1}{T}\tr\Q(z_1))(1+\frac{1}{T}\tr\Q(z_2))} \right] + o(1) \\
    &= \frac{\sigma^2 \tr(\Q_{-i}(z_1)\J_1^\T \J_1\Q_{-i}(z_2))/T}{(1+\frac{1}{T}\tr\Q(z_1))(1+\frac{1}{T}\tr\Q(z_2))} + o(1) = \frac{\sigma^2(\frac n T m(z_1)m(z_2)+c\eta(z_1,z_2)\frac n T)}{(1+c \sigma^2 m(z_1))(1+c \sigma^2 m(z_2))} + o(1) = \gamma(z_1,z_2) + o(1),
\end{align*}
where we used in the second line the Woodbury identity in \Cref{lem:Woodbury} (the rank-one case is known as the Sherman--Morrison formula) to write
\begin{align*}
    \Q\z_i = \frac{\Q_{-i}\z_i}{1+\frac{1}{T}\z_i^\H \Q_{-i}\z_i} =  \frac{\Q_{-i}\z_i}{1+\frac{1}{T}\tr\Q} + o(1),
\end{align*}
for $\Q_{-i}=(\frac 1 T \sum_{j \neq i}\x_j \x_j^\H-z\I_N )^{-1}$ independent of $\x_i$ so that $\| \Q_{-i} - \Q \| = O(N^{-1})$, and then \Cref{lem:DE} in the last line.

For off-diagonal entries of $\frac{1}{T} \Z^\H \Q(z_1) \J_1^\T \J_1 \Q(z_2) \Z$, we write, for $i \neq j$ that by Lemmas~\ref{lem:Woodbury}~and~\ref{lem:DE} that
\begin{align*}
    & \frac1T \z_i^\H \Q(z_1) \J_1^\T \J_1 \Q(z_2) \z_j = \frac1T \frac{\z_i^\H \Q_{-i}(z_1) \J_1^\T \J_1 \Q_{-j}(z_2) \z_j}{(1+\frac{1}{T}\tr\Q(z_1))(1+\frac{1}{T}\tr\Q(z_2))} + o(1) = \frac1T \frac{\z_i^\H \Q_{-i}(z_1) \J_1^\T \J_1 \Q_{-j}(z_2) \z_j}{(1+c \sigma^2 m(z_1))(1+c \sigma^2 m(z_2))} + o(1).
\end{align*}
Note that $\x_j$ still depends on $\Q_{-i}$, so we further write $\Q_{-i}=(\Q_{-ij}+\frac{1}{T}\z_j\z_j^\H)^{-1}=\Q_{-ij}-\frac{\Q_{-ij}\frac{1}{T}\z_j\z_j^\H\Q_{-ij}}{1+\frac{1}{T}\z_j^\H \Q_{-ij}\z_j}$ for $\Q_{-ij}$ that is independent of both $\x_i$ and $\x_j$.
Then,
\begin{align*}
    \frac1T \z_i^\H \Q_{-i}(z_1) \J_1^\T \J_1 \Q_{-j}(z_2) \z_j 
    &= \frac1T \z_i^\H \left(\Q_{-ij}(z_1)-\frac{\Q_{-ij}(z_1)\frac{1}{T}\z_j\z_j^\H\Q_{-ij}(z_1)}{ 1 + c \sigma^2 m(z_1) } \right) \\
    & \times \J_1^\T \J_1 \left(\Q_{-ij}(z_2) -\frac{\Q_{-ij}(z_2)\frac{1}{T}\z_i\z_i^\H\Q_{-ij}(z_2)}{1+c \sigma^2m(z_2)} \right) \z_j + o(1),
\end{align*}
the expectation of which is zero for $i \neq j$ by independence $\z_i, \z_j$ and $\Q_{-ij}$.
This concludes the proof of \Cref{lem:E_Omega}.
\end{proof}

\section{Mathematical Proofs}
\label{sec:proofs}

In this section, we present the proofs of our technical results of \Cref{rem:wide_correlated}, \Cref{rem:closely_equal} and \Cref{theo:CDT} in \Cref{subsec:proof_of_rem_wide_correlated}, \Cref{subsec:proof_rem_closely_equal} and \Cref{subsec:proof_of_main}, respectively.

\subsection{Proof of \Cref{rem:wide_correlated}}
\label{subsec:proof_of_rem_wide_correlated}

To prove \Cref{rem:wide_correlated}, we propose to check, in the case of $K = 2$ sources with widely-spaced (see \Cref{ass:widely_spaced}) DoAs $\theta_1 \neq \theta_2 \in (-\pi/2,\pi/2)$, that the classical ESPRIT estimates $\hat \theta_1, \hat \theta_2$ from \Cref{alg:TraESPRIT} \emph{cannot} be consistent \emph{unless} $\U_\P = \I_2$, that is, when the two sources are \emph{uncorrelated}.

To this end, recall from \Cref{rem:wide_correlated} that for $\U_{\P} = [\begin{smallmatrix} u_{11} & u_{12} \\ u_{21} & u_{22} \end{smallmatrix}] \in \CC^{2 \times 2}$ the eigenvectors of $\P$, it follows from \Cref{lem:AJJA} that 
\begin{align*}
    \bPhi_1 &= \U_\P^\H \A^\H \J_1^\H \J_1 \A \U_\P + O_{\| \cdot \|}(N^{-1}) = \tau \I_2 + O_{\| \cdot \|}(N^{-1}),\\
    \bPhi_2 &= \U_\P^\H \A^\H \J_1^\H \J_2 \A \U_\P + O_{\| \cdot \|}(N^{-1}) =\tau \U_{\P}^\H \diag\{ e^{\imath \Delta \theta_i} \}_{i=1}^K \U_{\P} + O_{\lVert \cdot \rVert}(N^{-1}),
\end{align*}
so that by \Cref{theo:main},
\begin{align*}
    \bar{ \mathbf{{\Phi}}} &= \diag(\sqrt{\mathbf{g}})\U_{\P}^\H  \diag\{e^{\imath \Delta \theta_k}\}_{k=1}^{2} \U_{\P} \diag(\sqrt{\mathbf{g}}) + O_{\| \cdot \|}(N^{-1}) \\ 
    &= \Big[\begin{smallmatrix} g_1 (\lvert u_{11} \rvert ^2 e^{\imath \Delta \theta_1} + \lvert u_{21}\rvert^2 e^{\imath \Delta \theta_2} )  & \sqrt{g_1 g_2} ( u_{11}^* u_{12} e^{\imath \Delta \theta_1} + u_{21}^* u_{22} e^{\imath \Delta \theta_2} ) \\ \sqrt{g_1 g_2} ( u_{12}^* u_{11} e^{\imath \Delta \theta_1} + u_{22}^* u_{21} e^{\imath \Delta \theta_2} ) & g_2 (\lvert u_{12}\rvert ^2 e^{\imath \Delta \theta_1} +\lvert u_{22}\rvert^2 e^{\imath \Delta \theta_2} ) \end{smallmatrix}\Big] + O_{\| \cdot \|}(N^{-1}).
\end{align*}
As such, for any $a \in \RR$, $\lambda_1 = a e^{\imath \Delta \theta_1}$ is an eigenvalue of $\bar{\mathbf{\Phi}}$ if and only if it satisfies asymptotically the following quadratic equation
\begin{align} 
\label{eq:characpoly_1}
    \lambda^2 - \tr(\bar{\bPhi}) \lambda + \det(\bar{\bPhi}) = \varepsilon, \quad |\varepsilon| =  O(N^{-1}).
\end{align}
This can be further written as
\begin{equation*}
    \begin{aligned}
        &\left(\lambda -(g_1 \lvert u_{11} \rvert ^2 + g_2 \lvert u_{12} \rvert ^2) e^{\imath \Delta \theta_1} \right)\left(\lambda - (g_1 \lvert u_{21} \rvert ^2 + g_2 \lvert u_{22} \rvert ^2) e^{ \imath \Delta \theta_2}\right) =  (g_1 -g_2)^2 \lvert u_{12} \rvert ^2  \lvert u_{22} \rvert ^2   e^{\imath \Delta \theta_1}
         e^{\imath \Delta \theta_2} + \varepsilon,
    \end{aligned}
\end{equation*}
with $g_1, g_2$ defined in \eqref{eq:def_g} of \Cref{theo:main}.
By substituting $\lambda = ae^{\imath \Delta \theta_1}$,
 we obtain that the equation holds if the corresponding real and imagery parts satisfy
\begin{equation}\label{eq:real_and_imag_in_proof}
    \begin{aligned}
    & \left(a-(g_1 \lvert u_{11} \rvert ^2 + g_2 \lvert u_{12} \rvert ^2)\right)a \cos(2 \Delta \theta_1) - \Re[\varepsilon] \\ 
    &= \left( \left(a-(g_1 \lvert u_{11} \rvert ^2 + g_2 \lvert u_{12} \rvert ^2)\right)(g_1 \lvert u_{21} \rvert^2 + g_2 \lvert u_{22} \rvert ^2) + (g_1 -g_2)^2 \lvert u_{12} \rvert ^2  \lvert u_{22} \rvert ^2 \right) \cos( \Delta (\theta_1 + \theta_2) ), \\
    & \left(a-(g_1 \lvert u_{11} \rvert ^2 + g_2 \lvert u_{12} \rvert ^2)\right)a \sin(2 \Delta \theta_1) - \Im[\varepsilon] \\ 
    &= \left( \left(a-(g_1 \lvert u_{11} \rvert ^2 + g_2 \lvert u_{12} \rvert ^2)\right)(g_1 \lvert u_{21} \rvert^2 + g_2 \lvert u_{22} \rvert ^2) + (g_1 -g_2)^2 \lvert u_{12} \rvert ^2  \lvert u_{22} \rvert ^2 \right) \sin( \Delta (\theta_1 + \theta_2) ).
\end{aligned}
\end{equation}
Therefore, for any $\theta_1 \neq \theta_2$ such that $\Delta \theta_1 \neq \Delta \theta_2 + m \pi$, $a > 0$ and $m,\Delta \in \mathbb{N}$, \eqref{eq:characpoly_1} holds \emph{if and only if} both equations in \eqref{eq:real_and_imag_in_proof} are satisfied.
It can be checked this is possible \emph{only} when $(g_1 -g_2)^2 \lvert u_{12} \rvert ^2  \lvert u_{22} \rvert ^2 = O(N^{-1})$.
Recall from \eqref{eq:def_g} and \Cref{ass:subspace} that $g_1 \neq g_2 > 0$, so that one must have $|u_{12}| |u_{22}| = O(N^{-1/2})$.
A similar conclusion can drawn by considering $\lambda_2 = b e^{\imath \Delta \theta_2}$, and one can thus conclude that classical ESPRIT estimates $\hat \theta_1, \hat \theta_2$ from \Cref{alg:TraESPRIT} \emph{cannot} be consistent \emph{unless} $\U_\P = \I_2$, that is, when the two sources are \emph{uncorrelated}.
This concludes the proof of \Cref{rem:wide_correlated}.

\subsection{Proof of \Cref{rem:closely_equal}}
\label{subsec:proof_rem_closely_equal}

To prove \Cref{rem:closely_equal}, we propose to check, in the case of $K=2$ sources with $\theta_2=\theta_1+\alpha/N$ for some $\alpha >0$, and uncorrelated signals with equal power (i.e., $\P=\I_2$), that the estimated DoAs $\hat{\theta}_1, \hat{\theta}_2$ from classical ESPRIT are \emph{not} $N$-inconsistent (see \Cref{coro:N_consistency}) as $N,T \to \infty$.

In this case, it follows from the second item of \Cref{lem:AJJA} that the nonzero eigenvalues of $\A \P \A^\H = \A \A^\H$ are the same as those of $\A^\H \A$ and are  (asymptotically up to an error of order $O(N^{-1})$) given by $1 \pm |\sinc(\alpha/2) |$, so that the subspace separation condition in \Cref{ass:subspace} becomes
\begin{equation}
    |\sinc(\alpha/2)| < 1- \sigma^2 \sqrt c,
\end{equation}
and we will be working in this setting.
Note in particular that we \emph{must} have $c \in (0,1)$.

Recall again from the second item of \Cref{lem:AJJA} that the top two eigenvectors of $\A \P \A^\H$ are approximately given by $\A \vv_\pm =\A [e^{\frac{\imath \alpha}{2}}, \pm 1]^\T/\sqrt 2$, up to some error in Euclidean norm of order $O(N^{-1})$.
We thus have, 
\begin{align*}
    \bPhi_1 &= \begin{bmatrix} \tau + \Re[e^{-\imath \alpha/2} \beta_1] & -\imath \Im[e^{-\imath \alpha/2} \beta_1] \\
        \imath \Im[e^{-\imath \alpha/2} \beta_1] & \tau - \Re[e^{-\imath \alpha/2} \beta_1] \end{bmatrix} + O_{\| \cdot \|}(N^{-1}),
\end{align*}
and 
\begin{align*}
    [\bPhi_2]_{ij} = &\tau e^{\imath \Delta \theta_1} + (-1)^{i+1}\beta_2 e^{\imath \alpha/2} e^{\imath \Delta \theta_1} + (-1)^{j+1} \beta_1 e^{-\imath \alpha/2} e^{\imath \Delta \theta_2} + (-1)^{i+j}\tau e^{\imath \Delta \theta_2}  + O(N^{-1}),
\end{align*}
with the shortcuts $\beta_1 = \frac{1-e^{\imath\alpha \tau}}{-\imath \alpha} $, $\beta_2 =\beta_1^{*}$.

It follows from \Cref{theo:main} that the DoA estimates given by classical ESPRIT are asymptotically given by the angles of the complex eigenvalues of $\bar \bPhi = \bar\bPhi^{-1} \bar \bPhi_2$ defined in \eqref{eq:def_bar_Phi}, the determinant of which is given by
\begin{align*}
    \det(\bar{\bPhi}) =& \nu_1 e^{\imath \Delta \theta_1} e^{\imath \Delta \theta_2}-\nu_2 +\varepsilon, \nu_1, \nu_2 \in \mathbb{R},|\varepsilon| =  O(N^{-1}),
\end{align*}
where $g_i$ are defined in \eqref{eq:def_g}, $h_i = 1 - g_1$, and $\nu_1, \nu_2$  defined as
    \begin{align*}
        \nu_1 &= g_1^2 g_2^2 ( \tau^2-\beta_1 \beta_2)^2 +(1-g_1)(1-g_2)g_1g_2  \tau^4 +g_1g_2 \tau^2 \left[ \tau^2 (g_1+g_2-2g_1g_2)+\frac{\tau(g_1-g_2)}{2}(e^{\frac{\imath \alpha}{2}}\beta_2+e^{-\frac{\imath \alpha}{2}}\beta_1) \right] ,\\
        \nu_2 &=\frac{g_1g_2\beta_1\beta_2 \tau}{\det^2(\bar{\bPhi}_1)} \left[\tau (1-g_1)(1-g_2)+\tau (g_1+g_2-2g_1g_2)+\frac12 (g_1-g_2) (e^{\frac{\imath \alpha }{2}}\beta_2+e^{-\frac{\imath \alpha }{2}}\beta_1) \right].
    \end{align*}

We now prove \Cref{rem:closely_equal} by contradiction: assume that $\lambda_1=a e^{\imath \Delta \theta_1}, \lambda_2 = b e^{\imath \Delta \theta_2}, a, b \in \mathbb{R}$ are two eigenvalues of $\bar{\bPhi}$, satisfying $\theta_2 = \theta_1 +\alpha/N$, then we must have asymptotically
\begin{align}
    \label{eq:det_function}
    \lambda_1 \lambda_2 = ab e^{\imath \Delta \theta_1}e^{\imath \Delta \theta_2} = \nu_1 e^{\imath \Delta \theta_1} e^{\imath \Delta \theta_2}-\nu_2 +\varepsilon,
\end{align}
for some $\varepsilon \in \RR$ such that $|\varepsilon| = O(N^{-1})$.
This can be further simplified as
\begin{align*}
    &(ab-\nu_1)[\sin(2\Delta\theta_1)\cos(\frac{\Delta \alpha}{N}) +\cos(2\Delta\theta_1)\sin(\frac{\Delta \alpha}{N}) ] = \Im[\varepsilon], \\
    &(ab-\nu_1)[\cos(2\Delta\theta_1)\cos(\frac{\Delta \alpha}{N}) -\sin(2\Delta\theta_1)\sin(\frac{\Delta \alpha}{N}) ] +\nu_2 = \Re[\varepsilon].
\end{align*}
The above equation holds if and only if
\begin{align}
    \nu_1 = ab + O(N^{-1}), \quad \nu_2 = O(N^{-1}).
\end{align}

Let us now focus on the term $\nu_2$. 
For any given $\alpha > 0$, $c \in (0, \infty)$, and $\tau = \lim n /N \in (0,1)$, note that $\beta_1 \beta_2 = | \frac{1-e^{\imath\alpha \tau}}{-\imath \alpha} |^2 \geq 0$ (with equality if and only if $\cos(\alpha \tau) = 1$) so that $\nu_2 > 0$ if and only if
\begin{align}
    \vspace{-2pt}
        \!\!\kappa(\alpha,\tau,c) \equiv & (\ell_1\!-\!\ell_2)\big(\frac{2}{\sigma^2}\!+\!c\!+\!\ell_1 \ell_2\big)\big(\! \sinc \frac{\alpha}2 \!-\! \frac{2}{\alpha} \sin (\frac{\alpha}{2}\!-\!\alpha \tau ) \big) \notag\\
        &+ 2 \tau \Big(\frac{4}{\sigma^2}+(\frac{2}{\sigma^2}-2)\ell_1\ell_2+c(\ell_1\ell_2-1) \Big), \label{eq:kappa_function}
        \vspace{-2pt}
\end{align}
by substituting $g_i$ defined in \eqref{eq:def_g}. 
Note that the (partial) derivative of $\kappa(\alpha,\tau,c)$ with respect to $\tau$ is given by 
\begin{align*}
    \vspace{-2pt}
    \frac{\partial}{\partial \tau} \kappa(\alpha,\tau,c) &= 2(\ell_1\!-\!\ell_2)\big(\frac{2}{\sigma^2}\!+\!c\!+\!\ell_1 \ell_2\big)\cos \left(\frac{\alpha}{2}-\alpha \tau \right) \\
    & + 2 \Big(\frac{4}{\sigma^2}+(\frac{2}{\sigma^2}-2)\ell_1\ell_2+c(\ell_1\ell_2-1) \Big).
    \vspace{-2pt} 
\end{align*}
For any \(\alpha > 0 \) and \( \tau \in (0,1) \), it can be observed that \( \frac{\partial}{\partial \tau} \kappa(\alpha,\tau,c)\) is linear in \( c \in (0,1) \), with
\begin{align*}
    &\lim_{c \to 0} \frac{\partial}{\partial \tau} \kappa(\alpha,\tau,c) > 0, \quad \lim_{c \to 1} \frac{\partial}{\partial \tau} \kappa(\alpha,\tau,c )  >0.
\end{align*}
where we recall that $ |\sinc(\alpha/2)| < 1- \sigma^2 \sqrt c$.
As such, we have that $\frac{\partial}{\partial \tau} \kappa(\alpha,\tau,c ) > 0$ for all $c \in (0,1)$, and \( \kappa(\alpha,\tau,c) \) is thus an increasing function of \( \tau \), so that 
\begin{equation}
    \kappa(\alpha,\tau,c) > \kappa(\alpha, 0 ,c) = 0,
\end{equation}
We thus conclude that the traditional ESPRIT is not $N$-consistent in the case of closely-spaced DoAs with equal power sources.
This concludes the proof of \Cref{rem:closely_equal}.

\subsection{Proof of \Cref{theo:CDT}}
\label{subsec:proof_CDT}

Here we provide the detailed proof of \Cref{theo:CDT}.
We first consider the diagonal entries and then the off-diagonal entries of $\hat{\bPhi}_1 $ and $ \hat{\bPhi}_2$.

\paragraph{Diagonal entries of $\hat{\bPhi}_1, \hat{\bPhi}_2$}
Here, we would like to show that for $\hat{\mathbf{\Phi}}_1, \hat{\mathbf{\Phi}}_2$ and $\bar{\mathbf{\Phi}}_1, \bar{\mathbf{\Phi}}_2$ defined in \eqref{eq:def_hat_Phi} and \eqref{eq:def_bar_Phi} respectively, we have 
\begin{align*}
    [\hat{\mathbf{\Phi}}_1]_{kk} - [\bar{\mathbf{\Phi}}_1]_{kk} \to 0, \quad [\hat{\mathbf{\Phi}}_2]_{kk} - [\bar{\mathbf{\Phi}}_2]_{kk} \to 0.
\end{align*}

Let us start with the diagonal entries of the asymmetric matrix $[\hat{\mathbf{\Phi}}_2]_{kk}=\hat \uu_k^\H \J_1^\T \J_2 \hat \uu_k$, for $\hat \uu_k$ the $k$th dominant eigenvector of the SCM $\hat \C = \X \X^\H/T$.
First note that under the subspace separation condition in \Cref{ass:subspace}, it follows from \Cref{theo:spikes} that the top-$K$ empirical eigenvalues $\hat \lambda_k$ of $\hat \C$ converge almost surely to different limits in the large $N,T$ limit.
We thus have, for $\Gamma_k$ a positively (i.e., counterclockwise) oriented contour circling around \emph{only} the $k$th largest eigenvalue of $\hat \C$, that
\begin{align*}
      &\hat \uu_k^\H \J_1^\H \J_2 \hat \uu_k = \tr  \left( \hat \uu_k \hat \uu_k^\H \J_1^\H \J_2  \right) = -\frac1{2 \pi \imath} \oint_{\Gamma_k} \tr ( \hat \C - z \I_N )^{-1} \J_1^\H \J_2\,dz = -\frac1{2 \pi \imath} \oint_{\Gamma_k} \tr \left( \frac1T \Z \Z^\H - z \I_N + \V \bLambda \V^\H \right)^{-1}\J_1^\H \J_2\,dz,
\end{align*}
where we used Cauchy's integral formula in the second line, and $\X = \A \S + \Z$, $\mathbf{P}=\mathbb{E}[\mathbf{s}(t)\mathbf{s}(t)^{\H}]$ the signal power matrix, as well as
\begin{equation*}
\V = \begin{bmatrix} \A & \frac1T \Z \S ^\H \end{bmatrix} \in \CC^{N \times 2K}, \, \bLambda = \begin{bmatrix} \P & \I_K \\ \I_K & \zo_K \end{bmatrix} \in \CC^{2K \times 2K},
\end{equation*}
in the third line.
We then get, by Woodbury identity in \Cref{lem:Woodbury} that
\begin{align*}
    \hat \uu_k^\H \J_1^\T \J_2 \hat \uu_k &= -\frac1{2 \pi \imath} \oint_{\Gamma_k} \tr (  \Q(z) - \Q(z) \V \left( \bLambda^{-1} + \V^\H \Q(z) \V \right)^{-1} \V^\H \Q(z) ) \J_1^\H \J_2 \,dz \\  
    &=\frac1{2 \pi \imath}   \oint_{\Gamma_k} \tr\left( \bLambda^{-1} + \V^\H \Q(z) \V \right)^{-1} \V^\H \Q(z) \J_1^\H \J_2 \Q(z) \V\,dz ,
\end{align*}
where we recall the resolvent $\Q(z)=(\Z\Z^\H/T-z\I_N)^{-1}$ as in \Cref{def:ST} and use the fact that the term $\tr (\Q(z) \J_1^\T \J_2)$ has no pole circled by $\Gamma_k$.
It is important to note that till now we have \emph{not} used any (asymptotic) approximation in the large $N,T$ limit.

We then use the deterministic equivalents for resolvent result in \Cref{lem:DE} to approximate this trace term, we start by approximating the block matrix $\left( \bLambda^{-1} + \V^\H \Q(z) \V \right)^{-1}$ as
\begin{equation}\label{eq:approx_inv}
    \left( \bLambda^{-1} + \V^\H \Q(z) \V \right)^{-1} = \begin{bmatrix}
      \mathbf{H}_1(z) & \mathbf{H}_2(z) \\ 
      \mathbf{H}_2^\H(z) & \mathbf{H}_3(z)  
    \end{bmatrix} + o_{\| \cdot \|}(1),
\end{equation}
with
\begin{equation}\label{eq:def_H}
\begin{aligned}
    \mathbf{H}_1(z) &=  \frac{z m(z) + 1}{\sigma^2m(z)} \mathbf{H}_2(z) \P, \\
    \mathbf{H}_2(z) &= \left(\I_K+\frac{1+zm(z)}{\sigma^2}\P \A^\H \A \right)^{-1}, \\ 
    \mathbf{H}_3(z) &= -m(z) \A^\H \A \mathbf{H}_2(z),
\end{aligned}
\end{equation}
and
\begin{align}
    \V^\H \Q(z_1) \J_1^\H \J_2 \Q(z_2) \V &= \begin{bmatrix}
    m(z_1) m(z_2)\A^\H \J_1^\H \J_2 \A & \mathbf{0}_{K}\\
    \mathbf{0}_{K}& \mathbf{0}_{K}
\end{bmatrix} + o_{\lVert \cdot \rVert}(1) \label{eq:approx_VQJ1J2QV}.
\end{align}
Here, we use the independence between $\mathbf{S}$ and $\mathbf{Z}$, so that the mixed terms involving both $\mathbf{S}$ and $\mathbf{Z}$ are absorbed into $o_{\lVert \cdot \rVert}(1)$ with high probability.
Using these spectral norm approximations, we obtain
\begin{align*}
    \hat \uu_k^\H \J_1^\H \J_2 \hat \uu_k &= \frac1{2 \pi \imath}  \oint_{\Gamma_k}  \tr ( m^2(z) \mathbf{H}_1(z) \A^\H \J_1^\H \J_2 \A ) \,dz +o(1)\\
    &= \frac1{2 \pi \imath}  \oint_{\Gamma_k}  \tr \left(\frac{zm(z)+1}{\sigma^2 m(z)}(\I_N + \frac{1+zm(z)}{\sigma^2} \P \A^\H \A )^{-1} \P \cdot m^2(z) \A^\H \J_1^\H \J_2 \A \right) \,dz +o(1)\\
    &= -\frac1{2 \pi \imath}  \oint_{\Gamma_k}  \tr \left(\frac{zm(z)+1}{\sigma^2 m(z)}\A\P\A^\H (\I_N + \frac{1+zm(z)}{\sigma^2} \A \P \A^\H)^{-1} \cdot m^2(z) \J_1^\H \J_2 \right) \,dz +o(1)\\
    &\equiv  -\frac1{2 \pi \imath}  \oint_{\Gamma_k}  \tr \left({\mathbf{T}_1(z)} \mathbf{T}_2(z, z) \right) dz +o(1),
\end{align*}
where we define the shortcuts $\mathbf{T}_1(z)$ and $\mathbf{T}_2(z)$ as
\begin{equation}\label{eq:def_T1_T2}
\begin{aligned}
 &\mathbf{T}_1(z) = \frac{zm(z)+1}{\sigma^2 m(z)}(\mathbf{L}^{-1} + \frac{1+zm(z)}{\sigma^2}\mathbf{I})^{-1} \\        
 &\mathbf{T}_2(z_1,z_2) = m(z_1) m(z_2) \U^\H \J_1^\H \J_2 \U,
\end{aligned}    
\end{equation}
for (asymptotic) eigendecomposition $\A \P\A^\H = \U \mathbf{L} \U^\H + o_{\| \cdot \|}(1)$ with diagonal $\mathbf{L} = \diag\{\lambda_1,\ldots,\lambda_K\} \in \RR^{K \times K}$ containing the eigenvalues and $\U=[\uu_1,\ldots,\uu_K] \in \CC^{N \times K}$ containing the associated eigenvectors, as in \Cref{ass:subspace}.
Note here that $\mathbf{T}_1(z)$ is diagonal, we have, by residue calculus, that
\begin{align*}
    &\hat \uu_k^\H \J_1^\H \J_2 \hat \uu_k = -\frac{1}{2\pi\iota}\oint_{\Gamma_k} \frac{(z m(z)+1)m(z)}{\sigma^2 \lambda_k^{-1}+(1+z m(z))} \cdot \mathbf u_k^\H\mathbf J_1^\H\mathbf J_2\mathbf u_k\,dz +o(1) \\ 
    &= - \lim_{z \to \bar \lambda_k} \frac{(z-\bar \lambda_k)(zm(z)+1) m(z)  }{\sigma^2 \lambda_k^{-1}+(1+z m(z))} \times \uu_k^\H \J_1^\H \J_2 \uu_k + o(1),\\
    &= \frac{1-c \ell_k^{-2}}{1+c \ell_k^{-1}} \uu_k^\H \J_1^\H \J_2 \uu_k + o(1) = g_k [\bPhi_2]_{kk}+o(1),
\end{align*}
for $\ell_k=\lambda_k/\sigma^2$ the SNR of $k$th source, $\bar \lambda_k$ the limiting spike in \eqref{eq:eigenvalue_bias} of \Cref{theo:spikes}, and $g_k $ defined in \eqref{eq:def_g} of \Cref{theo:main}.
This allows us to conclude that
\begin{equation}
 [\hat{\mathbf{\Phi}}_2]_{kk} - [\bar{\mathbf{\Phi}}_2]_{kk} \to 0,
\end{equation} 
almost surely as $N,T \to \infty$.

Similarly, we evaluate the diagonal entries of $\bPhi_1$ as
\begin{align*}
        &\hat \uu_k^\H \J_1^\H \J_1 \hat \uu_k =\frac1{2 \pi \imath} \oint_{\Gamma_k} \tr  (\left( \bLambda^{-1} + \V^\H \Q(z) \V \right)^{-1} \V^\H \Q(z) \J_1^\H \J_1 \Q(z) \V)\,dz ,
\end{align*}
for which we have 
\begin{align}\label{eq:approx_VQJ1J1QV}
   & \V^\H \Q(z) \J_1^\H \J_1 \Q(z) \V = \begin{bmatrix}
        m(z_1) m(z_2) \A^\H \J_1^\H \J_1 \A + c \tau \eta(z_1, z_2) \A^\H \A  & \mathbf{0}_{K}\\
    \mathbf{0}_{K}&  \gamma(z_1,z_2) \P
\end{bmatrix} + o_{\lVert \cdot \rVert}(1),  
\end{align}
using Lemmas~\ref{lem:DE}~and~\ref{lem:E_Omega}, for $\gamma(z,z)$ defined in \eqref{eq:def_gamma}.
We thus have
\begin{align*}
    &\hat \uu_k^\H \J_1^\H \J_1 \hat \uu_k = \frac1{2 \pi \imath} \oint_{\Gamma_k} \tr \left( \mathbf{H}_1(z) ( m^2(z) \A^\H \J_1^\H \J_1 \A + c \tau \eta(z, z) \A^\H \A) \right) \,dz + \frac1{2 \pi \imath} \oint_{\Gamma_k} \gamma(z,z) \tr \left( \mathbf{H}_3(z) \P \right) \,dz + o(1),
\end{align*}
for which we have
\begin{align*}
    &\frac1{2 \pi \imath}  \oint_{\Gamma_k}  \tr ( \mathbf{H}_1(z) ( m^2(z) \A^\H \J_1^\H \J_1 \A + c \tau \eta(z, z) \A^\H \A) ) dz =  -\frac1{2 \pi \imath}  \oint_{\Gamma_k}  \tr \left({\mathbf{T}_1(z)} \mathbf{T}_3(z,z) \right)  dz \\
    &= \frac{1-c \ell_k^{-2}}{1+c \ell_k^{-1}} \uu_k^\H \J_1^\H \J_1 \uu_k  + \frac{c \tau}{\ell_k^2+c\ell_k},
\end{align*}
where we define the shortcut
\begin{equation}\label{eq:def_T_3}
    \mathbf{T}_3(z_1,z_2) = m(z_1) m(z_2) \U^\H \J_1^\H \J_1 \U + c \tau \eta(z_1, z_2) \I_K,
\end{equation}
for $\eta(z_1,z_2)$ defined in \eqref{eq:def_eta}; as well as
\begin{align*}
    \frac1{2 \pi \imath}  \oint_{\Gamma_k}  \tr ( \mathbf{H}_3(z) \gamma(z,z) \P) dz &=-\frac1{2 \pi \imath}  \oint_{\Gamma_k} m(z) \gamma(z,z) \tr \left(\mathbf{L}^{-1} + \frac{1+zm(z)}{\sigma^2} \I_K \right)^{-1}\, dz  \\
    &=- \frac1{2 \pi \imath}  \oint_{\Gamma_k}  c \tau \cdot \frac{\sigma^2m(z)( m^2(z)+c\eta(z,z))}{(1+\sigma^2cm(z))^2 (\lambda_k^{-1}+\frac{1+zm(z)}{\sigma^2})} \,dz \\
    &= \frac{c \tau }{\ell_k+c}.
\end{align*}
Putting these together, we thus conclude that
\begin{align*}
    \hat \uu_k^\H \J_1^\T \J_1 \hat \uu_k &= \frac{1-c \ell_k^{-2}}{1+c \ell_k^{-1}} \uu_k^\H \J_1^\T \J_1 \uu_k + c \tau \frac{1+\ell_k^{-1}}{c+\ell_k} + o(1),\\
    &=g_k [\bPhi_1]_{kk} +h_k + o(1),
\end{align*}
for $g_k, h_k$ defined in \eqref{eq:def_g} of \Cref{theo:main}.
This allows us to conclude that
\begin{equation}
 [\hat{\mathbf{\Phi}}_1]_{kk} - [\bar{\mathbf{\Phi}}_1]_{kk} \to 0,
\end{equation} 
almost surely as $N,T \to \infty$.

\paragraph{Off-diagonal entries of $\hat{\bPhi}_1, \hat{\bPhi}_2$}


We now consider the off-diagonal entries, and in particular, those having their indices forming a circle.
Consider indices $1 \leq k_1 < \ldots < k_m \leq K$ that form a cycle of length $m$ and 
\begin{equation}\label{eq:def_M}
    \mathbf{M}_{k_j} = \J_1^\H\J_2~\mbox{or}~\J_1^\H\J_1, \quad j \in \{ 1, \ldots, m \}.
\end{equation}
Introducing the matrices $\mathbf{M}_{k_1}, \ldots, \mathbf{M}_{k_m}$ allows us to evaluate, in a unified fashion, off-diagonal entries of $\hat \bPhi_1, \hat{\bPhi}_2$ as well as their (arbitrary) products.
Note in particular that 
\begin{equation}
    \| \mathbf{M}_k \| \leq 1,
\end{equation}
and by \eqref{eq:approx_VQJ1J2QV}, \eqref{eq:approx_VQJ1J1QV} and \Cref{lem:DE} that 
\begin{equation}\label{eq:approx_VQMQV}
    \V^\H \Q(z_1) \mathbf{M}_k \Q(z_2) \V \equiv \begin{bmatrix}
    \bTheta_{k,1}(z_1,z_2) & \mathbf{0}_{K}\\
    \mathbf{0}_{K}& \bTheta_{k,2}(z_1,z_2)
\end{bmatrix} + o_{\lVert \cdot \rVert}(1),
\end{equation}
with
\begin{equation}\label{eq:def_Theta_1}
   \bTheta_{k,1}(z_1,z_2) = m(z_1) m(z_2) \A^\H \mathbf{M}_k \A + \eta(z_1, z_2) \frac1T \tr (\mathbf{M}_k) \cdot \A^\H \A,
\end{equation}
and
\begin{equation}\label{eq:def_Theta_2}
\bTheta_{k,2}(z_1,z_2) = 
    \begin{cases}
         \zo_K &\text{for}~\mathbf{M}_k = \J_1^\H\J_2, \\ 
         \gamma(z_1,z_2) \P &\text{for}~\mathbf{M}_{k_j} = \J_1^\H\J_1.
    \end{cases}
\end{equation}

Recall that the off-diagonal entries of $\hat{\bPhi}_1$ and $\hat{\bPhi}_2$ are respectively given by
\begin{equation}
    [\hat{\bPhi}_1]_{k_i k_j} = \hat \uu_{k_i}^\H \J_1^\H \J_1 \hat \uu_{k_j}, \quad [\hat{\bPhi}_2]_{k_i k_j} = \hat \uu_{k_i}^\H \J_1^\H \J_2 \hat \uu_{k_j},
\end{equation}
which can be written as the generic bilinear form involving $\mathbf{M}_{k_j} \in \CC^{N \times N}$ defined in \eqref{eq:def_M}, as
\begin{equation}\label{eq:generic_bilinear_form}
    [\hat{\bPhi}_1]_{k_i k_j}~\mbox{or}~[\hat{\bPhi}_2]_{k_i k_j} \equiv \hat \uu_{k_i}^\H \mathbf{M}_{k_j} \hat \uu_{k_j}.
\end{equation}

In the form of \eqref{eq:generic_bilinear_form}, consider now the following product $\hat{\psi}_{k_1, \ldots, k_m}$ involving the off-diagonal entries of both $\hat{\bPhi}_1$ and $\hat{\bPhi}_2$, with their indices $1 \leq k_1 < \ldots < k_m \leq K$ forming a cycle of length $m$ as
\begin{equation}\label{eq:def_hat_psi}
    \hat{\psi}_{k_1, \ldots, k_m} \equiv \hat \uu_{k_m}^\H \mathbf{M}_{k_1} \hat \uu_{k_1} \times \hat \uu_{k_1}^\H \mathbf{M}_{k_2} \hat \uu_{k_2} \times \ldots \times \hat \uu_{k_{m-1}}^\H \mathbf{M}_{k_m} \hat \uu_{k_m}.
\end{equation}
With the same arguments as for the diagonal entries, we get \eqref{eq:hat_psi_loop_1}, for $\Gamma_{k_1}, \ldots, \Gamma_{k_m}$ positively (i.e., counterclockwise) oriented contours circling around \emph{only} the $k_1, \ldots, k_m$th largest eigenvalue of $\hat \C$, respectively, where we used the approximation in \eqref{eq:approx_inv}, and the fact that $\Q(z_i)$ does not \emph{have} pole enclosed by any of the contours $\Gamma_{k_1}, \ldots \Gamma_{k_m}$.
We thus get
\begin{align}
    &\hat{\psi}_{k_1, \ldots, k_m} = \tr( \M_{k_1} \hat{\uu}_{k_1}\hat{\uu}_{k_1}^\H \times \M_{k_2} \hat{\uu}_{k_2}\hat{\uu}_{k_2}^\H \ldots \M_{k_m} \hat{\uu}_{k_m}\hat{\uu}_{k_m}^\H) \nonumber  \\ \nonumber
    &=\left(\frac{1}{2 \pi \imath} \right)^{m}\oint_{\Gamma_{k_1}}\ldots \oint_{\Gamma_{k_m}}  \tr \left( \prod_{i=1}^m \M_{k_i} \Q(z_i) \V \begin{bmatrix}
        \mathbf{H}_1(z_i) & \mathbf{H}_2(z_i) \\ 
        \mathbf{H}_2^\H(z_i) & \mathbf{H}_3(z_i)  
    \end{bmatrix} \V^\H \Q(z_i) \right) d z_1 \ldots d z_m + o(1) \\  \nonumber
    &=\left(\frac{1}{2 \pi \imath} \right)^{m}\oint_{\Gamma_{k_1}}\ldots \oint_{\Gamma_{k_m}}  \tr \left( \begin{bmatrix}
        \mathbf{H}_1(z_1) & \mathbf{H}_2(z_1) \\ 
        \mathbf{H}_2^\H(z_1) & \mathbf{H}_3(z_1)  
    \end{bmatrix} \begin{bmatrix}
        \bTheta_{k_1,1}(z_m,z_1) & \zo_K \\ 
        \zo_K & \bTheta_{k_1,2}(z_m,z_1)
    \end{bmatrix} \right.\\ \nonumber
    & \times \left. \prod_{i=1}^{m-1} 
    \begin{bmatrix}
        \mathbf{H}_1(z_{i+1}) & \mathbf{H}_2(z_{i+1}) \\ 
        \mathbf{H}_2^\H(z_{i+1}) & \mathbf{H}_3(z_{i+1})  
    \end{bmatrix} \begin{bmatrix}
        \bTheta_{k_{i+1},1}(z_i,z_{i+1}) & \zo_K \\ 
        \zo_K & \bTheta_{k_{i+1},2}(z_i,z_{i+1})
    \end{bmatrix}  \right) d z_1 \ldots d z_m + o(1), \\ \nonumber
        &=\left(\frac{1}{2 \pi \imath} \right)^{m}\oint_{\Gamma_{k_1}}\ldots \oint_{\Gamma_{k_m}}  \tr \left( \begin{bmatrix}
        \mathbf{H}_1(z_1) \bTheta_{k_1,1}(z_m,z_1) & \mathbf{H}_2(z_1) \bTheta_{k_1,2}(z_m,z_1) \\ 
        \mathbf{H}_2^\H(z_1) \bTheta_{k_1,1}(z_m,z_1) & \mathbf{H}_3(z_1) \bTheta_{k_1,2}(z_m,z_1)
    \end{bmatrix} \right.\\
    & \left. \times \prod_{i=1}^{m-1} \begin{bmatrix}
        \mathbf{H}_1(z_{i+1}) \bTheta_{k_{i+1},1}(z_i,z_{i+1}) & \mathbf{H}_2(z_{i+1}) \bTheta_{k_{i+1},2}(z_i,z_{i+1}) \\ 
        \mathbf{H}_2^\H(z_{i+1}) \bTheta_{k_{i+1},1}(z_i,z_{i+1}) & \mathbf{H}_3(z_{i+1}) \bTheta_{k_{i+1},2}(z_i,z_{i+1})
    \end{bmatrix} 
        \right) d z_1 \ldots d z_m + o(1), \label{eq:hat_psi_loop_1}
\end{align}

To treat the product of block matrices in \eqref{eq:hat_psi_loop_1}, we introduce the following result, on the (contour) integration over two $dz_{i}$  and $dz_{i+1}$ \emph{only}.

\begin{Lemma}\label{lem:matrix_prod}
Given $i \in \{ 1, \ldots, m-1 \} $, we have, for the following product of three matrices that
\begin{align}
&\oint_{\Gamma_{k_{i}}}\oint_{\Gamma_{k_{i+1}}} \prod_{j=i-1}^{i+1} \begin{bmatrix}
    \mathbf{H}_1(z_{j+1}) \bTheta_{k_{j+1},1}(z_j,z_{j+1}) & \mathbf{H}_2(z_{j+1}) \bTheta_{k_{j+1},2}(z_j,z_{j+1}) \\ 
    \mathbf{H}_2^\H(z_{j+1}) \bTheta_{k_{j+1},1}(z_j,z_{j+1}) & \mathbf{H}_3(z_{j+1}) \bTheta_{k_{j+1},2}(z_j,z_{j+1})
\end{bmatrix} dz_i dz_{i+1} \notag \\ 
&=\oint_{\Gamma_{k_{i}}}\oint_{\Gamma_{k_{i+1}}}\begin{bmatrix}
    \bm{X}_1(z_{i-1},z_i,z_{i+1},z_{i+2}) & \bm{X}_2(z_{i-1},z_i,z_{i+1},_{i+2}) \\ 
    \bm{X}_3(z_{i-1},z_i,z_{i+1},z_{i+2}) & \bm{X}_4(z_{i-1},z_i,z_{i+1},_{i+2})
\end{bmatrix} dz_i dz_{i+1} +o_{\| \cdot \|}(1), \label{eq:3matrices}
\end{align}
where
\begin{align*}
\bm{X}_1(z_{i-1},z_i,z_{i+1},z_{i+2}) =\mathbf{H}_1(z_{i}) \bTheta_{k_{i},1}(z_{i-1},z_{i})\mathbf{H}_1(z_{i+1}) \bTheta_{k_{i+1},1}(z_{i},z_{i+1})\mathbf{H}_1(z_{i+2}) \bTheta_{k_{i+2},1}(z_{i+1},z_{i+2}),\\
\bm{X}_2(z_{i-1},z_i,z_{i+1},z_{i+2}) = \mathbf{H}_1(z_{i}) \bTheta_{k_{i},1}(z_{i-1},z_{i})\mathbf{H}_1(z_{i+1}) \bTheta_{k_{i+1},1}(z_{i},z_{i+1})\mathbf{H}_2(z_{i+2}) \bTheta_{k_{i+2},2}(z_{i+1},z_{i+2}),\\
\bm{X}_3(z_{i-1},z_i,z_{i+1},z_{i+2}) = \mathbf{H}_2^\H(z_{i}) \bTheta_{k_{i},1}(z_{i-1},z_{i})\mathbf{H}_1(z_{i+1}) \bTheta_{k_{i+1},1}(z_{i},z_{i+1})\mathbf{H}_1(z_{i+2}) \bTheta_{k_{i+2},1}(z_{i+1},z_{i+2}),\\
\bm{X}_4(z_{i-1},z_i,z_{i+1},z_{i+2}) = \mathbf{H}_2^\H(z_{i}) \bTheta_{k_{i},1}(z_{i-1},z_{i})\mathbf{H}_1(z_{i+1}) \bTheta_{k_{i+1},1}(z_{i},z_{i+1})\mathbf{H}_1(z_{i+2}) \bTheta_{k_{i+2},1}(z_{i+1},z_{i+2}).
\end{align*}
\end{Lemma}


\begin{proof}[Proof of \Cref{lem:matrix_prod}]
Consider the following product of three (having index $i-1, i, i+1$) two-by-two block matrices, for which we evaluate only the integration with respect to $z_i$ and $z_{i+1}$,
\begin{align*}
    &\oint_{\Gamma_{k_{i}}}\oint_{\Gamma_{k_{i+1}}} \prod_{j=i-1}^{i+1} \begin{bmatrix}
      \mathbf{H}_1(z_{j+1}) \bTheta_{k_{j+1},1}(z_j,z_{j+1}) & \mathbf{H}_2(z_{j+1}) \bTheta_{k_{j+1},2}(z_j,z_{j+1}) \\ 
      \mathbf{H}_2^\H(z_{j+1}) \bTheta_{k_{j+1},1}(z_j,z_{j+1}) & \mathbf{H}_3(z_{j+1}) \bTheta_{k_{j+1},2}(z_j,z_{j+1})
    \end{bmatrix} dz_i dz_{i+1}.
\end{align*}
This will result in a two-by-two block matrix, each block is the sum of four matrices of the form 
\begin{equation}
    \oint_{\Gamma_{k_{i}}}\oint_{\Gamma_{k_{i+1}}}  \mathbf{H} \bTheta \cdot \mathbf{H} \bTheta \cdot \mathbf{H} \bTheta dz_i dz_{i+1},
\end{equation}
for $\mathbf{H} = \mathbf{H}_1, \mathbf{H}_2, \mathbf{H}_2^\H, \mathbf{H}_3$ and $\bTheta = \bTheta_1, \bTheta_2$.

As an example, let us consider the following term, as one of the four terms in the sum of the $(1,1)$ block.
By definitions of $\mathbf{H}_1,\mathbf{H}_2$ and $\bTheta_{k,1}$ in \eqref{eq:def_H} and \eqref{eq:def_Theta_1}, respectively, we have 
\begin{align*}
    &\oint_{\Gamma_{k_{i}}}\oint_{\Gamma_{k_{i+1}}}  \mathbf{H}_2(z_{i}) \bTheta_{k_{i},2}(z_{i-1},z_{i})\mathbf{H}_2^\H(z_{i+1}) \bTheta_{k_{i+1},1}(z_{i},z_{i+1})\mathbf{H}_1(z_{i+2}) \bTheta_{k_{i+2},1}(z_{i+1},z_{i+2}) \,dz_i dz_{i+1} \\
    &= \oint_{\Gamma_{k_{i}}}\oint_{\Gamma_{k_{i+1}}} {\left(\I_K+\frac{1+z_i m(z_i)}{\sigma^2}\P \A^\H \A \right)^{-1} \gamma(z_{i-1},z_{i})\P \left(\I_K+\frac{1+z_i m(z_i)}{\sigma^2}\A^\H \A\P  \right)^{-1}} \A^\H \\
    &\times \left( m(z_{i}) m(z_{i+1}) \mathbf{M}_{k_{i+1}} + \eta(z_{i}, z_{i+1}) \frac1T \tr \mathbf{M}_{k_{i+1}} \cdot \I_N \right) \A \frac{z_{i+2} m(z_{i+2}) + 1}{\sigma^2 m(z_{i+2})}\left(\I_K+\frac{1+z_{i+2}m(z_{i+2})}{\sigma^2}\P \A^\H \A \right)^{-1}\\
    &\times \P\A^\H \left( m(z_{i+1}) m(z_{i+2}) \mathbf{M}_{k_{i+1}} + \eta(z_{i+1}, z_{i+2}) \frac1T \tr \mathbf{M}_{k_{i+1}} \cdot \I_N \right) \A \,dz_i dz_{i+1} \\
    &= \oint_{\Gamma_{k_{i}}}\oint_{\Gamma_{k_{i+1}}} \underbrace{\left(\I_K+\frac{1+z_i m(z_i)}{\sigma^2}\P \A^\H \A \right)^{-1} \gamma(z_{i-1},z_{i})\P \left(\I_K+\frac{1+z_{i+1} m(z_{i+1})}{\sigma^2}\A^\H \A\P  \right)^{-1}}_{ \mathbf{H}_2(z_{i}) \bTheta_{k_{i},2}(z_{i-1},z_{i})\mathbf{H}_2^\H(z_{i+1}) } \bm{\Gamma}_1(z_i,z_{i+1},z_{i+2}) \,dz_i dz_{i+1}
\end{align*}
where we introduce
\begin{align*}
    \bm{\Gamma}_1(z_i,z_{i+1},z_{i+2})&=\A^\H \left( m(z_{i}) m(z_{i+1}) \mathbf{M}_{k_{i+1}} + \eta(z_{i}, z_{i+1}) \frac1T \tr \mathbf{M}_{k_{i+1}} \cdot \I_N \right) \A \frac{z_{i+2} m(z_{i+2}) + 1}{\sigma^2  m(z_{i+2})}\\
    &\times \left(\I_K+\frac{1+z_{i+2}m(z_{i+2})}{\sigma^2}\P \A^\H \A \right)^{-1} \P\A^\H \left( m(z_{i+1}) m(z_{i+2}) \mathbf{M}_{k_{i+1}} + \eta(z_{i+1}, z_{i+2}) \frac1T \tr \mathbf{M}_{k_{i+1}} \cdot \I_N \right) \A.
\end{align*}
Note that $\bm{\Gamma}_1(z_i,z_{i+1},z_{i+2})$ is a matrix polynomial that does not contain any pole (for $z_i,z_{i+1}$ under evaluation).
We thus have, by \Cref{lem:ABA=ABA} and the (asymptotic) eigendecomposition $\A \P\A^\H = \U \mathbf{L} \U^\H + o_{\| \cdot \|}(1)$ with diagonal $\mathbf{L} = \diag\{\lambda_1,\ldots,\lambda_K\} \in \RR^{K \times K}$ and $\U=[\uu_1,\ldots,\uu_K] \in \CC^{N \times K}$ as in \Cref{ass:subspace},
\begin{align*}
    &\oint_{\Gamma_{k_{i}}}\oint_{\Gamma_{k_{i+1}}}  \mathbf{H}_2(z_{i}) \bTheta_{k_{i},2}(z_{i-1},z_{i})\mathbf{H}_2^\H(z_{i+1}) \bTheta_{k_{i+1},1}(z_{i},z_{i+1})\mathbf{H}_1(z_{i+2}) \bTheta_{k_{i+2},1}(z_{i+1},z_{i+2}) dz_i dz_{i+1} \\
    &= \oint_{\Gamma_{k_{i}}}\oint_{\Gamma_{k_{i+1}}} \frac{z_i m(z_i)+1}{\sigma^2}\P \A^\H \left(\I_K+\frac{z_i m(z_i)+1}{\sigma^2}\A \P \A^\H  \right)^{-1}  \gamma(z_{i-1},z_{i})\frac{z_{i+1} m(z_{i+1})+1}{\sigma^2}\A \P \A^\H \\
    &\times \left(\I_K+(z_{i+1} m(z_{i+1}) + 1)\A \P \A^\H  \right)^{-1} \A\P \bm{\Gamma}_1(z_i,z_{i+1},z_{i+2}) dz_i dz_{i+1}\\
    &=\oint_{\Gamma_{k_{i}}}\oint_{\Gamma_{k_{i+1}}} \frac{(z_i m(z_i) + 1)(z_{i+1} m(z_{i+1}) + 1)}{\sigma^4} \gamma(z_{i-1},z_{i}) \P \A^\H \U\left(\I_K+\frac{z_i m(z_i)+1}{\sigma^2}\mathbf{L}  \right)^{-1} \mathbf{L}\\
    &\times  \left(\I_K+\frac{z_{i+1} m(z_{i+1})+1}{\sigma^2}\mathbf{L} \right)^{-1} \U^\H \A\P \bm{\Gamma}_1(z_i,z_{i+1},z_{i+2}) dz_i dz_{i+1} + o_{\| \cdot \|}(1) \\
    &=\oint_{\Gamma_{k_{i}}}\oint_{\Gamma_{k_{i+1}}}\!\!\!\! \frac{(z_i m(z_i) + 1)(z_{i+1} m(z_{i+1}) + 1)}{\sigma^4} \gamma(z_{i-1},z_{i}) \P \A^\H \sum_{n=1}^K  \frac{\sigma^4 \lambda_n \uu_n \uu_n^\H \times \A\P \bm{\Gamma}_1(z_i,z_{i+1},z_{i+2}) dz_i dz_{i+1}}{(\sigma^2+(z_{i+1} m(z_{i+1}) + 1)\lambda_n)(\sigma^2+(z_{i} m(z_{i}) + 1)\lambda_n)}\\
    &=\lim_{z\to \bar\lambda_i}\oint_{\Gamma_{k_{i+1}}}\!\!\!\! (z_i m(z_i) + 1)(z_{i+1} m(z_{i+1}) + 1) \gamma(z_{i-1},z_{i}) \P \A^\H \sum_{n=1}^K \!\frac{(z-\bar\lambda_i)\lambda_n \uu_n \uu_n^\H \times \A\P \bm{\Gamma}_1(z_i,z_{i+1},z_{i+2}) dz_{i+1}}{(\sigma^2\!+\!(z_{i+1} m(z_{i+1}) \!+ \!1)\lambda_n)(\sigma^2\!+\!(z_{i} m(z_{i}) \!+\! 1)\lambda_n)} \\
    &=o_{\| \cdot \|}(1),
\end{align*}
where we used residue calculus in the last line, with $\bar\lambda_i \equiv \sigma^2(1+\ell_i+c\frac{1+\ell_i}{\ell_i})$ the asymptotic position of the isolated eigenvalue circled by $\Gamma_{k_{i}}$, and the crucial observation that when integrating over $z_{i+1}$, the integrant does \emph{not} contain pole circled by the contour $\Gamma_{k_{i+1}}$. 
This is due to the fact that $\bm{\Gamma}_1(z_i,z_{i+1},z_{i+2})$ does not contain pole and the only pole $\lambda_i$ is already enclosed by $\Gamma_{k_{i}}$ and cannot be enclosed by $\Gamma_{k_{i+1}}$.


Similarly, other terms containing  $\mathbf{H}_2(z_{i}) \bTheta_{k_{i},2}(z_{i-1},z_{i})\mathbf{H}_2^\H(z_{i+1}) $, including 
\begin{align*}
    &\mathbf{H}_2(z_{i}) \bTheta_{k_{i},2}(z_{i-1},z_{i})\mathbf{H}_2^\H(z_{i+1}) \bTheta_{k_{i+1},1}(z_{i},z_{i+1})\mathbf{H}_2(z_{i+2}) \bTheta_{k_{i+2},2}(z_{i+1},z_{i+2}), \\
    &\mathbf{H}_1(z_{i}) \bTheta_{k_{i},1}(z_{i-1},z_{i})\mathbf{H}_2^\H(z_{i+1}) \bTheta_{k_{i+1},2}(z_{i},z_{i+1})\mathbf{H}_2^\H(z_{i+2}) \bTheta_{k_{i+2},1}(z_{i+1},z_{i+2}), \\
    &\mathbf{H}_2^\H(z_{i}) \bTheta_{k_{i},1}(z_{i-1},z_{i})\mathbf{H}_2(z_{i+1}) \bTheta_{k_{i+1},2}(z_{i},z_{i+1})\mathbf{H}_2^\H(z_{i+2}) \bTheta_{k_{i+2},1}(z_{i+1},z_{i+2}),
\end{align*}
lead to matrices of vanishing spectral norm after contour integration.

Following the same idea, we evaluate
\begin{align*}
    & \oint_{\Gamma_{k_{i}}}\oint_{\Gamma_{k_{i+1}}}  \mathbf{H}_2(z_{i}) \bTheta_{k_{i},2}(z_{i-1},z_{i})\mathbf{H}_3(z_{i+1}) \bTheta_{k_{i+1},2}(z_{i},z_{i+1})\mathbf{H}_2^\H(z_{i+2}) \bTheta_{k_{i+2},1}(z_{i+1},z_{i+2}) dz_i dz_{i+1} \\
    &= - \oint_{\Gamma_{k_{i}}}\oint_{\Gamma_{k_{i+1}}} \left(\I_K+\frac{1+z_i m(z_i)}{\sigma^2}\P \A^\H \A \right)^{-1} \gamma(z_{i-1},z_{i})\P m(z_{i+1}) \A^\H \A \left(\I_K+\frac{1+z_{i+1} m(z_{i+1})}{\sigma^2}\P\A^\H \A  \right)^{-1} \\
    & \times \gamma(z_i,z_{i+1})\P  \left(\I_K+\frac{1+z_{i+2} m(z_{i+2})}{\sigma^2}\A^\H \A\P  \right)^{-1} \A^\H \left( m(z_{i+1}) m(z_{i+2}) \mathbf{M}_{k_{i+1}} + \eta(z_{i+1}, z_{i+2}) \frac1T \tr \mathbf{M}_{k_{i+1}} \cdot \I_N \right) \A dz_i dz_{i+1} \\
    &= - \oint_{\Gamma_{k_{i}}}\oint_{\Gamma_{k_{i+1}}} \underbrace{\left(\I_K+\frac{1+z_i m(z_i)}{\sigma^2}\P \A^\H \A \right)^{-1} \gamma(z_{i-1},z_{i})\P m(z_{i+1}) \A^\H \A \left(\I_K+\frac{1+z_{i+1} m(z_{i+1})}{\sigma^2}\P\A^\H \A  \right)^{-1}}_{ \mathbf{H}_2(z_{i}) \bTheta_{k_{i},2}(z_{i-1},z_{i})\mathbf{H}_3(z_{i+1}) } \\
    & \times \bm{\Gamma}_2(z_i,z_{i+1},z_{i+2}) dz_i dz_{i+1},
\end{align*}
where 
\begin{align*}
    \bm{\Gamma}_2(z_i,z_{i+1},z_{i+2})&=\gamma(z_i,z_{i+1})\P  \left(\I_K+\frac{1+z_{i+2} m(z_{i+2})}{\sigma^2}\A^\H \A\P  \right)^{-1} \A^\H \\ 
    & \times \left( m(z_{i+1}) m(z_{i+2}) \mathbf{M}_{k_{i+1}} + \eta(z_{i+1}, z_{i+2}) \frac1T \tr \mathbf{M}_{k_{i+1}} \cdot \I_N \right) \A,
\end{align*}
is a matrix polynomial that does not contain any pole (for $z_i$ and $z_{i+1}$).
We thus get
\begin{align*}
    &\oint_{\Gamma_{k_{i}}}\oint_{\Gamma_{k_{i+1}}}  \mathbf{H}_2(z_{i}) \bTheta_{k_{i},2}(z_{i-1},z_{i})\mathbf{H}_3(z_{i+1}) \bTheta_{k_{i+1},2}(z_{i},z_{i+1})\mathbf{H}_2^\H(z_{i+2}) \bTheta_{k_{i+2},1}(z_{i+1},z_{i+2}) dz_i dz_{i+1} \\
    &= -\left(\frac{1}{2 \pi \imath} \right)^2\oint_{\Gamma_{k_{i}}}\oint_{\Gamma_{k_{i+1}}} \left(\I_K+\frac{1+z_i m(z_i)}{\sigma^2}\P \A^\H \A \right)^{-1} \gamma(z_{i-1},z_{i})\P m(z_{i+1}) \A^\H \A \left(\I_K+\frac{1+z_{i+1} m(z_{i+1})}{\sigma^2}\P\A^\H \A  \right)^{-1}  \\
    & \bm{\Gamma}_2(z_i,z_{i+1},z_{i+2}) dz_i dz_{i+1}\\
    &=-\oint_{\Gamma_{k_{i}}}\oint_{\Gamma_{k_{i+1}}} \frac{(z_i m(z_i) + 1)(z_{i+1} m(z_{i+1}) + 1)}{\sigma^4} m(z_{i+1})\P \A^\H \left(\I_K+\frac{z_i m(z_i)+1}{\sigma^2}\A \P \A^\H  \right)^{-1}  \gamma(z_{i-1},z_{i})\A \P \A^\H  \\
    &\left(\I_K+(z_{i+1} m(z_{i+1}) + 1)\A \P \A^\H  \right)^{-1} \A  \bm{\Gamma}_2(z_i,z_{i+1},z_{i+2}) dz_i dz_{i+1}\\
    &= -\left(\frac{1}{2 \pi \imath} \right)^2\oint_{\Gamma_{k_{i}}}\oint_{\Gamma_{k_{i+1}}} \frac{(z_i m(z_i) + 1)(z_{i+1} m(z_{i+1}) + 1)}{\sigma^4}m(z_{i+1}) \gamma(z_{i-1},z_{i}) \P \A^\H \U\left(\I_K+\frac{z_i m(z_i)+1}{\sigma^2}\mathbf{L}  \right)^{-1} \mathbf{L} \\
    &\left(\I_K+(z_{i+1} m(z_{i+1}) + 1)\mathbf{L} \right)^{-1} \U^\H \bm{\Gamma}_2(z_i,z_{i+1},z_{i+2}) dz_i dz_{i+1}\\
    &=- \oint_{\Gamma_{k_{i}}}\oint_{\Gamma_{k_{i+1}}}\frac{(z_i m(z_i) + 1)(z_{i+1} m(z_{i+1}) + 1)}{\sigma^4} m(z_{i+1})\gamma(z_{i-1},z_{i}) \P \A^\H \\
    & \times \sum_{n=1}^K \frac{\sigma^4 \lambda_n \uu_n \uu_n^\H}{(\sigma^2+(z_{i+1} m(z_{i+1}) + 1)\lambda_n)(\sigma^2+(z_{i} m(z_{i}) + 1)\lambda_n)} \bm{\Gamma}_2(z_i,z_{i+1},z_{i+2}) dz_i dz_{i+1} = o_{\| \cdot \|}(1),
\end{align*}
where the last line follows the same line of arguments as for $\bm{\Gamma}_1$ above.
Similarly, other terms containing $\mathbf{H}_2(z_{i}) \bTheta_{k_{i},2}(z_{i-1},z_{i})\mathbf{H}_3(z_{i+1}) $, including 
\begin{align*}
    &\mathbf{H}_1(z_{i}) \bTheta_{k_{i},1}(z_{i-1},z_{i})\mathbf{H}_2(z_{i+1}) \bTheta_{k_{i+1},2}(z_{i},z_{i+1})\mathbf{H}_3(z_{i+2}) \bTheta_{k_{i+2},2}(z_{i+1},z_{i+2}), \\
    &\mathbf{H}_2(z_{i}) \bTheta_{k_{i},2}(z_{i-1},z_{i})\mathbf{H}_3(z_{i+1}) \bTheta_{k_{i+1},2}(z_{i},z_{i+1})\mathbf{H}_3(z_{i+2}) \bTheta_{k_{i+2},2}(z_{i+1},z_{i+2}), \\
    &\mathbf{H}_2^\H(z_{i}) \bTheta_{k_{i},1}(z_{i-1},z_{i})\mathbf{H}_2(z_{i+1}) \bTheta_{k_{i+1},2}(z_{i},z_{i+1})\mathbf{H}_3(z_{i+2}) \bTheta_{k_{i+2},2}(z_{i+1},z_{i+2}),
\end{align*}
again lead to matrices of vanishing spectral norm after contour integration.

We then consider
    \begin{align*}
        &\oint_{\Gamma_{k_{i}}}\oint_{\Gamma_{k_{i+1}}}  \mathbf{H}_3(z_{i}) \bTheta_{k_{i},2}(z_{i-1},z_{i})\mathbf{H}_2^\H(z_{i+1}) \bTheta_{k_{i+1},1}(z_{i},z_{i+1})\mathbf{H}_1(z_{i+2}) \bTheta_{k_{i+2},1}(z_{i+1},z_{i+2}) dz_i dz_{i+1} \\
        &= - \oint_{\Gamma_{k_{i}}}\oint_{\Gamma_{k_{i+1}}} m(z_i)\A^\H \A \left(\I_K+\frac{1+z_i m(z_i)}{\sigma^2}\P \A^\H \A \right)^{-1}  \gamma(z_{i-1},z_{i})\P \left(\I_K+\frac{1+z_{i+1} m(z_{i+1})}{\sigma^2}\A^\H \A\P  \right)^{-1} \A^\H \\
        &\times \left( m(z_{i}) m(z_{i+1}) \mathbf{M}_{k_{i+1}} + \eta(z_{i}, z_{i+1}) \frac1T \tr \mathbf{M}_{k_{i+1}} \cdot \I_N \right) \A \frac{z_{i+2} m(z_{i+2}) + 1}{\sigma^2 m(z_{i+2})}\left(\I_K+\frac{1+z_{i+2} m(z_{i+2})}{\sigma^2}\P \A^\H \A \right)^{-1}\\
        &\times \P\A^\H \left( m(z_{i+1}) m(z_{i+2}) \mathbf{M}_{k_{i+1}} + \eta(z_{i+1}, z_{i+2}) \frac1T \tr \mathbf{M}_{k_{i+1}} \cdot \I_N \right) \A dz_i dz_{i+1} \\
        &= -\oint_{\Gamma_{k_{i}}}\oint_{\Gamma_{k_{i+1}}} \underbrace{m(z_i)\A^\H \A \left(\I_K+\frac{1+z_i m(z_i)}{\sigma^2}\P \A^\H \A \right)^{-1}  \gamma(z_{i-1},z_{i})\P \left(\I_K+\frac{1+z_{i+1} m(z_{i+1})}{\sigma^2}\A^\H \A\P  \right)^{-1}}_{ \mathbf{H}_3(z_{i}) \bTheta_{k_{i},2}(z_{i-1},z_{i})\mathbf{H}_2^\H(z_{i+1}) }\\
        &\times \bm{\Gamma}_3(z_i,z_{i+1},z_{i+2}) dz_i dz_{i+1}
    \end{align*}
    where 
    \begin{align*}
        \bm{\Gamma}_3(z_i,z_{i+1},z_{i+2})&=\A^\H \left( m(z_{i}) m(z_{i+1}) \mathbf{M}_{k_{i+1}} + \eta(z_{i}, z_{i+1}) \frac1T \tr \mathbf{M}_{k_{i+1}} \cdot \I_N \right) \A \frac{z_{i+2} m(z_{i+2}) + 1}{\sigma^2m(z_{i+2})}\\
        &\times \left(\I_K+\frac{1+z_{i+2} m(z_{i+2})}{\sigma^2}\P \A^\H \A \right)^{-1}\P\A^\H \left( m(z_{i+1}) m(z_{i+2}) \mathbf{M}_{k_{i+1}} + \eta(z_{i+1}, z_{i+2}) \frac1T \tr \mathbf{M}_{k_{i+1}} \cdot \I_N \right) \A
    \end{align*}
    is a matrix polynomial that does not contain any pole (for $z_i$ and $z_{i+1}$).
    We thus obtain
    \begin{align*}
        & \oint_{\Gamma_{k_{i}}}\oint_{\Gamma_{k_{i+1}}}  \mathbf{H}_3(z_{i}) \bTheta_{k_{i},2}(z_{i-1},z_{i})\mathbf{H}_2^\H(z_{i+1}) \bTheta_{k_{i+1},1}(z_{i},z_{i+1})\mathbf{H}_1(z_{i+2}) \bTheta_{k_{i+2},1}(z_{i+1},z_{i+2}) dz_i dz_{i+1} \\
        &= - \oint_{\Gamma_{k_{i}}}\oint_{\Gamma_{k_{i+1}}} m(z_i)\A^\H \A \left(\I_K+\frac{1+z_i m(z_i)}{\sigma^2}\P \A^\H \A \right)^{-1}  \gamma(z_{i-1},z_{i})\P \left(\I_K+\frac{1+z_{i+1} m(z_{i+1})}{\sigma^2}\A^\H \A\P  \right)^{-1} \\
        & \times \bm{\Gamma}_3(z_i,z_{i+1},z_{i+2}) dz_i dz_{i+1}\\
        &=-\oint_{\Gamma_{k_{i}}}\oint_{\Gamma_{k_{i+1}}} m(z_i)\frac{z_i m(z_i)+1}{\sigma^2} \A^\H \A\P\A^\H \left(\I_K+\frac{z_i m(z_i)+1}{\sigma^2}\A \P \A^\H  \right)^{-1}  \gamma(z_{i-1},z_{i})\frac{1+z_{i+1} m(z_{i+1})}{\sigma^2}\A \P \A^\H \\
        &\left(\I_K+\frac{1+z_{i+1} m(z_{i+1})}{\sigma^2}\A \P \A^\H  \right)^{-1} \A\P \bm{\Gamma}_3(z_i,z_{i+1},z_{i+2}) dz_i dz_{i+1}\\
        &= -\oint_{\Gamma_{k_{i}}}\oint_{\Gamma_{k_{i+1}}} m(z_i)\frac{z_i m(z_i)+1}{\sigma^2} \A^\H \U\mathbf{L}\left(\I_K+\frac{z_i m(z_i)+1}{\sigma^2}\mathbf{L}  \right)^{-1} \gamma(z_{i-1},z_{i})\frac{1+z_{i+1} m(z_{i+1})}{\sigma^2} \mathbf{L}\\
        &\times \left(\I_K+\frac{1+z_{i+1} m(z_{i+1})}{\sigma^2}\mathbf{L} \right)^{-1} \U^\H \A\P \bm{\Gamma}_3(z_i,z_{i+1},z_{i+2}) dz_i dz_{i+1}\\
        &=- \oint_{\Gamma_{k_{i}}}\oint_{\Gamma_{k_{i+1}}}m(z_i)\frac{(z_i m(z_i) + 1)(z_{i+1} m(z_{i+1}) + 1)}{\sigma^4} \gamma(z_{i-1},z_{i})\A^\H \\
        &\times \sum_{n=1}^K \frac{\sigma^4 \lambda_n^2 \uu_n \uu_n^\H}{(\sigma^2+(z_{i+1} m(z_{i+1}) + 1)\lambda_n)(\sigma^2+(z_{i} m(z_{i}) + 1)\lambda_n)}  \A\P \bm{\Gamma}_3(z_i,z_{i+1},z_{i+2}) dz_i dz_{i+1} = o_{\| \cdot \|}(1).
    \end{align*}
    Similarly, other terms containing $ \mathbf{H}_3(z_{i}) \bTheta_{k_{i},2}(z_{i-1},z_{i})\mathbf{H}_2^\H(z_{i+1}) $, including 
    \begin{align*}
        &\mathbf{H}_3(z_{i}) \bTheta_{k_{i},2}(z_{i-1},z_{i})\mathbf{H}_2^\H(z_{i+1}) \bTheta_{k_{i+1},1}(z_{i},z_{i+1})\mathbf{H}_2(z_{i+2}) \bTheta_{k_{i+2},2}(z_{i+1},z_{i+2}), \\
        &\mathbf{H}_3(z_{i}) \bTheta_{k_{i},2}(z_{i-1},z_{i})\mathbf{H}_3(z_{i+1}) \bTheta_{k_{i+1},2}(z_{i},z_{i+1})\mathbf{H}_2^\H(z_{i+2}) \bTheta_{k_{i+2},1}(z_{i+1},z_{i+2}),
    \end{align*}
    lead to matrices of vanishing spectral norm after contour integration.
    
    
    It remains to evaluate
    \begin{align*}
        & \oint_{\Gamma_{k_{i}}}\oint_{\Gamma_{k_{i+1}}}  \mathbf{H}_3(z_{i}) \bTheta_{k_{i},2}(z_{i-1},z_{i})\mathbf{H}_3(z_{i+1}) \bTheta_{k_{i+1},2}(z_{i},z_{i+1})\mathbf{H}_3(z_{i+2}) \bTheta_{k_{i+2},2}(z_{i+1},z_{i+2}) dz_i dz_{i+1} \\
        &= - \oint_{\Gamma_{k_{i}}}\oint_{\Gamma_{k_{i+1}}} m(z_i)\A^\H \A \left(\I_K+\frac{1+z_i m(z_i)}{\sigma^2}\P \A^\H \A \right)^{-1}  \gamma(z_{i-1},z_{i})\P m(z_{i+1})\A^\H \A \left(\I_K+\frac{1+z_{i+1} m(z_{i+1})}{\sigma^2}\P \A^\H \A \right)^{-1} \\
        & \times \gamma(z_{i},z_{i+1})\P m(z_{i+2})\A^\H \A \left(\I_K+\frac{1+z_{i+2} m(z_{i+2})}{\sigma^2}\P \A^\H \A \right)^{-1}  \gamma(z_{i+1},z_{i+2})\P dz_i dz_{i+1} \\
        &= - \oint_{\Gamma_{k_{i}}}\oint_{\Gamma_{k_{i+1}}} m(z_i)\gamma(z_{i-1},z_{i}) m(z_{i+1})\gamma(z_{i},z_{i+1}) m(z_{i+2}) \gamma(z_{i+1},z_{i+2})\A^\H \frac{1+z_i m(z_i)}{\sigma^2} \A\P\A^\H \\
        &\left(\I_K+\frac{z_i m(z_i)+1}{\sigma^2}\A \P \A^\H  \right)^{-1} \frac{1+z_{i+1} m(z_{i+1})}{\sigma^2}\A \P \A^\H \left(\I_K+\frac{1+z_{i+1} m(z_{i+1})}{\sigma^2}\A \P \A^\H  \right)^{-1} \A\P \A^\H \\
        & \times \frac{1+z_{i+2} m(z_{i+2})}{\sigma^2}\A \P \A^\H \left(\I_K+\frac{1+z_{i+2} m(z_{i+2})}{\sigma^2}\A \P \A^\H  \right)^{-1}\A \P  dz_i dz_{i+1} \\
        &= - \oint_{\Gamma_{k_{i}}}\oint_{\Gamma_{k_{i+1}}} m(z_i)\gamma(z_{i-1},z_{i}) m(z_{i+1})\gamma(z_{i},z_{i+1}) m(z_{i+2}) \gamma(z_{i+1},z_{i+2})\frac{1+z_i m(z_i)}{\sigma^2}  \frac{1+z_{i+1} m(z_{i+1})}{\sigma^2} \frac{1+z_{i+2} m(z_{i+2})}{\sigma^2}\\
        & \times \A^\H \U\mathbf{L}\left(\I_K+\frac{z_i m(z_i)+1}{\sigma^2}\mathbf{L}  \right)^{-1} \mathbf{L}\left(\I_K+\frac{1+z_{i+1} m(z_{i+1})}{\sigma^2} \mathbf{L} \right)^{-1}\mathbf{L}^2 \left(\I_K+\frac{1+z_{i+2} m(z_{i+2})}{\sigma^2}\mathbf{L} \right)^{-1} \U^\H \A \P  dz_i dz_{i+1} \\
        &= - \oint_{\Gamma_{k_{i}}}\oint_{\Gamma_{k_{i+1}}} m(z_i)\gamma(z_{i-1},z_{i}) m(z_{i+1})\gamma(z_{i},z_{i+1}) m(z_{i+2}) \gamma(z_{i+1},z_{i+2})\frac{1+z_i m(z_i)}{\sigma^2} \frac{(z_{i+1} m(z_{i+1}) + 1)(z_{i+2} m(z_{i+2}) + 1)}{\sigma^4}\\
        & \times \A^\H \sum_{n=1}^K \frac{\sigma^6 \lambda_n^4 \uu_n \uu_n^\H}{(\sigma^2+(z_{i+2} m(z_{i+2}) + 1)\lambda_n)(\sigma^2+(z_{i+1} m(z_{i+1}) + 1)\lambda_n)(\sigma^2+(z_{i} m(z_{i}) + 1)\lambda_n)}  \A \P  dz_i dz_{i+1} = o_{\| \cdot \|}(1).
    \end{align*}

Ignoring items of vanishing spectral norms in the resulting two-by-two block matrix, we conclude the proof of \Cref{lem:matrix_prod}.
 \end{proof}
    


In the following, we ignore, for the sake of notational convenience, the arguments of $\bm{X}_1, \bm{X}_2, \bm{X}_3, \bm{X}_4$.

In \Cref{lem:matrix_prod} we treat the product of three matrices and the integral over $z_i, z_{i+1}$.
Consider now the product of four matrices and its integral over $z_i, z_{i+1}$ and $z_{i+2}$.
It thus follows from \Cref{lem:matrix_prod} that
\begin{align*}
    &\oint_{\Gamma_{k_{i}}}\oint_{\Gamma_{k_{i+1}}}\oint_{\Gamma_{k_{i+2}}}
    \prod_{j=i-1}^{i+2} \begin{bmatrix}
    \mathbf{H}_1(z_{j+1}) \bTheta_{k_{j+1},1}(z_j,z_{j+1}) & \mathbf{H}_2(z_{j+1}) \bTheta_{k_{j+1},2}(z_j,z_{j+1}) \\ 
    \mathbf{H}_2^\H(z_{j+1}) \bTheta_{k_{j+1},1}(z_j,z_{j+1}) & \mathbf{H}_3(z_{j+1}) \bTheta_{k_{j+1},2}(z_j,z_{j+1})
\end{bmatrix} dz_i dz_{i+1} dz_{i+2} \\
    &= \oint_{\Gamma_{k_{i}}}\oint_{\Gamma_{k_{i+1}}}\oint_{\Gamma_{k_{i+2}}}\begin{bmatrix}
        \bm{X}_1 & \bm{X}_2 \\ 
        \bm{X}_3 & \bm{X}_4
    \end{bmatrix} \begin{bmatrix}
        \mathbf{H}_1(z_{i+3}) \bTheta_{k_{i+3},1}(z_{i+2},z_{i+3}) & \mathbf{H}_2(z_{i+3}) \bTheta_{k_{i+3},2}(z_{i+2},z_{i+3}) \\ 
        \mathbf{H}_2^\H(z_{i+3}) \bTheta_{k_{i+3},1}(z_{i+2},z_{i+3}) & \mathbf{H}_3(z_{i+3}) \bTheta_{k_{i+3},2}(z_{i+2},z_{i+3})
    \end{bmatrix} dz_i dz_{i+1} dz_{i+2} + o_{\| \cdot \|}(1) \\
    &= \oint_{\Gamma_{k_{i}}}\oint_{\Gamma_{k_{i+1}}}\oint_{\Gamma_{k_{i+2}}}\begin{bmatrix}
        \bm{X}_1 \mathbf{H}_1(z_{i+3}) \bTheta_{k_{i+3},1}(z_{i+2},z_{i+3}) & \bm{X}_1 \mathbf{H}_2(z_{i+3}) \bTheta_{k_{i+3},2}(z_{i+2},z_{i+3}) \\ 
        \bm{X}_3 \mathbf{H}_1(z_{i+3}) \bTheta_{k_{i+3},1}(z_{i+2},z_{i+3}) & \bm{X}_3 \mathbf{H}_2(z_{i+3}) \bTheta_{k_{i+3},2}(z_{i+2},z_{i+3})
    \end{bmatrix} dz_i dz_{i+1} dz_{i+2} + o_{\| \cdot \|}(1),
\end{align*}
where, similar to the proof of \Cref{lem:matrix_prod}, we ignore 
all block matrices of vanishing spectral norm.

Repeating the above approximating procedure on \eqref{eq:hat_psi_loop_1}, we conclude that 
\begin{align*}
    \hat{\psi}_{k_1, \ldots, k_m}&= \left(\frac{1}{2 \pi \imath} \right)^{m}\oint_{\Gamma_{k_1}}\ldots \oint_{\Gamma_{k_m}}  \tr \left( \begin{bmatrix}
        \mathbf{H}_1(z_1) \bTheta_{k_1,1}(z_m,z_1) & \mathbf{H}_2(z_1) \bTheta_{k_1,2}(z_m,z_1) \\ 
        \mathbf{H}_2^\H(z_1) \bTheta_{k_1,1}(z_m,z_1) & \mathbf{H}_3(z_1) \bTheta_{k_1,2}(z_m,z_1)
    \end{bmatrix} \right.\\
    & \left. \times \prod_{i=1}^{m-1} \begin{bmatrix}
        \mathbf{H}_1(z_{i+1}) \bTheta_{k_{i+1},1}(z_i,z_{i+1}) & \mathbf{H}_2(z_{i+1}) \bTheta_{k_{i+1},2}(z_i,z_{i+1}) \\ 
        \mathbf{H}_2^\H(z_{i+1}) \bTheta_{k_{i+1},1}(z_i,z_{i+1}) & \mathbf{H}_3(z_{i+1}) \bTheta_{k_{i+1},2}(z_i,z_{i+1})
    \end{bmatrix} 
        \right) d z_1 \ldots d z_m + o(1) \\
        &=\left(\frac{1}{2 \pi \imath} \right)^{m}\oint_{\Gamma_{k_1}}\ldots \oint_{\Gamma_{k_m}}  \tr \left( \begin{bmatrix}
            \bm{\Xi}_1 &  \bm{\Xi}_2\\ 
            \bm{\Xi}_3 &  \bm{\Xi}_4
            \end{bmatrix} \right) d z_1 \ldots d z_m + o(1),
    \end{align*}
where
\begin{align*}
    &\bm{\Xi}_1 = \mathbf{H}_1(z_1) \bTheta_{k_1,1}(z_m,z_1) \left( \prod_{i=1}^{m-1} \mathbf{H}_1(z_{i+1}) \bTheta_{k_{i+1},1}(z_i,z_{i+1}) \right), \\
    &\bm{\Xi}_2 = \mathbf{H}_1(z_1) \bTheta_{k_1,1}(z_m,z_1) \left(\prod_{i=1}^{m-2} \mathbf{H}_1(z_{i+1}) \bTheta_{k_{i+1},1}(z_i,z_{i+1}) \right) \mathbf{H}_2(z_{m}) \bTheta_{k_m,2}(z_{m-1},z_m), \\
    &\bm{\Xi}_3= \mathbf{H}_2^\H(z_1) \bTheta_{k_1,1}(z_m,z_1) \left(\prod_{i=1}^{m-1} \mathbf{H}_1(z_{i+1}) \bTheta_{k_{i+1},1}(z_i,z_{i+1}) \right),  \\
    &\bm{\Xi}_4 = \mathbf{H}_2^\H(z_1) \bTheta_{k_1,1}(z_m,z_1) \left( \prod_{i=1}^{m-2} \mathbf{H}_1(z_{i+1}) \bTheta_{k_{i+1},1}(z_i,z_{i+1}) \right)  \mathbf{H}_2(z_{m}) \bTheta_{k_m,2}(z_{m-1},z_m).
\end{align*}

Then, we have
\begin{align*}
    &\hat{\psi}_{k_1, \ldots, k_m} = \left(\frac{1}{2 \pi \imath} \right)^{m}\oint_{\Gamma_{k_1}}\ldots \oint_{\Gamma_{k_m}}  \tr \left( \bm{\Xi}_1 +  \bm{\Xi}_4 \right) d z_1 \ldots d z_m\\
    &=\left(\frac{1}{2 \pi \imath} \right)^{m}\oint_{\Gamma_{k_1}}\ldots \oint_{\Gamma_{k_m}}  \tr \left(\mathbf{H}_1(z_1) \bTheta_{k_1,1}(z_m,z_1)\prod_{i=1}^{m-1} \mathbf{H}_1(z_{i+1}) \bTheta_{k_{i+1},1}(z_i,z_{i+1}) \right)d z_1 \ldots d z_m\\
    &+\left(\frac{1}{2 \pi \imath} \right)^{m}\oint_{\Gamma_{k_1}}\ldots \oint_{\Gamma_{k_m}} \tr \left( \mathbf{H}_2^\H(z_1) \bTheta_{k_1,1}(z_m,z_1) \left(\prod_{i=1}^{m-2} \mathbf{H}_1(z_{i+1}) \bTheta_{k_{i+1},1}(z_i,z_{i+1}) \right) \mathbf{H}_2(z_{m}) \bTheta_{k_m,2}(z_{m-1},z_m)\right) d z_1 \ldots d z_m \\
    &=\left(\frac{1}{2 \pi \imath} \right)^{m}\oint_{\Gamma_{k_1}}\ldots \oint_{\Gamma_{k_m}}  \tr \left(\mathbf{H}_1(z_1) \bTheta_{k_1,1}(z_m,z_1)\prod_{i=1}^{m-1} \mathbf{H}_1(z_{i+1}) \bTheta_{k_{i+1},1}(z_i,z_{i+1}) \right)d z_1 \ldots d z_m\\
    &+\left(\frac{1}{2 \pi \imath} \right)^{m}\oint_{\Gamma_{k_1}}\ldots \oint_{\Gamma_{k_m}} \tr \left(  \bTheta_{k_1,1}(z_m,z_1) \left(\prod_{i=1}^{m-2} \mathbf{H}_1(z_{i+1}) \bTheta_{k_{i+1},1}(z_i,z_{i+1}) \right) \mathbf{H}_2(z_{m}) \bTheta_{k_m,2}(z_{m-1},z_m)\mathbf{H}_2^\H(z_1)\right) d z_1 \ldots d z_m\\
    &=\left(\frac{1}{2 \pi \imath} \right)^{m}\oint_{\Gamma_{k_1}}\ldots \oint_{\Gamma_{k_m}}  \tr \left(\mathbf{H}_1(z_1) \bTheta_{k_1,1}(z_m,z_1)\prod_{i=1}^{m-1} \mathbf{H}_1(z_{i+1}) \bTheta_{k_{i+1},1}(z_i,z_{i+1}) \right)d z_1 \ldots d z_m +o(1).
\end{align*}
    where we used in the last line the similar approximation as in the proof of \Cref{lem:matrix_prod}.

    Expanding this product, we get
\allowdisplaybreaks
\begin{align*}
    &\hat{\psi}_{k_1, \ldots, k_m}=\left(\frac{1}{2 \pi \imath} \right)^{m}\oint_{\Gamma_{k_1}}\ldots \oint_{\Gamma_{k_m}}  \tr \Big(\frac{z_{1} m(z_{1}) + 1}{\sigma^2 m(z_{1})}  \left(\I_K+\frac{z_1 m(z_1)+1}{\sigma^2 }\P \A^\H \A \right)^{-1} \P \A^\H  \Big( m(z_{1}) m(z_{1}) \mathbf{M}_{k_1} +  \\
    & \eta(z_{m}, z_{1}) \frac1T \tr \mathbf{M}_{k_1} \cdot \I_N \Big) \A  \prod_{i=1}^{m-1} \frac{z_{i+1} m(z_{i+1}) + 1}{\sigma^2 m(z_{i+1})} \left(\I_K+\frac{1+z_{i+1} m(z_{i+1})}{\sigma^2}\P \A^\H \A \right)^{-1} \P \A^\H  \Big( m(z_{i}) m(z_{i+1}) \mathbf{M}_{k_{i+1}} \\
    &+ \eta(z_{i}, z_{i+1}) \frac1T \tr \mathbf{M}_{k_{i+1}} \cdot \I_N \Big) \A \Big)d z_1 \ldots d z_m +o(1)\\
    &= \left(\frac{1}{2 \pi \imath} \right)^{m}\oint_{\Gamma_{k_1}}\ldots \oint_{\Gamma_{k_m}}  \tr  \Big(\frac{z_{1} m(z_{1}) + 1}{\sigma^2 m(z_{1})}\P \A^\H \left(\I_K+\frac{z_1 m(z_1)+1}{\sigma^2 } \A\P \A^\H \right)^{-1}  \Big( m(z_{m}) m(z_{1}) \mathbf{M}_{k_1}\\
    &+ \eta(z_{m}, z_{1}) \frac1T \tr \mathbf{M}_{k_i} \cdot \I_N \Big)\times  \A \prod_{i=1}^{m-1} \frac{z_{i+1} m(z_{i+1}) + 1}{\sigma^2 m(z_{i+1})} \P \A^\H \left(\I_K+\frac{z_{i+1} m(z_{i+1}) + 1}{\sigma^2 } \A\P \A^\H \right)^{-1} \Big( m(z_{i}) m(z_{i+1}) \mathbf{M}_{k_{i+1}} \\
    & + \eta(z_{i}, z_{i+1}) \frac1T \tr \mathbf{M}_{k_{i+1}} \cdot \I_N \Big) \A \Big)d z_1 \ldots d z_m +o(1)\\
    &= \left(\frac{1}{2 \pi \imath} \right)^{m}\oint_{\Gamma_{k_1}}\ldots \oint_{\Gamma_{k_m}}  \tr \left(\frac{z_{1} m(z_{1}) + 1}{\sigma^2 m(z_{1})}  \U \mathbf{L}\left(\I_K+\frac{z_1 m(z_1)+1}{\sigma^2 } \mathbf{L} \right)^{-1} \left( m(z_{m}) m(z_{1}) \mathbf{M}_{k_1} + \eta(z_{m}, z_{1}) \frac1T \tr \mathbf{M}_{k_1} \cdot \I_N \right)\right. \\
    & \times \left. \prod_{i=1}^{m-1} \frac{z_{i+1} m(z_{i+1}) + 1}{\sigma^2 m(z_{i+1})}\U \mathbf{L}\left(\I_K+\frac{z_{i+1} m(z_{i+1}) + 1}{\sigma^2 }  \mathbf{L} \right)^{-1} \left( m(z_{i}) m(z_{i+1}) \mathbf{M}_{k_{i+1}} + \eta(z_{i}, z_{i+1}) \frac1T \tr \mathbf{M}_{k_{i+1}} \cdot \I_N \right)  \right)\\
    &d z_1 \ldots d z_m +o(1)\\
    &=\left(\frac{1}{2 \pi \imath} \right)^{m}\oint_{\Gamma_{k_1}}\ldots \oint_{\Gamma_{k_m}}  \tr \left( \frac{z_{1} m(z_{1}) + 1}{\sigma^2 m(z_{1})}\sum_{n=1}^K \frac{\sigma^2 \lambda_n \uu_n \uu_n^\H}{\sigma^2 +(z_{1} m(z_{1}) + 1)\lambda_n} \left( m(z_{m}) m(z_{1}) \mathbf{M}_{k_1} + \eta(z_{m}, z_{i}) \frac1T \tr \mathbf{M}_{k_1} \cdot \I_N \right) \right. \\
    &\times  \left. \prod_{i=1}^{m-1} \frac{z_{i+1} m(z_{i+1}) + 1}{\sigma^2 m(z_{i+1})}\sum_{n=1}^K \frac{\sigma^2 \lambda_n \uu_n \uu_n^\H}{\sigma^2 +(z_{i+1} m(z_{i+1}) + 1)\lambda_n} \left( m(z_{i}) m(z_{i+1}) \mathbf{M}_{k_{i+1}} + \eta(z_{i}, z_{i+1}) \frac1T \tr \mathbf{M}_{k_{i+1}} \cdot \I_N \right) \right) d z_1 \ldots d z_m +o(1)\\
    &=\left(\frac{1}{2 \pi \imath} \right)^{m}\oint_{\Gamma_{k_1}}\ldots \oint_{\Gamma_{k_m}}  \tr \left(\frac{z_1 m(z_1) + 1}{\sigma^2 m(z_1)} \frac{\lambda_1 \uu_1 \uu_1^\H}{\sigma^2 +(z_{1} m(z_{1}) + 1)\lambda_1} \left( m(z_{m}) m(z_{1}) \mathbf{M}_{k_1} + \eta(z_{m}, z_{1}) \frac1T \tr \mathbf{M}_{k_1} \cdot \I_N \right) \right. \\
    & \times  \left. \prod_{i=1}^{m-1} \frac{z_{i+1} m(z_{i+1}) + 1}{\sigma^2 m(z_{i+1})} \frac{\lambda_i \uu_i \uu_i^\H}{\sigma^2 +(z_{i+1} m(z_{i+1}) + 1)\lambda_i} \left( m(z_{i}) m(z_{i+1}) \mathbf{M}_{k_{i+1}} + \eta(z_{i}, z_{i+1}) \frac1T \tr \mathbf{M}_{k_{i+1}} \cdot \I_N \right) \right)d z_1 \ldots d z_m +o(1)\\\\
    &=\left(\frac{1}{2 \pi \imath} \right)^{m}\oint_{\Gamma_{k_1}}\ldots \oint_{\Gamma_{k_m}}  \tr \left( \frac{ m(z_{m})(z_1 m(z_1) + 1) \lambda_1 \uu_1 \uu_1^\H}{\sigma^2 +(z_{1} m(z_{1}) + 1)\lambda_1}  \mathbf{M}_{k_1}\prod_{i=1}^{m-1}  \frac{m(z_{i})(z_{i+1} m(z_{i+1}) + 1) \lambda_i \uu_i \uu_i^\H}{\sigma^2 +(z_{i+1} m(z_{i+1}) + 1)\lambda_i}   \mathbf{M}_{k_{i+1}} \right)d z_1 \ldots d z_m +o(1)\\ \\
    &=\left(\frac{1}{2 \pi \imath} \right)^{m}\oint_{\Gamma_{k_1}}\ldots \oint_{\Gamma_{k_m}} \tr \left( \prod_{i=1}^{m} \frac{m(z_i)(z_i m(z_i) + 1) \lambda_i}{\sigma^2 +(z_{i} m(z_{i}) + 1)\lambda_i}  \uu_i \uu_i^\H \mathbf{M}_{k_i} \right)d z_1 \ldots d z_m +o(1)\\ \\
    &= \prod_{i=1}^{m} \lim_{z \to \bar{\lambda}_i}\frac{(z_i-\bar{\lambda}_i)m(z_i)(z_i m(z_i) + 1) }{\sigma^2 \lambda_i^{-1}+z_{i} m(z_{i}) + 1} \tr (\uu_i \uu_i^\H \mathbf{M}_{k_i} )  + o(1) \\
    &= (g_1 \ldots g_m)  \uu_{k_m}^\H \mathbf{M}_{k_1} \uu_{k_1} \times \uu_{k_1}^\H \mathbf{M}_{k_2}  \uu_{k_2} \times \ldots \times \uu_{k_{m-1}}^\H \mathbf{M}_{k_m}  \uu_{k_m} = \bar{\psi}_{k_1, \ldots, k_m} +o(1),
\end{align*}
by residue calculus.
This conclude the proof of \Cref{theo:CDT}.




\subsection{Proof of \Cref{lem:entrywise_bphi}}
\label{subsec:proof_of_lem:entrywise_bphi}
Before proceeding, we note that $\det(\hat{\bm{\Phi}}_1)$ is eventually bounded away from zero almost surely. 
Indeed, $\bm{\Phi}_1$ is assumed to be invertible in the ESPRIT formulation in (9), and hence $\bar{\bm{\Phi}}_1$ defined in (19) is also nonsingular for $g_k \in (0,1)$ and $\tau > 0$. 
Since $\hat{\bm{\Phi}}_1$ converges to $\bar{\bm{\Phi}}_1$, the continuity of the determinant yields $\det(\hat{\bm{\Phi}}_1) \to \det(\bar{\bm{\Phi}}_1) \neq 0$ almost surely. Therefore, there exists a constant $d>0$ such that $|\det(\hat{\bm{\Phi}}_1)| \ge d$ almost surely for all sufficiently large $N,T$.
Define the set of indices $ \mathcal{I}_m: 1\leq i_1 < \ldots < i_m \leq K$, and the permutations of $ \mathcal{I}_m$ as $\sigma \colon \{ i_1, \ldots, i_m \} \to \{ i_1, \ldots, i_m \}$.
Then, \emph{any} off-diagonal entry $[\hat \bPhi]_{i_1 i_2}$ of $\hat \bPhi$ can be rewritten as the following combination involving the product of the entries of \(\hat \bPhi_1, \hat \bPhi_2 \) as
    \begin{align}\label{eq:hatphi-i1i2}
        &[\hat \bPhi]_{i_1 i_2}  = \sum_{k=1}^{K} [\hat \bPhi_1^{-1}]_{i_1 k} [\hat \bPhi_2]_{k i_2}\\
        &= \sum_{k=1}^{K} \frac{(-1)^{i_1+k}}{\det(\hat \bPhi_1)}\sum_{\tilde{\sigma}}\sign(\tilde{\sigma}) \prod_{p \in \mathcal I_k^{'}, q \in \mathcal{I}^{'}_{i_1}} [\hat \bPhi_1]_{p q} [\hat \bPhi_2]_{k i_2},
    \end{align}
where we use the adjugate matrix to represent the inverse \( \hat \bPhi_1^{-1} \) in the second line, with $\mathcal I_k^{'} \equiv \mathcal I_m \backslash \{k\}, \mathcal I_{i_1}^{'} \equiv \mathcal I_m \backslash \{i_1\}$.

Note that the map $\tilde{\sigma}\colon\mathcal I_k^{'} \to \mathcal I_{i_1}^{'}$ can be obtained from $\sigma$ by deleting the vertex from node $k$ to $i_1$, and can thus be decomposed into the product of several cycles \emph{and} the open path from $i_1$ to $k$.
See \Cref{example:circle_decom} in \Cref{sec:techLemmas} for a concrete example of such decomposition.
Under the same notations of \Cref{theo:CDT}, Equation~\eqref{eq:hatphi-i1i2} can be written as
\begin{align*}
    [\hat \bPhi]_{i_1 i_2} &= \sum_{k=1}^{K} \frac{(-1)^{i_1+k}}{\det(\hat \bPhi_1)}\sum_{\tilde{\sigma}}\sign(\tilde{\sigma}) \prod_{p \in \mathcal I_k^{'}, q \in \mathcal{I}^{'}_{i_1}} [\hat \bPhi_1]_{p q} [\hat \bPhi_2]_{k i_2}\\
    &= \sum_{k=1}^{K} \frac{(-1)^{i_1+k}}{\det(\hat \bPhi_1)}\sum_{\tilde{\sigma}}\sign(\tilde{\sigma}) \hat{\psi}^{\text{cl}} \cdot \hat{\psi}^{\text{op}}_{(i_1,i_2)} ,
\end{align*}
where $\hat{\psi}^{\text{cl}}$ represents the product of off-diagonal entries of $\hat \bPhi_1,\hat \bPhi_2$ whose indices can form several disjoint circles, and $\hat{\psi}^{\text{op}}_{(i_1,i_2)}$ represents the product of entries whose indices start from $i_1$ and end at $i_2$.
Then, we have 
\begin{align*}
    &[\hat \bPhi]_{i_1 i_2}[\hat \bPhi]_{i_2 i_3}\ldots[\hat \bPhi]_{i_m i_1}\\
    & = \sum_{k=1}^{K} \frac{(-1)^{i_1+k}}{\det(\hat \bPhi_1)}\sum_{\tilde{\sigma}_1}\sign(\tilde{\sigma}_1) \ldots \sum_{k=1}^{K} \frac{(-1)^{i_m+k}}{\det(\hat \bPhi_1)}\sum_{\tilde{\sigma}_m}\sign(\tilde{\sigma}_m) \hat{\psi}^{\text{cl}} \cdot \hat{\psi}^{\text{op}}_{(i_1,i_2)}  \hat{\psi}^{\text{op}}_{(i_2,i_3)} \ldots  \hat{\psi}^{\text{op}}_{(i_m,i_1)} \\
    & =  \sum_{k=1}^{K} \frac{(-1)^{i_1+k}}{\det(\hat \bPhi_1)}\sum_{\tilde{\sigma}_1}\sign(\tilde{\sigma}_1) \ldots \sum_{k=1}^{K} \frac{(-1)^{i_m+k}}{\det(\hat \bPhi_1)}\sum_{\tilde{\sigma}_m}\sign(\tilde{\sigma}_m) \prod  \hat{\psi}^{\text{cl}}_i,
\end{align*}
where the indices of the elements in these open paths are connected end to end, forming again circles.
At this point, we conclude that $[\hat \bPhi]_{i_1 i_2}[\hat \bPhi]_{i_2 i_3}\ldots[\hat \bPhi]_{i_m i_1}$
can be expressed as a sum of products of a finite number of pairwise disjoint cycles composed of elements of $\hat \bPhi_1$ and $\hat \bPhi_2$.
This, combined with \Cref{theo:CDT}, concludes the proof of \Cref{lem:entrywise_bphi}.

\subsection{Proof of \Cref{rem:N-con_for_two}}
\label{subsec:proof_of_rem:N-con_for_two}

Note that in the case of $K = 2$ DoAs, the two complex eigenvalues of $\hat{\bPhi}^G$ can be explicitly and compactly given by its trace and determinant as $\lambda_\pm (\hat{\bPhi}^G) =\frac12 \left(\tr(\hat{\bPhi}^G)  \pm \sqrt{\tr^2(\hat{\bPhi}^G)-4 \det(\hat{\bPhi}^G)} \right)$.
It is known (see, for example \cite[Section~2.7]{couillet2022RMT4ML} and \Cref{lem:DE} in \Cref{sec:DE}) that for random matrix $\Z \in \CC^{N \times T}$ having i.i.d.\@ $\mathcal{CN}(0,1)$ entries, its expected resolvent $\EE[\Q(z)] = \EE[(\Z \Z^\H/T - z \I_N)^{-1}]$ can be well approximated by the deterministic equivalent $\bar{\mathbf{Q}}$ in a spectral norm sense $\lVert \EE[\mathbf{Q}] -\bar{\mathbf{Q}} \rVert = O (N^{-1/2})$.
This, together with \Cref{theo:consiseigen}, yields that $\tr (\hat{\bPhi}^G) = \tr (\bPhi) + O(N^{-1/2}), \det (\hat{\bPhi}^G) = \det (\bPhi) + O(N^{-1/2})$, so that
\begin{equation}
    \lambda_\pm(\hat{\bPhi}^G) =  \lambda_\pm(\bPhi)+O(N^{-1/2}).
\end{equation}
Recall that $\Delta \hat{\theta}_k^G = \arctan(\Im[\lambda_k(\hat{\bPhi}^G)]/ \Re[\lambda_(\hat{\bPhi}^G)])$, a Taylor expansion allows us to conclude that $\Delta \theta_k^G = \Delta \theta_k + O(N^{-1/2})$, so that
\begin{equation}
    \hat{\theta}_k^G - \theta_k = O(N^{-3/2}),
\end{equation}
for a distance $\Delta$ of order $N$ and $K=2$.
\end{appendices}

\end{document}